\documentclass[11pt]{article}
\usepackage[T1]{fontenc}
\usepackage[utf8]{inputenc}
\usepackage[a4paper,margin=1in]{geometry}
\usepackage{array}
\usepackage{float}
\usepackage{graphicx}
\usepackage{placeins}
\usepackage{booktabs}
\usepackage{tabularx}
\usepackage{mathtools}
\usepackage{amsthm}
\usepackage{amssymb}
\usepackage[authoryear]{natbib}
\usepackage{algorithm}
\usepackage{algorithmic}
\usepackage{authblk}
\usepackage[inline]{enumitem}
\usepackage{xspace}
\usepackage{setspace}
\usepackage{etoolbox}
\usepackage[unicode=true,bookmarks=true,breaklinks=true,colorlinks=false]{hyperref}
\hypersetup{
  pdftitle={Scalable Pontryagin-Guided Adjoint-to-Control Recovery for Constrained Dynamic Portfolio Choice},
  pdfauthor={Jaegi Jeon, Jeonggyu Huh, Hyeng Keun Koo, and Byung Hwa Lim},
  pdfborder={0 0 1},
  linkbordercolor={0 0 1},
  citebordercolor={0 0.55 0},
  urlbordercolor={0 0 1}
}
\allowdisplaybreaks[2]
\newtheorem{theorem}{Theorem}
\newtheorem{proposition}{Proposition}

\newtheorem{assumption}{Assumption}

\newtheorem{remark}{Remark}
\newcommand{\sg}{\operatorname{sg}}
\newcommand{\DOL}{\mathsf D^{\mathrm{ol}}}
\newcommand{\OLBPTT}{OL-BPTT\xspace}

\newtheorem{apptheorem}{Theorem}[section]
\newtheorem{appproposition}[apptheorem]{Proposition}
\newtheorem{applemma}[apptheorem]{Lemma}
\newtheorem{appassumption}[apptheorem]{Assumption}
\newtheorem{appcorollary}[apptheorem]{Corollary}
\newtheorem{appremark}[apptheorem]{Remark}
\AtBeginEnvironment{table}{\singlespacing}
\AtBeginEnvironment{figure}{\singlespacing}
\AtBeginEnvironment{algorithm}{\singlespacing}

\title{Scalable Pontryagin-Guided Adjoint-to-Control Recovery for Constrained Dynamic Portfolio Choice}
\date{August 29, 2026}
\author[1]{Jaegi Jeon}
\author[2]{Jeonggyu Huh\thanks{Corresponding author: \texttt{jghuh@skku.edu}}}
\author[3]{Hyeng Keun Koo}
\author[4]{Byung Hwa Lim}
\affil[1]{\small Graduate School of Data Science, Chonnam National University, Gwangju, Republic of Korea}
\affil[2]{\small Department of Mathematics, Sungkyunkwan University, Suwon, Republic of Korea}
\affil[3]{\small Department of Financial Engineering, Ajou University, Suwon, Republic of Korea}
\affil[4]{\small Department of Fintech, SKK Business School, Sungkyunkwan University, Seoul, Republic of Korea}

\begin{document}
\setstretch{1.15}
\maketitle

\begin{abstract}
We study continuous-time multi-asset portfolio choice and consumption under
smooth pointwise constraints, including state-dependent feasible sets.  The
method separates dynamic information acquisition from local constrained
recovery.  A pointwise-feasible neural actor generates reference rollouts;
after training, its realized latent outputs are frozen and first- and
second-order adjoints are harvested from a fixed-latent open-loop
backpropagation-through-time graph.  Feedback therefore generates the
reference trajectory without restricting the adjoint formulation to Markov
controls.  Conditional on the harvested adjoints, deployment solves a local
generalized Pontryagin-Hamiltonian problem: quadratic-affine portfolio blocks
are recovered exactly by a quadratic program, while a log barrier approximates
more general regular KKT branches.  We establish local chart representations,
an OL-BPTT-to-adjoint correspondence retaining orthogonal martingale
residuals, and an end-to-end bound from reference value loss and numerical
errors to recovered-policy and local QP-gap errors.  Analytical
constant- and predictable-opportunity benchmarks validate the adjoints.
Common-input experiments show that recovery reduces residual PMP/KKT error
left by finite-budget direct policy optimization, including under a
state-dependent consumption cap and with up to 100 risky assets.  Scalability
concerns the constrained action block rather than dimension-free state-space
complexity.
\end{abstract}


\section{Introduction}\label{sec:intro}

Dynamic portfolio choice has been central to financial economics since
\citet{samuelson1975lifetime} and
\citet{merton1969lifetime,merton1971optimum}; stochastic investment
opportunities were incorporated by \citet{merton1973intertemporal} and later
models such as \citet{kim1996dynamic,liu2007portfolio}.  Classical
frictionless specifications can admit closed-form or semi-analytic solutions
through martingale methods \citep{karatzas1987optimal}, while convex duality
extends this approach to important portfolio constraints
\citep{XuShreve1992a,XuShreve1992b,cvitanic1993hedging,karatzas1998methods}.
Explicit dynamic constrained policies nevertheless remain available mainly
under special structures.  Practical portfolios face short-sale, borrowing,
leverage, sector, and consumption limits.  We study smooth pointwise
restrictions for which the feasible control fiber $\mathcal U(t,s)$ is known
once the current state is fixed, although it may move with the state and its
active set may change.

This structure separates the dynamic and local tasks.  The dynamic problem is
to acquire the adapted first- and second-order adjoints over the deployment
region; conditional on those inputs, the control comparison at fixed $(t,s)$
is finite dimensional.  The separation is useful when the observable state is
low or moderate dimensional but the constrained portfolio vector is large.
Our scalability claim concerns this action block, not dimension-free
state-space complexity
\citep{campbell1999consumption,campbell2001should,campbell2003multivariate,
balduzzi1999transaction,lynch2000predictability,lynch2001portfolio,
brandt2005simulation,buraschi2010correlation,garlappi2010solving,jurek2011optimal}.

Deep backward stochastic differential equation (BSDE), physics-informed
neural network (PINN), policy-learning, and recent value-based jump-control
methods provide complementary approaches to stochastic control
\citep{han2018solving,e2017deep,raissi2019physics,han2016deep,hu2023recent,
cheridito2025deep}.
Hard-constrained PINNs (hPINNs), improved constrained PINNs (ICPINNs), and
activation-based direct policy optimization (DPO) can enforce common
pointwise constraints by construction
\citep{lu2021physics,tan2025improved}.  Pointwise feasibility, however, need not imply
that a finite-budget learned action satisfies the local Pontryagin maximum
principle (PMP) and Karush--Kuhn--Tucker (KKT) conditions.  Related approaches
include the stochastic-maximum-principle method of \citet{ji2022solving}, Deep
Hedging \citep{buehler2019deep}, and the Deep Galerkin Method (DGM)
\citep{sirignano2018dgm}.  In deterministic discrete-time reinforcement
learning, \citet{eberhard2025pontryagin} likewise replace Bellman recursion by
Pontryagin-based optimization of fixed open-loop action sequences.  Our
question is whether adjoints harvested from a pointwise-feasible learned policy
used to generate reference rollouts---the \emph{reference policy}---can remove
the remaining local optimization error of that frozen reference without
refitting a global policy.

Related computational motivation comes from the Pontryagin-guided direct
policy optimization (PG-DPO) framework of \citet{huh2025breaking}, where
ordinary backpropagation-through-time (BPTT) is used for neural-policy training
and BPTT-derived costate information is used to construct Pontryagin-guided
controls.  The fixed-latent/open-loop distinction formalized below is
introduced here.  In the present construction, a pointwise-feasible neural
Markov-form actor---the trainable time--state-to-control map---first generates a
reference rollout.  After training, its realized latent outputs are frozen
before differentiating the feasible coordinate map.  Fixed-latent open-loop
BPTT (\OLBPTT) then removes actor-state derivatives while retaining the
structural state derivatives of a moving feasible fiber.  The resulting
adjoints are associated with the realized adapted open-loop coordinate process;
the primitive admissible class is not restricted to Markov feedback controls.

Because risky weights enter the state diffusion, their first-order
Hamiltonian is affine and does not encode the diffusion-induced quadratic
curvature used in finite control comparison.  The second-order adjoint supplies that curvature.  Given
the estimated adjoint tuple, deployment solves a local generalized-Hamiltonian
problem.  We call the resulting deployment-time map from adjoints to a
recovered action the \emph{decoder}.  QP-PGDPO and B-PGDPO denote conditional
Stage~2 recovery maps, not additional actor-training procedures.  In the
quadratic-affine blocks studied numerically, the quadratic-program realization
is exact conditional on the supplied inputs; the barrier realization follows a
log-barrier central path toward the same local KKT solution on a regular smooth
branch.  The barrier is therefore an implementation of the recovery stage, not
its foundation.

The paper makes three contributions.
\begin{enumerate}[label=(\roman*)]
\item \emph{Reference-control adjoint acquisition.}  A local corner chart
reparameterizes progressively measurable controls on a regular feasible
patch.  Fixed-latent \OLBPTT converges to the chart-reduced first- and
second-order adjoints associated with the reference control while retaining
the orthogonal martingale residuals in the one-step Brownian projections.

\item \emph{Adjoint-to-control recovery.}  The adapted adjoint tuple defines a
solver-neutral local maximization.  We give an exact QP/KKT realization for
quadratic-affine portfolio blocks and a barrier approximation for general
smooth regular branches.  Under local quadratic growth and a common-state
stability bridge, reference loss $\varepsilon_{\rm ref}$ and aggregate
numerical error $\eta$ imply QP policy error
$O(\sqrt{\varepsilon_{\rm ref}}+\eta)$ and local QP gap
$O((\sqrt{\varepsilon_{\rm ref}}+\eta)^2)$.

\item \emph{Validation at scale.}  Analytical-rollout constant- and
predictable-opportunity benchmarks first isolate end-to-end adjoint-acquisition
error.  Common-input experiments---which hold fixed the reference rollout,
evaluation states, harvested adjoints, and Brownian continuations across the
recovery solvers---then test whether local recovery improves a frozen DPO
reference, whether the improvement persists across finite training budgets, and
whether the pipeline remains effective with predictable returns and a
state-dependent consumption fiber.  Additional hPINN-style and barrier audits
delimit the comparison and solver scope, with action dimensions up to 100 risky
assets.
\end{enumerate}

All theoretical conclusions are local to regular feasible patches, and QP
exactness is conditional on the supplied adjoints.  Smooth Markov identities
are used only for interpretation and analytical benchmarks.  Mathematical
symbols are defined at first use; Published Appendix~\ref{app:notation_summary}
summarizes the recurring notation and coordinate conventions
(Tables~\ref{tab:core_symbols} and~\ref{tab:discrete_symbols}).  The remainder
of the paper formulates the model in Section~\ref{sec:dp_multiasset}, develops
the chart-reduced adjoint systems and recovery map in
Section~\ref{sec:PMP_multiasset}, gives the numerical pipeline in
Section~\ref{sec:BPTT_multiasset}, and reports the experiments in
Section~\ref{sec:num_test}.  Core mathematical proofs, including the full fixed-latent OL-BPTT
consistency argument, are collected in the separately submitted Published
Appendix.  Classical open-loop/feedback specializations, quadratic-growth
scope results, implementation details, and supplementary diagnostics are
collected in the Online Supporting Information.

\section{Multi-Asset Portfolio Problems under Constraints}
\label{sec:dp_multiasset}

\subsection{Model and Feasible Control Fibers}
\label{sec:model_setup}

The model has the finite-dimensional state representation
\[
 \mathbf S_t=(X_t,Y_t)\in(0,\infty)\times\mathbb R^{d_Y},
 \qquad d_S:=1+d_Y,
\]
where $X_t$ is wealth and $Y_t$ is a vector of observable
investment-opportunity factors.  Work on a filtered probability space
$(\Omega,\mathcal F,\mathbb F,\mathbb P)$ carrying a
$d_W$-dimensional Brownian motion $\mathbf W$, where
$\mathbb F=(\mathcal F_t)_{0\le t\le T}$ is its augmented natural
filtration.  The coefficients below have a finite-dimensional Markovian
representation in $(t,\mathbf S_t)$.

There is one locally risk-free asset with rate $r(t,y)$ and $n$ risky assets
with excess-return vector
$\boldsymbol\alpha(t,y)\in\mathbb R^n$ and return loading
$\boldsymbol v(t,y)\in\mathbb R^{n\times d_W}$.  Define
\[
 \boldsymbol\Sigma(t,y)
 :=
 \boldsymbol v(t,y)\boldsymbol v(t,y)^\top
 \in\mathbb R^{n\times n},
\]
and assume that $\boldsymbol\Sigma(t,y)$ is uniformly positive definite on the
working region.  Throughout, $\mathbf1$ and $\boldsymbol0$ denote the
all-ones and zero vectors of the dimension implied by context.  Let
\[
 \boldsymbol\pi_t
 =
 (\pi_{1,t},\ldots,\pi_{n,t})^\top
 \in\mathbb R^n,
 \qquad
 \pi_{0,t}:=1-\mathbf1^\top\boldsymbol\pi_t
\]
denote the risky-asset weights and the residual risk-free weight,
respectively.  Let $C_t$ denote the consumption rate, and
define
\(
 \mathbf u_t:=(\boldsymbol\pi_t^\top,C_t)^\top,
 \qquad d_u:=n+1.
\)

The exogenous factor dynamics have drift
$\boldsymbol\beta_Y(t,y)\in\mathbb R^{d_Y}$ and diffusion
$\boldsymbol\sigma_Y(t,y)\in\mathbb R^{d_Y\times d_W}$.  The joint
wealth--factor dynamics are
\begin{equation}\label{eq:wealth_factor_dynamics}
\begin{aligned}
 dX_t
 &=
 \Bigl[
 X_t\{r(t,Y_t)
 +\boldsymbol\pi_t^\top\boldsymbol\alpha(t,Y_t)\}
 -C_t
 \Bigr]dt
 +X_t\boldsymbol\pi_t^\top
 \boldsymbol v(t,Y_t)\,d\mathbf W_t,
 \qquad X_0=x_0>0,
 \\
 dY_t
 &=
 \boldsymbol\beta_Y(t,Y_t)\,dt
 +\boldsymbol\sigma_Y(t,Y_t)\,d\mathbf W_t,
 \qquad Y_0=y_0.
\end{aligned}
\end{equation}
Because both equations are driven by $\mathbf W$, the model allows
instantaneous return--factor correlation.

The joint dynamics can be written compactly as
\[
 d\mathbf S_t
 =
 \boldsymbol b_S(t,\mathbf S_t;\mathbf u_t)\,dt
 +
 \boldsymbol\sigma_S(t,\mathbf S_t;\mathbf u_t)\,d\mathbf W_t,
\]
where $\boldsymbol b_S$ collects the drift terms in
\eqref{eq:wealth_factor_dynamics}, and
\[
 \boldsymbol\sigma_S(t,\mathbf s;\mathbf u)
 =
 \begin{pmatrix}
 x\boldsymbol\pi^\top\boldsymbol v(t,y)
 \\
 \boldsymbol\sigma_Y(t,y)
 \end{pmatrix}
 \in\mathbb R^{d_S\times d_W},
 \qquad
 \mathbf s=(x,y),\quad
 \mathbf u=(\boldsymbol\pi^\top,C)^\top.
\]
When $d_Y=0$, the factor equation and the lower block of
$\boldsymbol\sigma_S$ are omitted, so that $\mathbf S_t=X_t$ and $d_S=1$.

The objective is
\begin{equation}\label{eq:portfolio_objective_model}
J(\mathbf u)
=
\mathbb E\!\left[
\int_0^T e^{-\rho t}U(C_t)\,dt
+
K e^{-\rho T}U(X_T)
\right],
\end{equation}
where $\rho\ge0$, $K\ge0$, and
$U:(0,\infty)\to\mathbb R$ is $C^2$, strictly increasing, and strictly
concave.  If zero consumption belongs to the closure of the feasible set,
$U(0)$ is understood in the extended-value sense as
\(
U(0):=\lim_{c\downarrow0}U(c)\in[-\infty,\infty).
\)

All pointwise portfolio and consumption inequalities are collected in a
smooth constraint map
\[
\Gamma(t,\mathbf s;\mathbf u)
=
\bigl(
\Gamma_1(t,\mathbf s;\mathbf u),
\ldots,
\Gamma_{m_\Gamma}(t,\mathbf s;\mathbf u)
\bigr)^\top
\ge\boldsymbol0,
\qquad
\mathbf s=(x,y).
\]
For example, the borrowing-allowed no-short-sale benchmark imposes
\(
\boldsymbol\pi\ge\boldsymbol0
\)
with no upper bound on $\mathbf1^\top\boldsymbol\pi$, whereas the no-borrowing
models additionally impose
\(
1-\mathbf1^\top\boldsymbol\pi\ge0.
\)
Consumption bounds are represented by
\[
C-C_{\min}(t,\mathbf s)\ge0,
\qquad
C_{\max}(t,\mathbf s)-C\ge0,
\]
where $C_{\min}=0$ gives consumption nonnegativity and
$C_{\max}(t,\mathbf s)=\bar m x$ gives a proportional upper bound.

The feasible control fiber is
\begin{equation}\label{eq:feasible_control_fiber}
\mathcal U(t,\mathbf s)
=
\left\{
\mathbf u\in\mathbb R^{n+1}:
\Gamma(t,\mathbf s;\mathbf u)\ge\boldsymbol0
\right\}.
\end{equation}
The restrictions are primitive and do not depend on the value function.
When the feasible fiber $\mathcal U(t,\mathbf s)$ actually varies with the
state, as under $C\le\bar m x$, the local recovery problem holds
$\mathbf s$ fixed, whereas the fixed-latent state differentiation introduced
below retains the structural state dependence of the feasible chart.  A
control is admissible if it is progressively measurable, belongs to
$\mathcal U(t,\mathbf S_t)$ for $dt\otimes d\mathbb P$-almost every
$(t,\omega)$, the state equation admits a unique solution with positive
wealth, and the required integrability and moment conditions hold.  We call
any such progressively measurable process an \emph{open-loop control}; it may
depend on the full information history in $\mathcal F_t$.  A
\emph{Markov feedback rule} is a Borel measurable selector satisfying
\[
\varphi(t,\mathbf s)\in\mathcal U(t,\mathbf s).
\]
We call the induced process
$\mathbf u_t=\varphi(t,\mathbf S_t)$ an admissible Markov feedback control
when its closed-loop state equation satisfies the preceding admissibility
conditions.  The Markovian coefficient representation does not restrict the
admissible class to such feedback controls.

\begin{assumption}[Smooth Pointwise Constraint Geometry]
\label{ass:Gamma}
On the working domain:
\begin{enumerate}[label=(\roman*),leftmargin=*]
\item \emph{Smoothness.}
      Each $\Gamma_j$ is $C^1$ in time and jointly $C^2$ in
      $(\mathbf s,\mathbf u)$, with locally Lipschitz derivatives.
\item \emph{Nonemptiness.}
      The feasible fiber $\mathcal U(t,\mathbf s)$ is nonempty for every
      $(t,\mathbf s)$ in the working domain.
\end{enumerate}
\end{assumption}

Assumption~\ref{ass:Gamma} covers state-independent affine restrictions and
smooth mixed state--control constraints such as $C\le\bar mX$.
Throughout, $\Gamma$ contains pointwise control constraints and mixed
state--control constraints, but not pure state constraints independent of
$\mathbf u$.  The feasible fibers may have corners, as in nonnegativity and
polyhedral portfolio constraints; the local analysis is conducted on regular
active-set strata specified in Section~\ref{sec:PMP_multiasset}.  Nonsmooth
constraint functions, nonregular active-set intersections, reflected or
measure-valued multipliers, and singular controls are outside the present
framework.

\subsection{Open-Loop Value, Dynamic-Programming Benchmark, and Feedback Realization}
\label{sec:hjb_motivation}

This subsection records the dynamic-programming benchmark against which the
adjoint-based recovery below will be compared.  For the finite-dimensional
Markov state model above, while retaining the full open-loop admissible class,
define
\begin{equation}\label{eq:value_function}
V(t,x,y)
=
\sup_{\mathbf u}
\mathbb E\!\left[
\int_t^T e^{-\rho s}U(C_s)\,ds
+
K e^{-\rho T}U(X_T)
\,\Big|\,\mathbf S_t=(x,y)
\right],
\end{equation}
where the supremum is over admissible open-loop controls, not only Markov
feedback controls.  When the usual dynamic-programming and comparison
hypotheses hold, $V$ is characterized by the constrained
Hamilton--Jacobi--Bellman (HJB) equation in the viscosity sense; when $V$ is
smooth,
\begin{equation}\label{eq:HJB_constrained}
0= V_t +
\sup_{\mathbf u\in\mathcal U(t,\mathbf s)}
\Bigg\{
e^{-\rho t}U(C)
+
D_{\mathbf s}V^\top\boldsymbol b_S(t,\mathbf s;\mathbf u)
+
\frac12\operatorname{tr}\!\left[
\boldsymbol\sigma_S\boldsymbol\sigma_S^\top
D^2_{\mathbf s\mathbf s}V
\right]
\Bigg\},
\end{equation}
with all coefficients evaluated at $(t,\mathbf s;\mathbf u)$ and terminal
condition $V(T,x,y)=K e^{-\rho T}U(x)$.  The viscosity formulation is used
only for context.  Our verification statement is classical: if a
$C^{1,2}$ solution admits a Borel maximizer over the full feasible fiber and
the induced closed-loop equation is admissible, then the resulting feedback
attains \eqref{eq:value_function} among all admissible open-loop controls
\citep{fleming2006controlled,yong2012stochastic}.  This provides one optimal
feedback representative, not Markovianity of every optimizer.  A local-chart
selector alone does not verify the full value; a branchwise result requires a
separate branch-restricted value problem.

For the portfolio model, the consumption block is
$e^{-\rho t}U(C)-V_XC$.  The risky-weight block is
\begin{equation}\label{eq:hjb_factor_portfolio_block}
x
\left[
V_X\boldsymbol\alpha(t,y)
+
\boldsymbol\Sigma_{RY}(t,y)\nabla_yV_X
\right]^\top
\boldsymbol\pi
+
\frac12x^2V_{XX}
\boldsymbol\pi^\top\boldsymbol\Sigma(t,y)\boldsymbol\pi,
\end{equation}
where
$\boldsymbol\Sigma_{RY}:=\boldsymbol v\boldsymbol\sigma_Y^\top$
is the return--factor cross-covariance.  The mixed term is the intertemporal
hedging component and vanishes when $d_Y=0$.  With constant coefficients, an
unconstrained interior maximizer satisfies
\[
\boldsymbol\pi_t^*
=-\frac{V_X(t,x)}{xV_{XX}(t,x)}
\boldsymbol\Sigma^{-1}\boldsymbol\alpha,
\qquad V_{XX}(t,x)<0.
\]
For constant relative risk aversion (CRRA) utility,
$U(z)=z^{1-\gamma}/(1-\gamma)$, the prefactor is $1/\gamma$.

At a smooth constrained maximizer, the HJB block satisfies the usual KKT
conditions.  Thus, if the smooth value derivatives were available, they would
reduce the dynamic portfolio problem to a static constrained control block.
Our method replaces those unavailable derivatives by adapted adjoint
information harvested from a reference trajectory; HJB identities are used
only for comparison when a smooth Markov value function exists.

\section{Adjoint Systems and Local Control Recovery}
\label{sec:PMP_multiasset}

We now formulate the local comparison within the open-loop control problem.
The first-order stochastic Pontryagin maximum principle (PMP) yields an
infinitesimal condition under convex admissible perturbations.  For the finite
risky-weight recovery comparison used here, the second-order adjoint and
generalized Hamiltonian retain the diffusion-induced curvature generated by
finite control changes
\citep{bensoussan1982lectures,peng1990general,yong2012stochastic}.  We first
construct local feasible coordinates, then define the chart-reduced adjoint
systems and the solver-neutral recovery map; smooth Markov identities are
discussed only as an optional interpretation.

\subsection{Open-Loop Controls and Local Feasible Coordinates}
\label{sec:chart_open_loop_adjoints}

An admissible control is a progressively measurable process satisfying
$\mathbf u_t\in\mathcal U(t,\mathbf S_t)$ almost everywhere, together with
the state-equation, positivity, integrability, and moment conditions in
Section~\ref{sec:model_setup}.  Here \emph{open loop} means progressively
measurable with respect to $\mathbb F$; no Markov feedback form is imposed,
and the control may depend on all information in $\mathcal F_t$.

\subsubsection{Geometry of State-Dependent Feasible Controls}

Write $b=(t,\mathbf s)$ and
\[
\operatorname{Gr}(\mathcal U)
:=\{(b,\mathbf u):\mathbf u\in\mathcal U(b)\}.
\]
At a feasible point define the active set
\[
\mathcal I(b,\mathbf u):=\{j:\Gamma_j(b;\mathbf u)=0\}.
\]
The point is \emph{vertically regular} when the active control gradients are
linearly independent,
\[
\operatorname{rank}D_{\mathbf u}\Gamma_{\mathcal I(b,\mathbf u)}(b;\mathbf u)
=|\mathcal I(b,\mathbf u)|.
\]
For a fixed active set $I$, the corresponding smooth piece
\[
\mathcal U_I(b)
:=\{\mathbf u:\Gamma_i(b;\mathbf u)=0\ (i\in I),\ 
\Gamma_j(b;\mathbf u)>0\ (j\notin I)\}
\]
is locally a manifold at vertically regular points.  To include both an
active face and inward feasible directions, we use a corner chart rather than
a chart of the face alone.

Fix a vertically regular point $(b_0,\mathbf u_0)$, let
$I_0=\mathcal I(b_0,\mathbf u_0)$ and $r=|I_0|$.  After shrinking
neighborhoods, the parameterized inverse-function theorem gives an open
neighborhood $O_a\subset\mathbb R^{d_u}$ and the relatively open corner
domain
\[
\mathfrak A_{\rm cor}
=O_a\cap\bigl([0,\infty)^r\times\mathbb R^{d_u-r}\bigr),
\]
a time--state neighborhood $B_0$, and a fiber-preserving coordinate map
\[
\widehat\Psi(b,\mathbf a)
:=\bigl(b,\Psi(b,\mathbf a)\bigr)
:
B_0\times\mathfrak A_{\rm cor}
\longrightarrow
\operatorname{Gr}(\mathcal U)
\]
with a local inverse.  The active slacks may be taken as the first $r$
coordinates, so $a_i=0$ represents the active face and $a_i>0$ an inward
direction.  The map $\Psi$ is only a coordinate device.  For example,
$C=Xc$, $c\in[0,\bar m]$, parametrizes the full fiber under
$0\le C\le\bar mX$.

We localize one such chart to a compact patch
$\mathcal D_{\rm ch}\subset B_0$ along the reference or optimal trajectory,
write $\mathfrak A:=\mathfrak A_{\rm cor}$, and suppress the chart index.

\begin{assumption}[Local Corner-Chart Regularity and Adjoint Well-Posedness]
\label{ass:chart_pmp_regularity}
On each compact chart patch used for adjoint acquisition:
\begin{enumerate}[label=(\roman*),leftmargin=*]

\item \emph{Corner-chart regularity.}
The map $\widehat\Psi(t,\mathbf s,\mathbf a)
=(t,\mathbf s,\Psi(t,\mathbf s,\mathbf a))$ is one-to-one onto the selected
local feasible graph and has an inverse
$(t,\mathbf s,\mathbf u)\mapsto
(t,\mathbf s,\Phi(t,\mathbf s,\mathbf u))$ there.  The maps $\Psi$ and $\Phi$
are jointly Borel measurable.  The inverse coordinate map $\Phi$ is locally
Lipschitz on the selected graph, with a bounded Lipschitz constant on the
compact localization.  The control component $\Psi$ is $C^1$ in time, jointly
$C^2$ in $(\mathbf s,\mathbf a)$, and admits an extension with this regularity
to an open neighborhood of the corner domain.  Its derivatives through order
two are bounded on the compact localization.

\item \emph{Coefficient regularity.}
The primitive drift, diffusion, and running reward are $C^1$ in time and
jointly $C^2$ in state and control, while the terminal reward is $C^2$ in the
state.  The chart compositions of the drift, diffusion, and running reward
have locally Lipschitz derivatives through order two and satisfy the required
growth and moment bounds.  On the localization used for the adjoint analysis,
the wealth and consumption components remain in compact subsets of
$(0,\infty)$.

\item \emph{Well-posedness.}
The chart-reduced state equation and the corresponding first- and second-order
adjoint systems admit unique adapted solutions.

\end{enumerate}
\end{assumption}

\subsubsection{Open-Loop Processes in Local Feasible Coordinates}

The coordinate process $\mathbf a_t$ may be any progressively measurable
$\mathfrak A$-valued process satisfying the localization and integrability
conditions; it is not assumed to be a feedback map.  All state derivatives of
chart-composed coefficients are taken with $\mathbf a$ fixed.  Define
\[
\mathcal U_{\rm ch}(t,\mathbf s)
:=\Psi(t,\mathbf s,\mathfrak A)
\subset\mathcal U(t,\mathbf s).
\]

\begin{proposition}[Local chart representation of admissible open-loop controls]
\label{prop:local_open_loop_chart}
Suppose Assumption~\ref{ass:chart_pmp_regularity} holds on
$\mathcal D_{\rm ch}$, and let $\Phi$ be its local inverse on the chart image.

Let $\mathbf a$ be any progressively measurable $\mathfrak A$-valued process
for which the chart-reduced state equation is well posed, the trajectory
remains in $\mathcal D_{\rm ch}$, and the usual admissibility and integrability
conditions hold.  Then
\[
\mathbf u_t=\Psi(t,\mathbf S_t,\mathbf a_t)
\]
is a progressively measurable admissible control taking values in
$\mathcal U_{\rm ch}(t,\mathbf S_t)$.  Conversely, if $\mathbf u$ is a
primitive admissible open-loop control whose state remains in
$\mathcal D_{\rm ch}$ and such that
$\mathbf u_t\in\mathcal U_{\rm ch}(t,\mathbf S_t)$ almost everywhere, then
\[
\mathbf a_t=\Phi(t,\mathbf S_t,\mathbf u_t)
\]
is progressively measurable and satisfies
$\mathbf u_t=\Psi(t,\mathbf S_t,\mathbf a_t)$ almost everywhere.

Hence the chart is a reparameterization of the localized admissible open-loop
controls in the selected regular corner patch, including its boundary faces
and adjacent inward feasible directions; it does not impose a Markov-feedback
restriction.
\end{proposition}

\begin{proof}
Progressive measurability follows from joint measurability of $\Psi$ and
$\Phi$; feasibility and the representation identities follow from the chart
and its inverse.  The remaining admissibility conditions are assumed in the
statement.
\end{proof}

The result is local: it does not cover controls leaving the patch or crossing
a nonregular active-set intersection, nor does it assert a global chart or a
Markov optimizer.

\subsection{First-Order Adjoint System and Its Variational Scope}
\label{sec:first_order_scope}

Let
\[
\ell(t,\mathbf s;\mathbf u):=e^{-\rho t}U(C)
\]
denote the running reward.  The primitive first-order Hamiltonian is
\begin{equation}\label{eq:first_order_hamiltonian}
\mathcal H(t,\mathbf s,\mathbf u,\lambda,\mathbf Z)
=
\ell(t,\mathbf s;\mathbf u)
+
\lambda^\top\boldsymbol b_S(t,\mathbf s;\mathbf u)
+
\left\langle
\mathbf Z,
\boldsymbol\sigma_S(t,\mathbf s;\mathbf u)
\right\rangle_F.
\end{equation}
Here $\langle\cdot,\cdot\rangle_F$ is the Frobenius inner product.

The fixed-coordinate (fixed-latent) chart compositions are
\[
\begin{aligned}
\bar\ell(t,\mathbf s,\mathbf a)
&:=
\ell\bigl(t,\mathbf s,\Psi(t,\mathbf s,\mathbf a)\bigr),\\
\bar{\boldsymbol b}(t,\mathbf s,\mathbf a)
&:=
\boldsymbol b_S
\bigl(t,\mathbf s;\Psi(t,\mathbf s,\mathbf a)\bigr),\\
\bar{\boldsymbol\sigma}(t,\mathbf s,\mathbf a)
&:=
\boldsymbol\sigma_S
\bigl(t,\mathbf s;\Psi(t,\mathbf s,\mathbf a)\bigr).
\end{aligned}
\]
The corresponding chart-reduced Hamiltonian is
\begin{equation}\label{eq:chart_reduced_hamiltonian}
\bar{\mathcal H}(t,\mathbf s,\mathbf a,\lambda,\mathbf Z)
=
\bar\ell
+
\lambda^\top\bar{\boldsymbol b}
+
\left\langle\mathbf Z,\bar{\boldsymbol\sigma}\right\rangle_F,
\end{equation}
where the barred quantities are evaluated at $(t,\mathbf s,\mathbf a)$.

For a progressively measurable coordinate process, set
$\mathbf u_t=\Psi(t,\mathbf S_t,\mathbf a_t)$.  State perturbations hold
$\mathbf a_t$ fixed, so the structural state dependence of $\Psi$ remains.
The chart-reduced state equation is
\[
d\mathbf S_t
=
\bar{\boldsymbol b}(t,\mathbf S_t,\mathbf a_t)\,dt
+
\bar{\boldsymbol\sigma}(t,\mathbf S_t,\mathbf a_t)\,d\mathbf W_t.
\]

The associated first-order adjoint is
$(\lambda_t,\mathbf Z_t)$, where
\[
\lambda_t\in\mathbb R^{d_S},
\qquad
\mathbf Z_t\in\mathbb R^{d_S\times d_W},
\]
and
\begin{equation}\label{eq:chart_reduced_first_adjoint}
d\lambda_t
=
-
D_{\mathbf s}\bar{\mathcal H}
(t,\mathbf S_t,\mathbf a_t,\lambda_t,\mathbf Z_t)\,dt
+
\mathbf Z_t\,d\mathbf W_t,
\end{equation}
with terminal condition
\[
\lambda_T
=
K e^{-\rho T}U'(X_T)e_X,
\qquad
e_X=(1,0,\ldots,0)^\top.
\]

At fixed geometric coordinate, the first state derivative is
\begin{equation}\label{eq:fixed_latent_first_hamiltonian_chain_rule}
D_{\mathbf s}\bar{\mathcal H}[h]
=
\mathcal H_{\mathbf s}[h]
+
\mathcal H_{\mathbf u}^\top\Psi_{\mathbf s}h,
\end{equation}
with the right-hand side evaluated at
$\mathbf u=\Psi(t,\mathbf s,\mathbf a)$ and $(\lambda,\mathbf Z)$ fixed.

\begin{remark}[Scope of the first-order condition]
\label{rem:first_order_scope}
If $\mathbf a^*$ is locally optimal for the chart-reduced open-loop problem
on the selected regular patch and the convex-amplitude segment toward a bounded
adapted comparator $\mathbf v$ remains admissible in the chart, then
\[
\left\langle
D_{\mathbf a}\bar{\mathcal H}
(t,\mathbf S_t^*,\mathbf a_t^*,\lambda_t^*,\mathbf Z_t^*),
\mathbf v_t-\mathbf a_t^*
\right\rangle\le0
\]
for almost every $(t,\omega)$, for each fixed comparator, with the
exceptional null set allowed to depend on that comparator.  This infinitesimal
condition remains valid with controlled diffusion, but is
neither a finite-control maximum condition nor a sufficient condition without
additional concavity or verification hypotheses
\citep{bensoussan1982lectures,cadenillas1995stochastic,yong2012stochastic}.
Online Supporting Information~\ref{app:first_order_variation} gives the precise statement.
\end{remark}

Consumption does not enter the diffusion, so under a convex feasible interval
and concave $U$, its scalar block
$e^{-\rho t}U(C)-\lambda^X C$ is characterized by first-order KKT
conditions.  Risky weights do enter the wealth diffusion; their first-order
Hamiltonian is affine and therefore does not encode the diffusion-induced
quadratic curvature required by the finite risky-weight comparison below.
Online Supporting Information~\ref{app:second_order_dropout} records the corresponding reduction
when the compared controls do not change the diffusion.

\subsection{Second-Order Adjoints and the Generalized Hamiltonian}
\label{sec:generalized_hamiltonian_portfolio}

The second-order adjoint supplies the diffusion curvature used in our finite
risky-weight recovery comparison
\citep{peng1990general,yong2012stochastic}.

Let $P_t\in\mathbb S_{\rm sym}^{d_S}$ denote the second-order adjoint and let
$\mathbf R_t=(R_t^1,\ldots,R_t^{d_W})$, with
$R_t^\ell\in\mathbb S_{\rm sym}^{d_S}$, denote its martingale coefficients.
For notational brevity, write
\[
\bar{\boldsymbol b}_t
:=
\bar{\boldsymbol b}(t,\mathbf S_t,\mathbf a_t),
\qquad
\bar{\boldsymbol\sigma}_t
:=
\bar{\boldsymbol\sigma}(t,\mathbf S_t,\mathbf a_t),
\]
and
\[
\bar{\mathcal H}_t
:=
\bar{\mathcal H}
(t,\mathbf S_t,\mathbf a_t,\lambda_t,\mathbf Z_t).
\]
Let $\bar{\boldsymbol\sigma}_t^{\,\ell}$ denote the $\ell$th column of
$\bar{\boldsymbol\sigma}_t$.  With the maximization convention, the
second-order adjoint satisfies
\begin{equation}\label{eq:chart_reduced_second_adjoint}
\begin{aligned}
-dP_t
={}&
\Bigg[
D^2_{\mathbf s\mathbf s}\bar{\mathcal H}_t
+
(D_{\mathbf s}\bar{\boldsymbol b}_t)^\top P_t
+
P_tD_{\mathbf s}\bar{\boldsymbol b}_t
+
\sum_{\ell=1}^{d_W}
(D_{\mathbf s}\bar{\boldsymbol\sigma}_t^{\,\ell})^\top
P_t
D_{\mathbf s}\bar{\boldsymbol\sigma}_t^{\,\ell}
\\
&\quad
+
\sum_{\ell=1}^{d_W}
\left[
(D_{\mathbf s}\bar{\boldsymbol\sigma}_t^{\,\ell})^\top R_t^\ell
+
R_t^\ell
D_{\mathbf s}\bar{\boldsymbol\sigma}_t^{\,\ell}
\right]
\Bigg]dt
-
\sum_{\ell=1}^{d_W}R_t^\ell\,dW_t^\ell,
\\
P_T
={}&
K e^{-\rho T}U''(X_T)e_Xe_X^\top.
\end{aligned}
\end{equation}

The second fixed-coordinate derivative retains the state dependence of the
feasible chart.  For directions $h,k\in\mathbb R^{d_S}$,
\begin{equation}\label{eq:fixed_latent_hamiltonian_chain_rule}
\begin{aligned}
D^2_{\mathbf s\mathbf s}\bar{\mathcal H}[h,k]
={}&
\mathcal H_{\mathbf s\mathbf s}[h,k]
+
\mathcal H_{\mathbf s\mathbf u}
[h,\Psi_{\mathbf s}k]
+
\mathcal H_{\mathbf u\mathbf s}
[\Psi_{\mathbf s}h,k]
\\
&+
\mathcal H_{\mathbf u\mathbf u}
[\Psi_{\mathbf s}h,\Psi_{\mathbf s}k]
+
\mathcal H_{\mathbf u}^\top
\Psi_{\mathbf s\mathbf s}[h,k].
\end{aligned}
\end{equation}
All right-hand-side derivatives are evaluated at
$\mathbf u=\Psi(t,\mathbf s,\mathbf a)$ with
$(\lambda,\mathbf Z)$ fixed.  Thus fixed-coordinate differentiation removes
actor derivatives but retains chart-state derivatives.  These adjoints are
defined for any admissible charted control and coincide with the corresponding
PMP adjoints when the control satisfies the stated local optimality and
maximum-principle hypotheses.

For the adjoints associated with a fixed charted control
$\mathbf u_t=\Psi(t,\mathbf S_t,\mathbf a_t)$, define the shifted martingale
adjoint by
\begin{equation}\label{eq:shifted_zeta}
\boldsymbol\zeta_t
:=
\mathbf Z_t-P_t\bar{\boldsymbol\sigma}_t
=
\mathbf Z_t
-
P_t\boldsymbol\sigma_S(t,\mathbf S_t;\mathbf u_t)
\in\mathbb R^{d_S\times d_W}.
\end{equation}
For a generic state--control--adjoint tuple, define the generalized
Hamiltonian by
\begin{equation}\label{eq:generalized_hamiltonian}
\begin{aligned}
\mathbb H
(t,\mathbf s;\mathbf u,\lambda,\boldsymbol\zeta,P)
:={}&
\ell(t,\mathbf s;\mathbf u)
+
\lambda^\top
\boldsymbol b_S(t,\mathbf s;\mathbf u)
\\
&+
\left\langle
\boldsymbol\zeta,
\boldsymbol\sigma_S(t,\mathbf s;\mathbf u)
\right\rangle_F
\\
&+
\frac12
\operatorname{tr}\!\left[
\boldsymbol\sigma_S(t,\mathbf s;\mathbf u)^\top
P
\boldsymbol\sigma_S(t,\mathbf s;\mathbf u)
\right].
\end{aligned}
\end{equation}
At an anchor control $\mathbf u^\circ$, write
$\boldsymbol\sigma(\mathbf u):=\boldsymbol\sigma_S(t,\mathbf s;\mathbf u)$
and $\boldsymbol\sigma^\circ:=\boldsymbol\sigma(\mathbf u^\circ)$.  With
$\boldsymbol\zeta=\mathbf Z-P\boldsymbol\sigma^\circ$,
\eqref{eq:generalized_hamiltonian} is equivalent, up to a control-independent
constant, to the standard second-order maximum-principle comparison obtained
by adding
$\tfrac12\langle P(\boldsymbol\sigma(\mathbf u)-\boldsymbol\sigma^\circ),
\boldsymbol\sigma(\mathbf u)-\boldsymbol\sigma^\circ\rangle_F$
to the first-order Hamiltonian.
The adjoints are chart-reduced, but recovery is expressed in the original
control coordinates at fixed $(t,\mathbf s)$.  It is local to the selected
branch unless concavity of the stagewise block and convexity of the feasible
set globalize the comparison.

The definition \eqref{eq:shifted_zeta} yields the anchoring identity
\begin{equation}\label{eq:hamiltonian_gradient_identity}
\nabla_{\mathbf u}\mathbb H
(t,\mathbf S_t;\mathbf u_t,
 \lambda_t,\boldsymbol\zeta_t,P_t)
=
\nabla_{\mathbf u}\mathcal H
(t,\mathbf S_t,\mathbf u_t,
 \lambda_t,\mathbf Z_t).
\end{equation}
At $\mathbf u=\mathbf u_t$, the shifted linear and quadratic diffusion terms
have cancelling first derivatives, while the quadratic term supplies the
risky-weight curvature.

For a locally optimal charted control, the second-order stochastic maximum
principle gives
\begin{equation}\label{eq:open_loop_max_condition}
\mathbb H
(t,\mathbf S_t^*;\mathbf u_t^*,
 \lambda_t^*,\boldsymbol\zeta_t^*,P_t^*)
\ge
\mathbb H
(t,\mathbf S_t^*;\mathbf u,
 \lambda_t^*,\boldsymbol\zeta_t^*,P_t^*)
\end{equation}
for every
$\mathbf u\in\mathcal U_{\rm ch}(t,\mathbf S_t^*)$
sufficiently close to $\mathbf u_t^*$, for
$dt\otimes d\mathbb P$-almost every $(t,\omega)$.

Partition the adjoints according to $\mathbf S=(X,Y)$ and write
$P^{YX}=(P^{XY})^\top$.  The part of
\eqref{eq:generalized_hamiltonian} that depends on the risky-asset weights is
\begin{equation}\label{eq:open_loop_portfolio_block}
\begin{aligned}
&
x
\left[
\lambda^X\boldsymbol\alpha(t,y)
+
\boldsymbol v(t,y)(\boldsymbol\zeta^X)^\top
+
\boldsymbol\Sigma_{RY}(t,y)P^{YX}
\right]^\top
\boldsymbol\pi
\\
&\qquad
+
\frac12
P^{XX}x^2
\boldsymbol\pi^\top
\boldsymbol\Sigma(t,y)
\boldsymbol\pi.
\end{aligned}
\end{equation}
Only the wealth row $\boldsymbol\zeta^X$ enters because the factor diffusion is
control independent; $\boldsymbol\Sigma_{RY}P^{YX}$ is the hedging term.
The separable consumption block is
$e^{-\rho t}U(C)-\lambda^X C$.  If $P^{XX}<0$ and
$\boldsymbol\Sigma(t,y)\succ0$, the risky-weight block is a strictly
concave QP under affine constraints, without requiring definiteness of the
full matrix $P$.  The mixed term vanishes when $d_Y=0$.

\subsection{Adjoint-to-Control Recovery}
\label{sec:constrained_local_recovery}

Fix $(t,\mathbf s)$ and an adapted adjoint input
\[
\vartheta
:=
(\lambda,\boldsymbol\zeta,P),
\]
and write
\[
\mathbb H_\vartheta(\mathbf u)
:=
\mathbb H(t,\mathbf s;\mathbf u,
\lambda,\boldsymbol\zeta,P).
\]
Let $\mathcal R(t,\mathbf s;\vartheta)$ denote the set of local maximizers of
$\mathbb H_\vartheta$ over the selected feasible branch of
$\mathcal U(t,\mathbf s)$.  A selected recovered control therefore satisfies
\[
\mathbf u^\star(\vartheta)
\in
\mathcal R(t,\mathbf s;\vartheta).
\]
In general, $\mathcal R$ is a local set-valued correspondence.  When
$\vartheta=\vartheta_t^*$ and $\mathbf s=\mathbf S_t^*$,
\eqref{eq:open_loop_max_condition} implies
$\mathbf u_t^*\in\mathcal R(t,\mathbf S_t^*;\vartheta_t^*)$; selecting that
local branch gives $\mathbf u^\star(\vartheta_t^*)=\mathbf u_t^*$.

Recovery estimates $\vartheta$ and solves the resulting static constrained
problem at each deployment point.  It is conditional on the supplied adjoints
and local unless concavity and convexity globalize the stagewise problem.
Full investment is already encoded by
$\pi_0=1-\mathbf1^\top\boldsymbol\pi$; other equalities may be added through
standard KKT multipliers.

For a selected local maximizer $\mathbf u^\star$, write
$\mathcal I^\star:=\mathcal I((t,\mathbf s),\mathbf u^\star)$ for the active
set defined in Section~\ref{sec:chart_open_loop_adjoints}.
A local KKT branch is called regular when its active set is locally constant
and the active control gradients
\[
\left\{
\nabla_{\mathbf u}\Gamma_j
(t,\mathbf s;\mathbf u^\star)
\right\}_{j\in\mathcal I^\star}
\]
satisfy the linear independence constraint qualification (LICQ).  At a smooth
local maximizer on such a branch, there exists a multiplier vector
$\nu^\star\ge\boldsymbol0$ satisfying the KKT necessary conditions
\begin{equation}\label{eq:local_recovery_kkt}
\begin{aligned}
\nabla_{\mathbf u}\mathbb H_\vartheta(\mathbf u^\star)
+
D_{\mathbf u}\Gamma
(t,\mathbf s;\mathbf u^\star)^\top\nu^\star
&=
\boldsymbol0,
\\
\Gamma(t,\mathbf s;\mathbf u^\star)
&\ge
\boldsymbol0,
\\
\nu_j^\star
\Gamma_j(t,\mathbf s;\mathbf u^\star)
&=
0,
\qquad
j=1,\ldots,m_\Gamma.
\end{aligned}
\end{equation}

At the optimal tuple, let $\nu^\star$ be the generalized-Hamiltonian KKT
multiplier.  By \eqref{eq:hamiltonian_gradient_identity}, it also satisfies the
first-order Hamiltonian stationarity relation, and hence
\begin{equation}\label{eq:moving_fiber_multiplier_identity}
D_{\mathbf s}\bar{\mathcal H}
=
D_{\mathbf s}\mathcal H
+
D_{\mathbf s}\Gamma^\top\nu^\star.
\end{equation}
At the selected KKT point, the fixed-coordinate second derivative is the
chart pullback of the Hessian of
$\mathcal H+\nu^{\star\top}\Gamma$, with $\nu^\star$ held fixed, including
the mixed terms in \eqref{eq:fixed_latent_hamiltonian_chain_rule}.  Among the
adjoint components, portfolio recovery uses only $\lambda^X$,
$\boldsymbol\zeta^X$, $P^{XX}$, and $P^{YX}$.

\subsection{Exact QP and Barrier Realizations}
\label{sec:recovery_realizations}

Exact QP recovery solves the quadratic-affine stagewise problem without
barrier bias.  The barrier realization targets the same local KKT
correspondence on a regular smooth branch and extends to more general smooth
geometry.

\subsubsection{Exact QP Realization of the Risky-Weight Block}

When the portfolio restrictions are affine and separable from the remaining
controls, let $\mathcal K(t,\mathbf s)$ denote the resulting convex feasible
set for the risky weights.  For $\mathbf s=(x,y)$, define
\begin{equation}\label{eq:factor_qp_coefficients}
\begin{aligned}
\mathbf g(t,\mathbf s;\vartheta)
&:=
x
\left[
\lambda^X\boldsymbol\alpha(t,y)
+
\boldsymbol v(t,y)(\boldsymbol\zeta^X)^\top
+
\boldsymbol\Sigma_{RY}(t,y)P^{YX}
\right],
\\
Q(t,\mathbf s;\vartheta)
&:=
-x^2P^{XX}\boldsymbol\Sigma(t,y).
\end{aligned}
\end{equation}
The exact risky-weight recovery is
\begin{equation}\label{eq:exact_qp_recovery}
\boldsymbol\pi^{\rm QP}(t,\mathbf s;\vartheta)
\in
\operatorname*{arg\,max}_{\boldsymbol\pi\in\mathcal K(t,\mathbf s)}
\left\{
\mathbf g(t,\mathbf s;\vartheta)^\top\boldsymbol\pi
-
\frac12
\boldsymbol\pi^\top
Q(t,\mathbf s;\vartheta)
\boldsymbol\pi
\right\}.
\end{equation}
For the borrowing-allowed no-short-sale model,
\[
\mathcal K(t,\mathbf s)=\mathbb R_+^n.
\]
The no-borrowing model additionally imposes
\[
\mathbf1^\top\boldsymbol\pi\le1.
\]
More general affine portfolio restrictions can be incorporated into
$\mathcal K(t,\mathbf s)$.

If $P^{XX}<0$ and $\boldsymbol\Sigma(t,y)\succ0$, then
$Q(t,\mathbf s;\vartheta)\succ0$, so
\eqref{eq:exact_qp_recovery} has a unique solution.  When the affine portfolio constraints are separable from the consumption
constraints, it exactly evaluates the risky-weight component of
$\mathcal R(t,\mathbf s;\vartheta)$ without barrier bias, conditional on the
supplied adjoints, and the scalar consumption block can be solved
independently.  When
$d_Y=0$, the mixed term $\boldsymbol\Sigma_{RY}P^{YX}$ is absent.

\subsubsection{Log-Barrier Realization on a Regular KKT Branch}

Define the strict interior by
\[
\mathcal U^\circ(t,\mathbf s)
:=
\left\{
\mathbf u\in\mathbb R^{n+1}:
\Gamma(t,\mathbf s;\mathbf u)>\boldsymbol0
\right\}.
\]
Explicit equality constraints, if present, are retained in the stagewise KKT
system or eliminated by parameterization; they are not included in the
logarithmic barrier.
For $\varepsilon_{\rm bar}>0$, let
$\mathbf u_{\varepsilon_{\rm bar}}$ denote a strict local maximizer, on the
selected branch, of
\begin{equation}\label{eq:barrier_stage_problem}
\mathbf u_{\varepsilon_{\rm bar}}
\in
\operatorname*{arg\,max}^{\rm loc}_{
\mathbf u\in\mathcal U^\circ(t,\mathbf s)}
\left\{
\mathbb H_\vartheta(\mathbf u)
+
\varepsilon_{\rm bar}
\sum_{j=1}^{m_\Gamma}
\log\Gamma_j(t,\mathbf s;\mathbf u)
\right\}.
\end{equation}
Define the associated barrier multipliers by
\[
\nu_{\varepsilon_{\rm bar},j}
:=
\frac{\varepsilon_{\rm bar}}
{\Gamma_j(t,\mathbf s;\mathbf u_{\varepsilon_{\rm bar}})}.
\]
They satisfy the perturbed complementarity conditions
\[
\nu_{\varepsilon_{\rm bar},j}
\Gamma_j(t,\mathbf s;\mathbf u_{\varepsilon_{\rm bar}})
=
\varepsilon_{\rm bar},
\qquad
j=1,\ldots,m_\Gamma.
\]
These solutions form the local central path, which approaches the selected
KKT point as $\varepsilon_{\rm bar}\downarrow0$ under the regularity below.
Finite-barrier bias is distinct from numerical solver error.

\begin{assumption}[Barrier/KKT Stagewise Regularity]
\label{ass:barrier_kkt}
Fix $(t,\mathbf s,\vartheta)$ and a selected KKT solution
$\mathbf u^\star(\vartheta)$.  Locally:
\begin{enumerate}[label=(\roman*),leftmargin=*]

\item \emph{Interior localization.}
The selected branch admits strict interior points, and the relevant local
barrier superlevel sets remain in a common compact neighborhood for all
sufficiently small $\varepsilon_{\rm bar}>0$.

\item \emph{Strong KKT regularity.}
At $\mathbf u^\star(\vartheta)$, LICQ, strict complementarity, and the strong
second-order sufficient condition hold.  Consequently, the bordered KKT
Jacobian is nonsingular.

\end{enumerate}
For the quadratic portfolio block, $P^{XX}<0$ and
$\boldsymbol\Sigma(t,y)\succ0$ provide the required curvature in the
risky-weight directions.  For a general nonlinear control block, the
corresponding reduced-Hessian condition is imposed directly.
\end{assumption}

For fixed $(t,\mathbf s,\vartheta)$, define the stagewise Lagrangian
\[
\mathcal L_\vartheta(\mathbf u,\nu)
:=
\mathbb H_\vartheta(\mathbf u)
+
\nu^\top\Gamma(t,\mathbf s;\mathbf u).
\]

\begin{proposition}[Local barrier--KKT approximation]
\label{prop:barrier_policy}
Suppose Assumptions~\ref{ass:Gamma},
\ref{ass:chart_pmp_regularity}, and
\ref{ass:barrier_kkt} hold.  Let
$\mathbf u_{\varepsilon_{\rm bar}}$ be the strict local barrier maximizer on
the central-path branch issuing from the selected KKT maximizer
$\mathbf u^\star$.  There exist positive constants
\[
C_{\rm cp},
\qquad
c_{\rm int},
\qquad
C_{\mathcal L},
\qquad
C_{H,1},
\qquad
C_{H,2},
\qquad
\varepsilon_0,
\]
depending on the fixed stagewise state and adjoint input, such that for
$0<\varepsilon_{\rm bar}\le\varepsilon_0$,
\begin{align}
\left\|
\mathbf u_{\varepsilon_{\rm bar}}
-
\mathbf u^\star
\right\|
&\le
C_{\rm cp}\varepsilon_{\rm bar},
\label{eq:barrier_control_error}
\\
\left|
\mathcal L_\vartheta(\mathbf u^\star,\nu^\star)
-
\mathcal L_\vartheta
(\mathbf u_{\varepsilon_{\rm bar}},\nu^\star)
\right|
&\le
C_{\mathcal L}\varepsilon_{\rm bar}^2,
\label{eq:barrier_lagrangian_gap}
\\
0
\le
\mathbb H_\vartheta(\mathbf u^\star)
-
\mathbb H_\vartheta(\mathbf u_{\varepsilon_{\rm bar}})
&\le
C_{H,1}\varepsilon_{\rm bar}
+
C_{H,2}\varepsilon_{\rm bar}^2.
\label{eq:barrier_hamiltonian_gap}
\end{align}
Moreover,
\begin{equation}\label{eq:barrier_slack_bound}
\Gamma_j
(t,\mathbf s;\mathbf u_{\varepsilon_{\rm bar}})
\ge
c_{\rm int}\varepsilon_{\rm bar},
\qquad
j=1,\ldots,m_\Gamma.
\end{equation}
If no inequality is active at $\mathbf u^\star$, the linear term in
\eqref{eq:barrier_hamiltonian_gap} vanishes.
\end{proposition}

\begin{proof}
See Online Supporting Information~\ref{app:proof_barrier_policy}.
\end{proof}

Numerical realizations of the exact QP and log-barrier solvers are detailed in
Online Supporting Information~\ref{app:solver_details}.

\subsection{Feedback Optimizers and the Smooth Markov Case}
\label{sec:markov_reduction}

The open-loop formulation does not require a smooth value function or a
Markov optimizer.  Under the full-fiber classical verification conditions of
Section~\ref{sec:hjb_motivation}, a Markov feedback can represent the
open-loop value, but a local-chart maximizer alone cannot
\citep{fleming2006controlled,yong2012stochastic}.
Online Supporting Information~\ref{app:feedback_verification} states the corresponding classical
verification result formally.

For the additional smooth identification used for interpretation, suppose
that $V$ is sufficiently smooth on the state region corresponding to the
chart patch and that a Markov latent selector
$\mathbf a^*(t,\mathbf s)\in\mathfrak A$ represents a full-fiber HJB maximizer
on that region.  Along the resulting optimal trajectory, define
\[
\mathbf u^*(t,\mathbf s)
:=
\Psi(t,\mathbf s,\mathbf a^*(t,\mathbf s)),
\qquad
\mathbf u_t^*
:=
\mathbf u^*(t,\mathbf S_t^*),
\qquad
\boldsymbol\sigma_t^*
:=
\boldsymbol\sigma_S(t,\mathbf S_t^*;\mathbf u_t^*).
\]
The smooth maximum-principle/dynamic-programming identification then gives
\begin{equation}\label{eq:markov_adjoint_identification}
\lambda_t^*
=
D_{\mathbf s}V(t,\mathbf S_t^*),
\qquad
\mathbf Z_t^*
=
D^2_{\mathbf s\mathbf s}V(t,\mathbf S_t^*)
\boldsymbol\sigma_t^*,
\qquad
\boldsymbol\zeta_t^*
=
\left[
D^2_{\mathbf s\mathbf s}V(t,\mathbf S_t^*)
-
P_t^*
\right]
\boldsymbol\sigma_t^*.
\end{equation}
These identities follow from the smooth PMP--dynamic-programming relation,
It\^o's formula, and \eqref{eq:shifted_zeta}; they are not used to define the
chart-reduced adjoint systems or recovery map
\citep{zhou1991unified,yong2012stochastic}.  In general
$P_t^*\ne D^2_{\mathbf s\mathbf s}V(t,\mathbf S_t^*)$, so smoothness and
Markovianity do not imply $\boldsymbol\zeta_t^*=0$.
Online Supporting Information~\ref{app:smooth_pmp_dp_identities} records the corresponding
wealth-component and portfolio specializations.

\begin{remark}[Markovization]
\label{rem:markovization_scope}
A finite-dimensional Markov representation is exact only when an adapted
finite-dimensional statistic has closed controlled dynamics and carries all
coefficients, rewards, and constraints.  Otherwise state compression defines
a surrogate problem whose error is separate from the errors analyzed here.
Within the finite-dimensional state formulation of this paper, the
chart-reduced adjoint systems and \OLBPTT correspondence do not rely on a
Markov feedback optimizer or on smooth dynamic-programming identities.
\end{remark}

\section{From a Reference Policy to Chart-Reduced Adjoint Processes}
\label{sec:BPTT_multiasset}

A pointwise-feasible Markov-form actor generates reference rollouts.  After
training, its realized latent outputs are detached, continuation payoffs are
differentiated once and twice on the fixed-latent graph, and conditional
projection yields the adapted adjoints used by the exact local-QP or barrier
recovery stage.  Actor
derivatives are removed during harvesting, but structural chart derivatives
remain.  Figure~\ref{fig:computational_architecture} summarizes this
computational separation.

\begin{figure}[!htbp]
\centering
\includegraphics[width=0.92\textwidth]{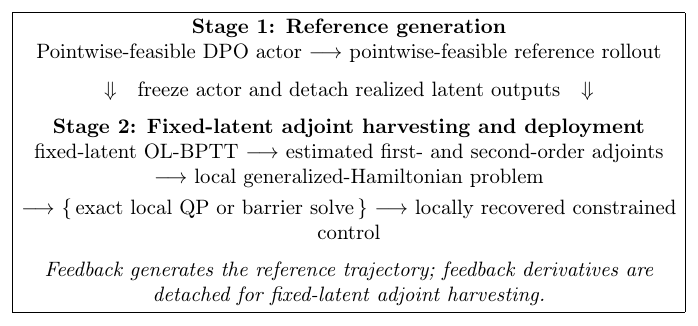}
\caption{Computational architecture.  The chart remains state dependent after
detachment, so its structural state derivatives are retained.}
\label{fig:computational_architecture}
\end{figure}

\subsection{Feasible Reference-Policy Training}
\label{sec:feedback_fixed_latent}

Let $\Theta=(\theta,\phi)$ parameterize the portfolio and consumption outputs,
and let $\varpi$ be a sampling law for initial time--state pairs
$(t_{\rm in},\mathbf s_{\rm in})$.  The reported implementation uses discounting relative to each sampled initial
time.  Accordingly, the Stage~1 objective is
\begin{equation}\label{eq:training_objective}
\begin{aligned}
J_{\rm train}(\Theta)
:=
\mathbb E_{(t_{\rm in},\mathbf s_{\rm in})\sim\varpi}
\mathbb E\!\left[
\int_{t_{\rm in}}^T
e^{-\rho(s-t_{\rm in})}U(C_\Theta(s,\mathbf S_s))\,ds
+
K e^{-\rho(T-t_{\rm in})}U(X_T)
\right].
\end{aligned}
\end{equation}
For a fixed initial pair, this is the calendar-time payoff of
Section~\ref{sec:model_setup} multiplied by the positive constant
$e^{\rho t_{\rm in}}$, so it has the same conditional optimizer.  When
$t_{\rm in}$ is randomized, the multiplication reweights the initial-time law;
reference-performance statements therefore refer to the averaged objective
actually used in training.
The state process starts from
$\mathbf S_{t_{\rm in}}=\mathbf s_{\rm in}$.  For Stage~1 only, we use the
Markov-form actor implemented in the experiments and generate the
pointwise-feasible Markov-form reference policy
\begin{equation}\label{eq:charted_feedback}
\mathbf u_\Theta(t,\mathbf s)
=
\Psi\bigl(t,\mathbf s,\mathbf a_\Theta(t,\mathbf s)\bigr).
\end{equation}
At training update $m$, fresh initial states and Brownian increments define a
Monte Carlo objective $\widehat J_{\rm train}^{(m)}$.  The policy gradient is
computed by ordinary BPTT through the closed-loop rollout graph and clipped in
Euclidean norm at one:
\[
 g_m
 :=
 \nabla_\Theta\!\left[-\widehat J_{\rm train}^{(m)}(\Theta_m)\right],
 \qquad
 \bar g_m
 :=
 \frac{g_m}{\max\{1,\|g_m\|\}}.
\]
The reported learned checkpoints use PyTorch Adam with learning rate
$10^{-4}$, default $(\beta_1,\beta_2)=(0.9,0.999)$ and numerical constant
$10^{-8}$, and no weight decay:
\[
 \Theta_{m+1}
 =
 \operatorname{Adam}_{10^{-4}}(\Theta_m,\bar g_m).
\]
In the tables, one such fresh-batch Adam update is called one \emph{epoch}.
Training uses 500 epochs except in the explicit 250/500/1,000-epoch budget
audit.  No Stage~2 adjoint or local-solver quantity enters this update.

Equation~\eqref{eq:charted_feedback} is only a reference-policy ansatz.  After
freezing the realized latent process, the adjoint calculation treats it as an
adapted open-loop coordinate process.  Any reference-to-target mismatch is
kept as $\delta_{\rm ref}$; Stage~2 does not feed back into training.  Smooth
output maps enforce pointwise feasibility, and
Published Appendix~\ref{app:fixed_latent_setup}
records their formulas, time normalization, and transition scheme.

\subsection{Freeze the Actor and Form Fixed-Latent Sensitivities}
\label{sec:pathwise_projection}

After training, the actor is frozen and its latent output is detached before
the feasible chart is applied:
\begin{equation}\label{eq:fixed_latent_control}
\bar{\mathbf a}_k
:=
\sg\!\left(
\mathbf a_\Theta(t_k,\mathbf S_k)
\right),
\qquad
\mathbf u_k^{\rm ref}
:=
\Psi(t_k,\mathbf S_k,\bar{\mathbf a}_k).
\end{equation}
Here $\sg$ denotes the stop-gradient (or detach) operator: it returns its
argument unchanged in the forward rollout but assigns it zero derivative in
the adjoint-harvesting graph.
For a state derivative initiated at time $t_k$, all later latent values
$\{\bar{\mathbf a}_\ell\}_{\ell\ge k}$ are held fixed.  Thus the closed-loop terms carried by
$D_{\mathbf s}\mathbf a_\Theta$ vanish, so the graph targets the realized
open-loop coordinate process rather than the closed-loop sensitivity obtained
by differentiating through the frozen feedback rule.  The structural
derivatives $D_{\mathbf s}\Psi$ and $D^2_{\mathbf s\mathbf s}\Psi$ of the
moving feasible fiber remain.  For example,
under the proportional chart
\[
C=Xc,
\]
detaching $c$ still leaves
\[
D_X^{\rm ol}C=c.
\]
For the theoretical sensitivity formulas below, we use calendar-value
normalization.  The relative-discount implementation differs by the
deterministic factor $e^{\rho t_{\rm in}}$, as detailed in
Published Appendix~\ref{app:fixed_latent_setup}.
Consider a rollout initialized at grid index $k_{\rm in}$.  For
$k\ge k_{\rm in}$, define the one-step and terminal payoffs by
\[
\mathfrak r_k
:=
e^{-\rho t_k}U(C_k)h,
\qquad
\mathfrak g_N
:=
K e^{-\rho T}U(X_N),
\]
and define the continuation payoff
\begin{equation}\label{eq:pathwise_continuation_payoff}
\mathcal J_k
:=
\sum_{\ell=k}^{N-1}\mathfrak r_\ell
+
\mathfrak g_N.
\end{equation}
Variable initial times may be implemented by padding trajectories to a common
grid and masking contributions before $k_{\rm in}$; this does not alter the
derivatives of $\mathcal J_k$ for $k\ge k_{\rm in}$.

The raw fixed-latent pathwise sensitivities are
\begin{equation}\label{eq:vector_raw_pathwise_sensitivities}
\widetilde\lambda_k
:=
\DOL_{\mathbf S_k}\mathcal J_k,
\qquad
\widetilde P_k
:=
\mathsf D_{\mathbf S_k\mathbf S_k}^{2,\rm ol}
\mathcal J_k.
\end{equation}
The Hessian $\widetilde P_k$ differentiates the realized continuation payoff
with future latent coordinates fixed.  Under the conditions of
Section~\ref{sec:olbptt_pmp_correspondence}, the standard interpolant of the
adapted projections $P_k=\mathbb E_k[\widetilde P_k]$ converges, as
$h\downarrow0$, to the second-order adjoint associated with the frozen
reference control, rather than to the Hessian of a reoptimized value
function.  The risky-weight decoder requires the
wealth-relevant blocks of $\lambda$ and $P$, together with the shifted wealth
row $\boldsymbol\zeta^X$ estimated below.  Accordingly, define
\[
\widetilde\lambda_k^X
:=
e_X^\top\widetilde\lambda_k,
\qquad
\widetilde P_k^{XX}
:=
e_X^\top\widetilde P_k e_X,
\]
and let $\widetilde P_k^{XY}$ denote the wealth--factor block of
$\widetilde P_k$.  This block is harvested directly rather than inferred from
a value Hessian or by imposing $\boldsymbol\zeta^X=\boldsymbol0$.

Published Appendix~\ref{app:fixed_latent_setup} records the exact chain rules; they use
neither stationarity nor optimality.

Let
\[
\mathbb E_k[\cdot]
:=
\mathbb E[\cdot\mid\mathcal F_{t_k}].
\]
The population adapted projections are
\begin{equation}\label{eq:adapted_discrete_adjoints}
\lambda_k
:=
\mathbb E_k[\widetilde\lambda_k],
\qquad
P_k
:=
\mathbb E_k[\widetilde P_k].
\end{equation}
Their numerical Monte Carlo or regression approximations are denoted by
$\widehat\lambda_k$ and $\widehat P_k$.

The intervalwise Brownian projections, standard in BSDE time
discretization \citep{bouchard2004discrete,zhang2004numerical}, retain their
orthogonal residuals:
\begin{equation}\label{eq:orthogonal_one_step_decomposition}
\begin{aligned}
\lambda_{k+1}
&=
\mathbb E_k[\lambda_{k+1}]
+\mathbf Z_k^h\Delta\mathbf W_k+\xi_{k+1},
\\
P_{k+1}
&=
\mathbb E_k[P_{k+1}]
+\sum_{\ell=1}^{d_W}R_k^{h,\ell}\Delta W_k^\ell
+\xi_{k+1}^P,
\end{aligned}
\end{equation}
where $\xi_{k+1}$ and $\xi_{k+1}^P$ have zero conditional mean and are
orthogonal to the span of the Brownian increment $\Delta\mathbf W_k$.  Published Appendix~\ref{app:martingale_projection_residuals}
controls their aggregate $L^2$ energy; no pointwise $O(h)$ order is imposed.

\subsection{\OLBPTT--Adjoint Correspondence}
\label{sec:olbptt_pmp_correspondence}

Adjoint/backpropagation correspondences are classical in deterministic
control and neural networks
\citep{dreyfus1962numerical,lecun1988theoretical}, while stochastic first- and
second-order adjoints, pathwise SDE sensitivities, and conditional Brownian
projections for BSDE discretization are also established
\citep{peng1990general,yong2012stochastic,li2020scalable,
bouchard2004discrete,zhang2004numerical}.  A recent complementary line uses
Pontryagin recursions for deterministic open-loop reinforcement learning by
optimizing fixed action sequences and estimating the required dynamics
Jacobians \citep{eberhard2025pontryagin}.  Our result is narrower in a different
direction: for a feedback-generated stochastic reference rollout, fixed-latent
\OLBPTT jointly targets the chart-reduced adapted first- and second-order
open-loop adjoints, while retaining moving-fiber state derivatives and
orthogonal martingale residuals.  For the continuous-time limit, following
\citet{zhang2004numerical}, we impose the standard piecewise-constant $L^2$
approximation property with aggregate squared error $O(h)$; the condition is
integrated in time, not pointwise.

\begin{assumption}[Regularity for the chart-reduced \OLBPTT--adjoint correspondence]
\label{ass:bptt_pmp_reg}
In addition to Assumptions~\ref{ass:Gamma} and
\ref{ass:chart_pmp_regularity}, assume:
\begin{enumerate}[label=(\roman*),leftmargin=*]

\item The rollout remains in a common compact chart patch, and the frozen
latent processes are uniformly square integrable.

\item The states and raw vector/matrix sensitivities are uniformly square
integrable, and the interval martingale densities have mesh-uniform
$L^2$ energy.

\item The frozen latent processes, charted controls, states, barred
coefficients, and their first two fixed-latent state derivatives converge in
the corresponding $L^2$ spaces.

\item For the interval martingale densities in
\eqref{eq:appendix_first_interval_density}--
\eqref{eq:appendix_second_interval_density} and their conditional
piecewise-constant projections in
\eqref{eq:discrete_martingale_coefficients}, there is a mesh-independent
constant $C_{\rm proj}$ such that, uniformly over the common localization and
the admissible initial index,
\begin{equation}\label{eq:bptt_z_h_regularization_assumption}
\begin{aligned}
\sum_{k=k_{\rm in}}^{N-1}
\mathbb E\int_{t_k}^{t_{k+1}}
\Bigg[
&
\|\mathcal Z_s^{\lambda,k+1}-\mathbf Z_k^h\|_F^2
\\
&+
\sum_{\ell=1}^{d_W}
\|\mathcal Z_s^{P,k+1,\ell}-R_k^{h,\ell}\|_F^2
\Bigg]ds
\le
C_{\rm proj}h.
\end{aligned}
\end{equation}

\end{enumerate}
\end{assumption}

Lemma~\ref{lem:fixed_latent_one_step_expansion} gives the required one-step
stochastic expansion and remainder bounds.

\begin{theorem}[Vector-state \OLBPTT sensitivities and chart-reduced adjoint correspondence]
\label{thm:bptt_pmp_constrained}
Suppose Assumptions~\ref{ass:Gamma},
\ref{ass:chart_pmp_regularity}, and
\ref{ass:bptt_pmp_reg} hold on a common compact chart patch.

\noindent\textbf{(a) Exact pathwise identities.}
The raw derivatives satisfy
\eqref{eq:bptt_chain_rule}--\eqref{eq:bptt_second_order} exactly.  Derivatives
through the actor output vanish, while the state derivatives of the feasible
chart remain.

\noindent\textbf{(b) Adapted discrete recursions.}
The conditional projections satisfy the discrete first- and second-adjoint
recursions \eqref{eq:first_projected_recursion}--
\eqref{eq:second_projected_recursion}, derived in
Published Appendix~\ref{app:proof_bptt_pmp}.  Both recursions retain the orthogonal residuals in
\eqref{eq:orthogonal_one_step_decomposition}.  The aggregate $L^2$ projection
and remainder conditions, rather than a pointwise order for either residual,
control their continuous-time limit.

\noindent\textbf{(c) Continuous-time limit.}
Let
$(\lambda^h,\mathbf Z^h,P^h,\mathbf R^h)$
denote the standard adapted interpolants of the discrete adjoints and
one-step Brownian coefficients, with
\[
\|\mathbf R\|_F^2
:=
\sum_{\ell=1}^{d_W}\|R^\ell\|_F^2.
\]
Then
\begin{equation}\label{eq:full_adjoint_bsde_convergence}
\begin{aligned}
&
\mathbb E
\sup_{t_{\rm in}\le t\le T}
\left(
\|\lambda_t^h-\lambda_t\|^2
+
\|P_t^h-P_t\|_F^2
\right)
\\
&\quad
+
\mathbb E
\int_{t_{\rm in}}^T
\left(
\|\mathbf Z_t^h-\mathbf Z_t\|_F^2
+
\|\mathbf R_t^h-\mathbf R_t\|_F^2
\right)dt
\longrightarrow0.
\end{aligned}
\end{equation}
The limit is the chart-reduced vector first-adjoint and matrix second-adjoint
system associated with the limiting fixed-latent reference control.  If that
reference control is optimal, the limits are the corresponding PMP adjoints for the open-loop control problem.

\noindent\textbf{(d) Wealth-relevant blocks.}
For $\mathbf S=(X,Y)$, the conditional projections of
$\widetilde\lambda^X$, $\widetilde P^{XX}$, and
$\widetilde P^{XY}$ converge to
$\lambda^X$, $P^{XX}$, and $P^{XY}$.  These are the blocks entering the
risky-weight coefficients in
\eqref{eq:factor_qp_coefficients}.  No identification of $P^{XY}$ with
$V_{XY}$ is required.
\end{theorem}

\begin{proof}
See Published Appendix~\ref{app:proof_bptt_pmp}.
\end{proof}

The recovery stage uses the harvested $\lambda$ and $P$ together with the
shifted wealth row $\boldsymbol\zeta^X$ estimated in
Section~\ref{sec:shifted_input_estimation}; $\mathbf R$ closes the
second-adjoint dynamics but is not a decoder input.

\subsection{Adjoint Estimation and Deployment}
\label{sec:shifted_input_estimation}

Conditional projection provides $\lambda^X$, $P^{XX}$, and $P^{XY}$ directly.
We estimate the remaining shifted row $\boldsymbol\zeta^X$ from the
one-step shock response rather than obtain it by imposing
$\boldsymbol\zeta^X=\boldsymbol0$ or a value-Hessian identity.

To make that decomposition explicit, let
\(
\vartheta^{\rm ref}
=(\lambda^{\rm ref},\boldsymbol\zeta^{\rm ref},P^{\rm ref})
\)
be the exact chart-reduced tuple associated with the frozen reference control.
By \eqref{eq:shifted_zeta}, its wealth row satisfies
\begin{equation}\label{eq:wealth_row_shift_decomposition}
\mathbf Z_t^{X,\rm ref}
=
e_X^\top P_t^{\rm ref}
\boldsymbol\sigma_S
(t,\mathbf S_t^{\rm ref};\mathbf u_t^{\rm ref})
+
\boldsymbol\zeta_t^{X,\rm ref}.
\end{equation}
The first term is the anchor and the second the residual component.  To reduce
future-shock noise, we use antithetic current increments, common future paths,
nested averaging, and regression of the odd paired response.  An even
$M_{\rm out}$ counts signed outer paths, hence $M_{\rm out}/2$ independent
pairs.  For
$i=1,\ldots,M_{\rm out}/2$, draw
$d_i\sim N(0,hI_{d_W})$, pair it with $-d_i$, and let
$\bar\lambda_{i,+}^X$ and $\bar\lambda_{i,-}^X$ denote the corresponding
nested-average estimates of the next-step conditional adjoints.  Define
\[
(D_k)_{i\cdot}:=d_i^\top,
\qquad
y_i:=\frac{\bar\lambda_{i,+}^X-\bar\lambda_{i,-}^X}{2},
\qquad
y:=(y_1,\ldots,y_{M_{\rm out}/2})^\top,
\]
so that $D_k\in\mathbb R^{(M_{\rm out}/2)\times d_W}$ and
$y\in\mathbb R^{M_{\rm out}/2}$.  The anchored residual estimator is
\begin{align}
c_k^{X,\rm anc}
&:=
\left[
e_X^\top\widehat P_k
\boldsymbol\sigma_S
(t_k,\mathbf S_k;\mathbf u_k^{\rm ref})
\right]^\top,
\notag\\
\widehat\zeta_{k,\rm col}^X
&:=
(D_k^\top D_k+\kappa_ZI)^{-1}D_k^\top
\bigl(y-D_kc_k^{X,\rm anc}\bigr),
\qquad
\widehat{\boldsymbol\zeta}_k^X
:=
(\widehat\zeta_{k,\rm col}^X)^\top,
\quad \kappa_Z\ge0,
\label{eq:anchored_shift_regression}
\\
\widehat{\mathbf Z}_k^X
&:=
\left(c_k^{X,\rm anc}+\widehat\zeta_{k,\rm col}^X\right)^\top.
\label{eq:anchored_Z_reconstruction}
\end{align}
Here $\kappa_Z$ is a ridge parameter; $\kappa_Z=0$ gives anchored ordinary
least squares when the design has full column rank.  Stage~2 uses
$\widehat{\boldsymbol\zeta}_k^X$ directly, while
$\widehat{\mathbf Z}_k^X$ is reconstructed only for validation.  Online Supporting Information~\ref{app:shifted_row_estimator}
gives implementation details.  Algorithm~\ref{algo:end_to_end_recovery}
summarizes the complete acquisition-and-recovery pipeline.

\begin{algorithm}[H]
\caption{Fixed-latent adjoint acquisition and constrained local recovery}
\label{algo:end_to_end_recovery}
\begin{algorithmic}[1]
\STATE Sample initial time--state pairs and Brownian increments, and train the
pointwise-feasible Markov-form reference policy by maximizing
\eqref{eq:training_objective}.
\STATE Freeze the actor and detach every latent output before applying the
state-dependent feasible chart.
\STATE Differentiate each continuation payoff once and twice with respect to
the full state to obtain
$(\widetilde\lambda_k,\widetilde P_k)$.
\STATE Approximate the conditional projections
$(\lambda_k,P_k)$ by
$(\widehat\lambda_k,\widehat P_k)$.
\STATE When required, estimate
$\widehat{\boldsymbol\zeta}_k^X$
using \eqref{eq:anchored_shift_regression}.
\STATE Substitute the estimated wealth-relevant tuple into the exact local-QP
or barrier recovery formulation of
Section~\ref{sec:recovery_realizations}.
\end{algorithmic}
\end{algorithm}

\subsection{End-to-End Recovery from Reference Performance}
\label{sec:main_end_to_end_recovery}

The harvested tuple belongs to the frozen reference control, not directly to
the optimizer.  The next result makes the resulting two-stage error chain
explicit.  Here $\delta_{\rm ol}$ collects time-discretization, conditional
Monte Carlo, and projection/regression errors in
$(\widehat\lambda,\widehat P)$; $\delta_\zeta$ is the shifted-row error; and
$\delta_{\rm qp}$ bounds the stationarity residual of the numerical QP solve.
Let $G_t^{\rm qp}$ denote the loss of the recovered risky weights in the
true-input local quadratic objective in
\eqref{eq:exact_qp_recovery}.  We use the process norms
$\mathbb S^2$, $\mathbb H^2$, and $L^{q,p}$ defined in
Published Appendix~\ref{app:notation_summary}.

For transparency, we record the three quantitative conditions that drive the
end-to-end rate.  On a localized comparison neighborhood
$\mathfrak U_0$, assume the quadratic-growth bound
\begin{equation}\label{eq:main_qg_condition}
J(\mathbf u^*)-J(\mathbf u)
\ge
\kappa_{\rm qg}\|\mathbf u-\mathbf u^*\|_{\mathbb H^2}^2,
\qquad \mathbf u\in\mathfrak U_0,
\end{equation}
the acquisition-error bounds
\begin{equation}\label{eq:main_acquisition_error_condition}
\| (\widehat\lambda,\widehat P)
-(\lambda^{\rm ref},P^{\rm ref})\|_{L^{2,2}}
\le C_{\rm ol}\delta_{\rm ol},
\qquad
\|\widehat{\boldsymbol\zeta}
-\boldsymbol\zeta^{\rm ref}\|_{L^{2,2}}
\le\delta_\zeta,
\end{equation}
and, at the common deployment-state process $\mathbf S$, the bridge
\begin{equation}\label{eq:main_common_state_bridge}
\|\vartheta^{\rm ref}(\cdot,\mathbf S_\cdot)
-\vartheta^*(\cdot,\mathbf S_\cdot)\|_{L^{2,2}}
\le
C_{\rm cs}\sqrt{\varepsilon_{\rm ref}/\kappa_{\rm qg}}.
\end{equation}
The remaining coefficient-jet, adjoint-boundedness, common-state, and
positive-definiteness conditions are stated precisely in Published
Appendices~C--D.

\begin{theorem}[End-to-end stability of QP-PGDPO]
\label{thm:end_to_end_qp_recovery}
Suppose Assumptions~\ref{ass:Gamma} and
\ref{ass:chart_pmp_regularity}, the displayed bounds
\eqref{eq:main_qg_condition}--\eqref{eq:main_common_state_bridge}, and the
localized regularity conditions of Published Appendices~C--D hold.  Let
$\mathbf u^{\rm ref}\in\mathfrak U_0$ and set
\[
\varepsilon_{\rm ref}
:=
J(\mathbf u^*)-J(\mathbf u^{\rm ref})
\ge0.
\]
Then, along the paired reference and optimal trajectories,
\begin{align}
\|\mathbf u^{\rm ref}-\mathbf u^*\|_{\mathbb H^2}
&\le
\sqrt{\frac{\varepsilon_{\rm ref}}{\kappa_{\rm qg}}},
\label{eq:main_reference_control_rate}
\\
\|\mathbf S^{\rm ref}-\mathbf S^*\|_{\mathbb S^2}
+
\|\vartheta^{\rm ref}-\vartheta^*\|_{L^{2,2}}
&\le
C_{\rm traj}
\sqrt{\frac{\varepsilon_{\rm ref}}{\kappa_{\rm qg}}}.
\label{eq:main_reference_adjoint_rate}
\end{align}
If, in addition, the uniform positive-definiteness, Lipschitz,
bounded-localization, and numerical-QP residual conditions of
Proposition~\ref{prop:qp_policy_gap} in the Published Appendix hold, define
\begin{equation}\label{eq:main_end_to_end_error}
\eta_{\rm e2e}
:=
\delta_{\rm ol}
+
\delta_\zeta
+
\delta_{\rm qp}
+
\sqrt{\frac{\varepsilon_{\rm ref}}{\kappa_{\rm qg}}}.
\end{equation}
Then, at common deployment states, the QP-PGDPO risky-weight output satisfies
\begin{align}
\|\widehat{\boldsymbol\pi}^{\rm QP}-\boldsymbol\pi^*\|_{L^{2,2}}
&\le
C_\pi\eta_{\rm e2e},
\label{eq:main_qp_policy_rate}
\\
\|G^{\rm qp}\|_{L^{1,1}}
&\le
C_G\eta_{\rm e2e}^2.
\label{eq:main_qp_gap_rate}
\end{align}
Consequently, with exact adjoint acquisition, shifted-row estimation, and QP
solution, the recovered risky-weight error is
$O(\sqrt{\varepsilon_{\rm ref}})$ and the integrated true local QP gap is
$O(\varepsilon_{\rm ref})$.  The constants are uniform only on the stated
localizations and may depend on portfolio dimension and covariance
conditioning.
\end{theorem}

\begin{proof}
Equations~\eqref{eq:main_reference_control_rate}--
\eqref{eq:main_reference_adjoint_rate} follow from
Theorem~\ref{thm:reference_performance_adjoint}.  The common-state bridge
replaces the recovery input mismatch by
$O(\sqrt{\varepsilon_{\rm ref}/\kappa_{\rm qg}})$.  Substitution into the
QP stability and quadratic-gap bounds of
Proposition~\ref{prop:qp_policy_gap} gives
\eqref{eq:main_qp_policy_rate}--\eqref{eq:main_qp_gap_rate}; the detailed
argument is recorded in Corollary~\ref{cor:qp_consistency}.
\end{proof}

The distinction between paired trajectories and common deployment states is
essential: the first follows from quadratic growth, whereas the second is an
additional coverage/stability requirement.  Under the regular barrier-branch conditions in
Assumption~\ref{ass:barrier_kkt} and Online Supporting Information
Assumption~\ref{ass:barrier_recovery_stability},
Online Supporting Information Theorem~\ref{thm:ppgdpo_gap_barrier}
similarly gives
\[
\|\widehat{\mathbf u}^{\rm bar}-\mathbf u^*\|_{L^{2,2}}
\lesssim
\delta_{\rm ol}
+\delta_\zeta
+\delta_{\rm bar}
+\varepsilon_{\rm bar}
+\sqrt{\frac{\varepsilon_{\rm ref}}{\kappa_{\rm qg}}}.
\]
Unlike the QP gap, the barrier Hamiltonian gap can retain linear terms when
constraints are active.

\section{Numerical Validation}\label{sec:num_test}

We compare pointwise-feasible DPO deployment with B-PGDPO and QP-PGDPO, which
use the same estimated wealth-relevant adjoints and differ only in the Stage~2
solver.  All adjoints come from the fixed-latent \OLBPTT graph, and every
recovery table uses the estimated shifted row.  Because the reported portfolio
blocks are quadratic-affine conditional on the adjoints, QP-PGDPO is the
common-input exact local-QP benchmark; the barrier runs audit its central-path
approximation.  Unless otherwise stated, the frozen Stage~1 DPO actor is the
\emph{reference policy}; an \emph{analytical target} comes from a closed form
or an ordinary differential equation (ODE), and a \emph{full-shift}
decoder uses $\widehat{\boldsymbol\zeta}^X$ (a zero-shift diagnostic sets only
this row to zero).  The experiments test adjoint accuracy, local recovery,
predictable returns, a moving consumption fiber, external neural comparisons,
and action dimensions up to $n=100$.

For runs with $\rho>0$, the implementation restarts each held-out initial
state with discounts relative to that state.  Numerical adjoints and local
Hamiltonian diagnostics are therefore reported in the corresponding
current-value normalization; we suppress a ``cv'' subscript on hatted
numerical quantities.  Calendar-value adjoints are obtained by multiplying by
$e^{-\rho t}$.  The exact QP/KKT control is unchanged by this common positive
scaling.  The reported barrier parameter $\varepsilon_{\rm bar}$ is likewise
stated in current-value units (its calendar-value counterpart at fixed $t$ is
$e^{-\rho t}\varepsilon_{\rm bar}$).

\subsection{Questions, Held-Out States, and Diagnostics}\label{sec:benchmark_scope}

All policy and KKT diagnostics are evaluated cross-sectionally at held-out
initial time--state points, independently of training minibatches and within
the declared compact ranges.  Paired method comparisons reuse common states
and Brownian continuations; shifted-input replication uses an independent
outer--inner bank.  Online Supporting Information~\ref{app:diagnostic_protocol} records the
sampling ranges and seeds.
Table~\ref{tab:market_calibration_summary} records the market and constraint calibrations.  The unconstrained affine-factor analytical
validation is in Online Supporting Information~\ref{app:unconstrained_predictable_return_validation}
(Table~\ref{tab:jun_liu_mixed_adjoint_validation}).

Let $N_{\rm test}$ denote the number of held-out diagnostic states.  For an
available target control $\mathbf u_i^*$, the coordinatewise root mean squared
error (RMSE) is
\[
{\rm RMSE}_u
=\left[\frac{1}{N_{\rm test}d_u}
\sum_{i=1}^{N_{\rm test}}
\|\widehat{\mathbf u}_i-\mathbf u_i^*\|_2^2\right]^{1/2}.
\]
For a generic analytical target $A_i$, the normalized root mean squared error
(nRMSE) is
\[
 {\rm nRMSE}(A)
 :=\frac{\left[N_{\rm test}^{-1}\sum_i
 \|\widehat A_i-A_i\|_2^2\right]^{1/2}}
 {\left[N_{\rm test}^{-1}\sum_i\|A_i\|_2^2\right]^{1/2}}.
\]
At state $i$, let $\widehat{\mathbb H}_i$ denote the unregularized generalized
Hamiltonian formed from the common estimated adjoint input.  For a deployment
rule $\mathsf m$, define
\[
 {\rm Gain}_i(\mathsf m)
 :=\widehat{\mathbb H}_i(\mathbf u_i^{\mathsf m})
   -\widehat{\mathbb H}_i(\mathbf u_i^{\rm DPO}),
 \qquad
 {\rm QPGap}_i(\mathsf m)
 :=\widehat{\mathbb H}_i(\mathbf u_i^{\rm QP})
   -\widehat{\mathbb H}_i(\mathbf u_i^{\mathsf m}).
\]
These are conditional local generalized-Hamiltonian diagnostics under the
common estimated adjoint input; they are not direct estimates of global
expected-utility improvement.  The tables report statewise averages of these
quantities.  The shifted-input signal-to-noise ratio (SNR) is also formed statewise,
\[
 {\rm SNR}_i
 :=\frac{\|\widehat{\boldsymbol\zeta}_i^X\|_2}
 {\|\widehat{\operatorname{se}}
   (\widehat{\boldsymbol\zeta}_i^X)\|_2},
 \qquad
 {\rm SNR}:=N_{\rm test}^{-1}\sum_i{\rm SNR}_i,
\]
The reported SNR averages these statewise ratios; it is a Monte Carlo
resolution diagnostic, not a formal test statistic.

Local optimality is summarized by the KKT--projected-gradient (KKT--PG)
residual of the unregularized estimated Hamiltonian, together with the gain and
QP gap.  The residual is the projected displacement after a curvature-scaled
ascent step and vanishes at a constrained KKT point; B-PGDPO is evaluated
without its barrier term.  Online Supporting Information~\ref{app:diagnostic_protocol} gives the
formula and the independent and identically distributed (IID), edge, and
out-of-distribution (OOD) ranges.  Without an analytical constrained policy,
these are conditional recovery diagnostics rather than ground-truth errors.
Table~\ref{tab:deployment_rule} summarizes the three deployment rules used
throughout the numerical comparison.

\begin{table}[!htbp]
\centering
\small
\setlength{\tabcolsep}{3pt}
\caption{Deployment rules.  All methods share the same activation- or
chart-based Stage~1 rollout; B-PGDPO and QP-PGDPO differ only in Stage~2.}
\label{tab:deployment_rule}
\begin{tabular}{lccc}
\toprule
Method & Stage~1 rollout & Stage~2 recovery & Barrier term \\
\midrule
DPO & activation/chart & none & no \\
B-PGDPO & activation/chart & log-barrier solve & yes \\
QP-PGDPO & activation/chart & exact KKT/QP solve & no \\
\bottomrule
\end{tabular}
\end{table}

All reported controls satisfy the pointwise constraints; the errors and
residuals measure conditional local optimality under the common estimated
adjoints.

\paragraph{Conditional averaging and evaluation budgets.}
At held-out point $i$, each adjoint-level estimate averages $M$ raw derivatives,
\[
 \widehat A_i=M^{-1}\sum_{m=1}^M\widetilde A_i^{(m)}.
\]
Here $M$ is the continuation count for adjoint levels.  The shifted-input
budget uses even $M_{\rm out}$ signed paths
($M_{\rm out}/2$ antithetic pairs) and $M_{\rm in}$ common inner
continuations per branch.  Table~\ref{tab:main_mc_budgets} gives the budgets.

\begin{table}[!htbp]
\centering
\scriptsize
\setlength{\tabcolsep}{2.5pt}
\caption{Main checkpoint and Monte Carlo evaluation budgets.  $N_s$ is the
number of held-out states, $M$ is the number of continuations per state used
to estimate the adjoint levels $\lambda$ and $P$, and
$M_{\rm out}\!\times\!M_{\rm in}$ is the shifted-input projection budget.
Here $M_{\rm out}$ counts signed outer paths; antithetic pairing therefore
provides $M_{\rm out}/2$ independent outer units.}
\label{tab:main_mc_budgets}
\begin{tabularx}{\textwidth}{@{}lrrlX@{}}
\toprule
Experiment & $N_s$ & $M$ & $M_{\rm out}\!\times\!M_{\rm in}$
& Checkpoint/grid and recovery specification \\
\midrule
Adjoint validation, Table~\ref{tab:adjoint_level_validation}
& 128 & 2,048 & $512\!\times\!8$
& Same $n=100$ checkpoint; DPO-rollout and analytical-rollout estimators \\
No-short-sale recovery, Table~\ref{tab:noshort_ablation}
& 128 & 8,192
& \shortstack{$512\!\times\!8$ ($n=10$)\\
                 $1{,}024\!\times\!8$ ($n=100$)}
& 500 epochs; B-PGDPO $\varepsilon_{\rm bar}=10^{-6}$ \\
Training-budget study, Table~\ref{tab:training_budget_multiseed}
& 128 & 128 & $512\!\times\!4$
& Three seeds; checkpoints at 250/500/1,000 epochs; B-PGDPO $3\times10^{-6}$ \\
Predictable-return recovery, Table~\ref{tab:liu_ko_factor_bpgdpo}
& 128 & 8,192
& \shortstack{$512\!\times\!8$ ($n=10,50$)\\
                 $1{,}024\!\times\!8$ ($n=100$)}
& Fixed simplex; common estimated adjoints for DPO/B/QP \\
Constant-opportunity cap, Table~\ref{tab:consumption_cap_qp}
& 128 & 8,192
& \shortstack{$512\!\times\!8$ ($n=2,10$)\\
                 $1{,}024\!\times\!8$ ($n=100$)}
& 500 epochs; exact portfolio QP plus scalar consumption clipping \\
Unconstrained predictable-return validation,
Table~\ref{tab:jun_liu_mixed_adjoint_validation}
& 128 & 8,192 & $512\!\times\!16$
& Analytical policy; 128 rollout steps; independent replication bank \\
Barrier/switching audits, Tables~\ref{tab:barrier_final_audit}--\ref{tab:switching_region_audit}
& 64/split & 8,192 & $1{,}024\!\times\!8$
& $n=100$; double-precision central path \\
\bottomrule
\end{tabularx}
\end{table}

The no-short-sale comparison uses $\varepsilon_{\rm bar}=10^{-6}$; the
multiseed and cap/factor audits use $3\times10^{-6}$.  These are Stage~2
solver settings, not trained hyperparameters.

\subsection{Are the Harvested Adjoints Accurate?}

The constant-opportunity experiments in this and the next two subsections use
the terminal-wealth-only criterion
\[
 \mathbb E[U(X_T)],
 \qquad
 U(x)=\frac{x^{1-\gamma}}{1-\gamma},
\]
with no intermediate consumption, $r=0.03$, $\gamma=2$, $T=1.5$, and
$\Delta t=0.075$.  The unconstrained constant-opportunity solution is the
classical benchmark of \citet{merton1969lifetime,merton1971optimum}; its
no-short-sale constrained analogue is characterized by
\citet{XuShreve1992a,XuShreve1992b}.  This benchmark supplies analytical wealth
derivatives and the constrained optimal policy.
The terminal CRRA reward and wealth-homogeneous fixed-latent controls give the
benchmark-specific pathwise factorization
\begin{equation}\label{eq:benchmark_fixed_latent_crra_identity}
 \widetilde{\mathcal J}_t^a(x)
 =\frac{x^{1-\gamma}}{1-\gamma}\widetilde F_t^a,
 \qquad
 \widetilde P_t^a=-\frac{\gamma}{x}\widetilde\lambda_t^a,
 \qquad
 P_t^a=-\frac{\gamma}{X_t}\lambda_t^a.
\end{equation}
Here $\widetilde F_t^a$ is the raw pathwise random scaling factor left after
the CRRA wealth term has been factored out.  It depends on the fixed latent
control path, the market coefficients, and the future Brownian continuation,
but not on the current wealth argument $x$; the tilde indicates that no
conditional averaging has yet been applied.
In this factor-free benchmark, substituting
$P^{XX}=-(\gamma/X)\lambda^X$ into the local QP and dividing by the positive
factor $X\lambda^X$ gives the equivalent objective
\begin{equation}\label{eq:merton_crra_decoder_invariance}
 \left[
 \boldsymbol\alpha
 +\boldsymbol v
 \left(\frac{\boldsymbol\zeta^X}{\lambda^X}\right)^{\!\top}
 \right]^{\!\top}\boldsymbol\pi
 -\frac{\gamma}{2}
 \boldsymbol\pi^\top\boldsymbol\Sigma\boldsymbol\pi.
\end{equation}
Thus the absolute first-adjoint level cancels from the risky-weight decoder;
when the shift is zero, any positive multiplicative error in $\lambda^X$ leaves
the QP target unchanged.  The constant-opportunity experiment is therefore a
structural consistency test of the fixed-latent identities, the estimated
shift-to-level ratio, and the local decoder, rather than a test of sensitivity
to the absolute first-adjoint level.  The consumption-cap experiment tests that
absolute level, and the factor experiments test the mixed adjoint input.

This identity is used only as a graph-contract check; the shifted input is
estimated independently by the anchored nested regression
(see Table~\ref{tab:z_estimator_ablation} for the estimator comparison).  We evaluate 128
held-out time--wealth states with 2,048 continuations per state.  In the ``analytical-rollout \OLBPTT'' row, the analytical constrained policy
generates the rollouts, but the same finite-grid, continuation, \OLBPTT, and
conditional-projection pipeline estimates the adjoints.  Only the rollout
source is analytical.

\begin{table}[!htbp]
\centering
\scriptsize
\setlength{\tabcolsep}{2pt}
\caption{Direct adjoint validation in the $n=100$ borrowing-allowed
no-short-sale constant-opportunity benchmark.  The analytical-rollout row uses
the analytical constrained policy for rollouts but estimates the adjoints
through the same \OLBPTT pipeline.  The error columns in both rows use the
analytical optimal-policy fields as targets.  The DPO-rollout row therefore
combines DPO-reference mismatch and numerical estimation, whereas the
analytical-rollout row isolates finite-grid, continuation, and projection
error.}
\label{tab:adjoint_level_validation}
\begin{tabularx}{\textwidth}{@{}l*{6}{>{\centering\arraybackslash}X}@{}}
\toprule
Estimator
& \shortstack{$\lambda^X$\\mean rel.}
& \shortstack{$\lambda^X$\\nRMSE}
& \shortstack{$P^{XX}$\\nRMSE}
& \shortstack{$\mathbf Z^X$\\nRMSE}
& \shortstack{mean $\|\widehat\zeta^X\|$\\$/\|\widehat Z^X\|$}
& \shortstack{Shift\\SNR} \\
\midrule
DPO-rollout \OLBPTT
& $0.4565\%$
& $0.6775\%$
& $0.7312\%$
& $55.8958\%$
& $0.529\%$
& $23.48$ \\
Analytical-rollout \OLBPTT
& $0.0139\%$
& $0.0312\%$
& $0.0347\%$
& $0.3264\%$
& $0.235\%$
& $1.18$ \\
\bottomrule
\end{tabularx}
\end{table}

The analytical-rollout row removes DPO-reference mismatch and therefore
isolates continuation, discretization, and projection error.  It validates all
three wealth-relevant objects: the first-adjoint level, the fixed-latent
curvature, and the Brownian coefficient.  The DPO-rollout $\mathbf Z^X$
discrepancy,
$55.8958\%$, is measured against the analytical optimal-policy coefficient and
is therefore not a pure projection-error statistic.  The DPO reference also
has a small but clearly resolved shifted component: its mean magnitude is about
$0.53\%$ of $\|\widehat Z^X\|$, with mean statewise SNR $23.5$.
Under the analytical rollout, by contrast, the estimated residual is comparable
to its Monte Carlo uncertainty.  This contrast shows that the estimator does
not impose a zero shift; it is not interpreted as a general monotone relation
between policy optimality and shift magnitude.  Equation
\eqref{eq:benchmark_fixed_latent_crra_identity} explains the close agreement
between the first- and second-level diagnostics but is not used to construct
$\widehat\zeta^X$.  The analytical homogeneous benchmarks for which a shift
ground truth is available have $\boldsymbol\zeta^{X,*}=\boldsymbol0$, but a
finite-budget DPO reference need not satisfy that identity; the DPO-rollout SNR
in Table~\ref{tab:adjoint_level_validation} illustrates this distinction.
Accordingly, these experiments validate that the reference-policy shift is
estimated rather than hard-coded and that the estimator approaches a known
zero target under an analytical rollout.  A benchmark with a genuinely
nonzero \emph{optimal} wealth-row shift remains outside the present numerical
validation.

\subsection{Does Local Recovery Improve the DPO Reference?}\label{sec:ssban}

Using the same terminal-wealth CRRA calibration, a two-layer width-200
LeakyReLU network with softplus
outputs enforces $u\ge0$.  Every reported result uses the intended orthant
$\mathcal K=\mathbb R_+^n$ with no cap on $\mathbf1^\top u$, so the
residual cash weight may be negative.  The analytical benchmark policy is the
orthant-constrained constant-opportunity CRRA policy characterized by
\citet{XuShreve1992a,XuShreve1992b}, namely the no-short-sale constrained
analogue of Merton's constant-opportunity allocation.
For each $n$, DPO, B-PGDPO, and QP-PGDPO in Table~\ref{tab:noshort_ablation} use the same Stage~1 reference,
held-out states, estimated adjoint tuple, and indexed Brownian banks.  Hence
DPO versus QP-PGDPO isolates the effect of replacing the finite-budget Stage~1
neural action by a conditional local solve, whereas QP-PGDPO versus
B-PGDPO isolates the finite central-path approximation.

\begin{table}[!htbp]
\centering
\scriptsize
\setlength{\tabcolsep}{3.5pt}
\caption{No-short-sale benchmark with borrowing allowed.  Policy RMSE is
measured against the analytical orthant-constrained constant-opportunity CRRA
policy characterized by Xu and Shreve; KKT--PG, gain, and QP gap use the full
estimated generalized Hamiltonian.}
\label{tab:noshort_ablation}
\begin{tabular}{rlrrrr}
\toprule
$n$ & Method
& Coordinatewise $u$-RMSE
& Mean KKT--PG
& $\mathbb H$ gain vs DPO
& $\mathbb H$ gap to QP \\
\midrule
10 & DPO
& $9.0657\times10^{-2}$
& $1.9777\times10^{-2}$
& $0$
& $4.260\times10^{-3}$ \\
10 & B-PGDPO
& $1.667\times10^{-3}$
& $5.810\times10^{-4}$
& $4.256\times10^{-3}$
& $4.77\times10^{-6}$ \\
10 & QP-PGDPO
& $1.99\times10^{-4}$
& $5.52\times10^{-18}$
& $4.260\times10^{-3}$
& $0$ \\
\addlinespace
100 & DPO
& $1.0852\times10^{-2}$
& $2.2022\times10^{-2}$
& $0$
& $3.895\times10^{-3}$ \\
100 & B-PGDPO
& $9.059\times10^{-3}$
& $3.571\times10^{-4}$
& $3.889\times10^{-3}$
& $6.26\times10^{-6}$ \\
100 & QP-PGDPO
& $2.86\times10^{-4}$
& $1.45\times10^{-17}$
& $3.895\times10^{-3}$
& $0$ \\
\bottomrule
\end{tabular}
\end{table}

DPO satisfies the pointwise constraints but leaves a substantial local
residual.  B-PGDPO reduces it
sharply and has positive pointwise gain at all 128 states in both dimensions;
at $n=100$ its mean QP gap is $6.26\times10^{-6}$.  QP-PGDPO solves the
estimated orthant block to numerical tolerance and remains within
$2.86\times10^{-4}$ coordinatewise RMSE of the analytical $n=100$ policy.
KKT--PG and Hamiltonian gap are the primary recovery diagnostics because the
full-shift input is estimated rather than analytically supplied.  The
$n=100$ B-PGDPO action RMSE can remain close to the DPO RMSE even when its
Hamiltonian gap to QP is small.  This is consistent with perturbed
complementarity: on an active branch,
$\Gamma_j(\mathbf u_{\varepsilon})=\varepsilon_{\rm bar}/\nu_{\varepsilon,j}$,
so a weakly active coordinate with a small multiplier can exhibit visible
Euclidean displacement at little local objective loss.  Policy distance and
barrier value gap are therefore complementary rather than interchangeable
diagnostics.

\FloatBarrier

\subsection{Finite-Budget Robustness}

The three-seed checkpoint analysis is reported in
Online Supporting Information~\ref{sec:training_budget} and
Table~\ref{tab:training_budget_multiseed}.  Most Stage~1 improvement occurs by
500 epochs, after which policy RMSE, KKT--PG, and the QP gap largely plateau.
At 1,000 epochs, QP-PGDPO has coordinatewise RMSE
$1.37\times10^{-4}$ and solves the estimated QP to numerical tolerance.
This is a finite-budget robustness diagnostic rather than a claim that no
other architecture or substantially larger training budget could improve the
DPO reference.
The decoder initialization diagnostic is reported in
Table~\ref{tab:initialization_diagnostic}.

\subsection{Does Recovery Handle Predictable Investment Opportunities?}
\label{sec:state-dependent}

We next let a scalar Ornstein--Uhlenbeck predictor drive affine excess returns
and correlate its shock with risky-return shocks.  This is an affine
predictable-return specification in the intertemporal portfolio tradition of
\citet{merton1973intertemporal,kim1996dynamic,liu2007portfolio}.  Its
unconstrained analytical structure overlaps the factor experiment in
\citet{huh2025breaking}; the distinct question here is whether the harvested
mixed and shifted adjoint inputs support \emph{constrained} local recovery.
The overlapping unconstrained analytical validation is therefore retained in
Online Supporting Information~\ref{app:unconstrained_predictable_return_validation} rather than
used as the principal evidence in this subsection.

The risky weights obey the fixed no-short/no-borrowing simplex
\[
 \mathcal K
 =\{\boldsymbol\pi\in\mathbb R^n:
       \boldsymbol\pi\ge0,\ \mathbf1^\top\boldsymbol\pi\le1\}.
\]
Thus the feasible fiber in this experiment does \emph{not} vary with time or
the predictor.  What varies with the factor $Y_t$ are the excess returns and
the resulting local-objective and adjoint inputs.  The separate experiment in
Section~\ref{sec:cons_cap} addresses a genuinely state-dependent feasible
fiber.

Conditional on the harvested adjoints, the risky-weight block is
\[
\max_{\boldsymbol\pi\in\mathcal K}
\left\{
\mathbf g(t,x,y)^\top\boldsymbol\pi
-\frac12\boldsymbol\pi^\top Q(t,x,y)\boldsymbol\pi
\right\},
\]
where
\[
\mathbf g
=x\left[
\widehat\lambda^X\boldsymbol\alpha(y)
+\boldsymbol\Sigma_{RY}(y)\widehat P^{YX}
+\boldsymbol v(y)(\widehat{\boldsymbol\zeta}^X)^\top
\right],
\qquad
Q=x^2[-\widehat P^{XX}]\boldsymbol\Sigma(y).
\]
Here $\boldsymbol\Sigma_{RY}$ is the risky-return/factor cross-covariance and
$\boldsymbol v$ is the risky-return diffusion loading in the Brownian basis
used for $\widehat{\boldsymbol\zeta}^X$.  The mixed block $P^{YX}$ is harvested
directly, while the shifted wealth row is estimated by
\eqref{eq:anchored_shift_regression}.  There is no closed-form constrained
policy target for this experiment.  Consequently QP-PGDPO is an exact solve
only of the local quadratic program induced by the supplied estimated
adjoints.  Its reported $10^{-17}$-scale KKT--PG residual is numerical zero---a
trace of floating-point solver and projection arithmetic---not a ground-truth
policy error of zero.

Each dimension in Table~\ref{tab:liu_ko_factor_bpgdpo} uses 128 held-out states and 8,192 continuations per state.
The $n=10,50$ checkpoints use 16 rollout steps and training batch size 1,024;
the $n=100$ checkpoint uses 12 steps and batch size 512.  B-PGDPO uses
$\varepsilon_{\rm bar}=3\times10^{-6}$, and all three rules use the same
Stage~1 checkpoint and harvested adjoint tuple.

\begin{table}[!htbp]
\centering
\scriptsize
\setlength{\tabcolsep}{2.7pt}
\caption{Constrained affine predictable-return recovery with the anchored
shifted input.  Residuals and gains are conditional on the common harvested
adjoint tuple; QP-PGDPO is not an analytical policy benchmark.}
\label{tab:liu_ko_factor_bpgdpo}
\begin{tabular}{rccccc rr}
\toprule
& \multicolumn{3}{c}{Mean KKT--PG residual}
& \multicolumn{2}{c}{$\mathbb H$ gain vs DPO}
& \multicolumn{2}{c}{Shift diagnostics} \\
\cmidrule(lr){2-4}\cmidrule(lr){5-6}\cmidrule(lr){7-8}
$n$
& DPO & B-PGDPO & QP-PGDPO
& B-PGDPO & QP-PGDPO
& $\|\widehat{\boldsymbol\zeta}^X\|/
   \|\widehat{\mathbf Z}^X\|$
& SNR \\
\midrule
10
& $1.6205\times10^{-2}$
& $2.61\times10^{-4}$
& $1.25\times10^{-17}$
& $5.482\times10^{-3}$
& $5.499\times10^{-3}$
& $0.319\%$ & $2.27$ \\
50
& $2.2639\times10^{-2}$
& $1.098\times10^{-3}$
& $3.76\times10^{-17}$
& $1.0977\times10^{-2}$
& $1.1096\times10^{-2}$
& $0.474\%$ & $2.21$ \\
100
& $2.4702\times10^{-2}$
& $2.006\times10^{-3}$
& $5.64\times10^{-17}$
& $1.5725\times10^{-2}$
& $1.5979\times10^{-2}$
& $0.308\%$ & $1.63$ \\
\bottomrule
\end{tabular}
\end{table}

QP-PGDPO has positive pointwise gain at every held-out state and solves the
estimated local QP to numerical tolerance.  B-PGDPO has positive mean gain at
all three dimensions, with a small finite-barrier gap along the central path.
These results
show that the recovery step remains effective when the local portfolio
coefficients and the mixed adjoint input depend on a predictor.  They do not
establish an analytical solution of the constrained factor problem.  The
independent-bank shift replication and the alternative frozen-feedback graph
diagnostics are reported in Online Supporting Information~\ref{app:lko_recovery_graph_audit}.

\FloatBarrier

\subsection{Does the Method Handle a State-Dependent Feasible Fiber?}\label{sec:cons_cap}

For $0\le C_t\le\bar mX_t$, the chart
\begin{equation}\label{eq:proportional_cap_chart}
 C_t=X_tc_t,\qquad c_t\in[0,\bar m]
\end{equation}
converts the moving bound to a fixed latent box; numerically, the lower endpoint
is $10^{-8}X_t$ rather than $0$ to avoid evaluating CRRA utility at zero.
The homothetic finite-horizon consumption--investment structure follows the
classical Merton framework \citep{merton1971optimum,karatzas1987optimal}.  The
specific proportional upper cap in \eqref{eq:proportional_cap_chart} is imposed
here directly through the scalar KKT condition; the analytical benchmark does
not rely on a separate closed-form theorem for this cap.  \OLBPTT holds $c_t$
fixed while differentiating the decoder, so $D_{X_t}^{\rm ol}C_t=c_t$.  At a
fixed state, consumption is the maximizer of the concave scalar KKT problem.  In the
reported CRRA calibration, $e^{\rho t}\lambda_t^X$ lies in the range of $U'$
and the maximizer has the clipping form
\begin{equation}\label{eq:proportional_cap_clipping}
 C_t^*=\min\!\left\{\max\!\left\{(U')^{-1}(e^{\rho t}\lambda_t^X),0\right\},
 \bar mX_t\right\}.
\end{equation}
Outside that range, the same scalar KKT problem selects the appropriate
endpoint rather than requiring an ordinary inverse of $U'$.

For the analytical benchmark, let $\tau:=T-t$.  The implementation uses the
same relative-discount (current-value) normalization as Stage~1, so define
\[
V^{\rm cv}(t,x)
:=e^{\rho t}V(t,x)
=a_M(\tau)\frac{x^{1-\gamma}}{1-\gamma}.
\]
After substituting the exact simplex-constrained risky allocation, $a_M$ is
obtained from the corresponding scalar CRRA coefficient ODE.  The analytical
actor uses that risky allocation and the proportional consumption rate
$\min\{\max\{a_M(\tau)^{-1/\gamma},0\},\bar m\}$.  In this subsection the
analytical benchmark fields use the same current-value normalization as the
numerical estimates:
\[
\lambda^{X,\mathrm{an}}=a_M(\tau)X^{-\gamma},\qquad
P^{XX,\mathrm{an}}=-\frac{\gamma}{X}\lambda^{X,\mathrm{an}},\qquad
\mathbf Z^{X,\mathrm{an}}
=P^{XX,\mathrm{an}}\boldsymbol\sigma_X^{\mathrm{an}},\qquad
\boldsymbol\zeta^{X,\mathrm{an}}=\boldsymbol0.
\]
The corresponding calendar-value fields are obtained by multiplying these
adjoint quantities by $e^{-\rho t}$.  Thus the scalar first-order condition can
be written equivalently as
$U'(C_t)=\lambda_t^{X,\rm cv}$ or
$e^{-\rho t}U'(C_t)=\lambda_t^X$; both give the clipping rule above.
The analytical-rollout row uses this actor only for rollout generation; the
same numerical \OLBPTT/projection/recovery pipeline estimates the adjoints,
shifted input, and recovered action.  The risky portfolio uses the full-shift
exact local QP.  The
comparison in Table~\ref{tab:consumption_cap_qp} uses 128 states and 8,192 continuations per state.

\begin{table}[!htbp]
\centering
\small
\setlength{\tabcolsep}{3pt}
\caption{Consumption-cap benchmark.  QP-PGDPO combines the full-shift exact
portfolio QP with scalar KKT clipping.  Panel C applies the same pipeline under
the analytical rollout policy to display the pipeline's finite-grid numerical
baseline.  Panel A reports RMSE against the analytical optimal controls,
whereas Panel B reports optimality residuals for the local problem formed from
the estimated adjoints.}
\label{tab:consumption_cap_qp}
\begin{tabular}{lcccc}
\toprule
\multicolumn{5}{c}{Panel A. RMSE against analytical KKT target} \\
\midrule
& \multicolumn{2}{c}{$u$-RMSE}
& \multicolumn{2}{c}{$C$-RMSE} \\
\cmidrule(lr){2-3}\cmidrule(lr){4-5}
$n$ & DPO & QP-PGDPO & DPO & QP-PGDPO \\
\midrule
2
& $5.0345\times10^{-2}$
& $2.318\times10^{-3}$
& $1.3108\times10^{-1}$
& $1.333\times10^{-3}$ \\
10
& $2.3940\times10^{-2}$
& $8.88\times10^{-4}$
& $1.2759\times10^{-1}$
& $1.285\times10^{-3}$ \\
100
& $8.450\times10^{-3}$
& $1.93\times10^{-4}$
& $1.2502\times10^{-1}$
& $1.795\times10^{-3}$ \\
\midrule
\multicolumn{5}{c}{Panel B. Local KKT--PG residuals} \\
\midrule
& \multicolumn{2}{c}{Portfolio KKT--PG}
& \multicolumn{2}{c}{Consumption KKT--PG} \\
\cmidrule(lr){2-3}\cmidrule(lr){4-5}
$n$ & DPO & QP-PGDPO & DPO & QP-PGDPO \\
\midrule
2
& $1.0673\times10^{-2}$
& $9.86\times10^{-18}$
& $4.3405\times10^{-1}$
& $6.55\times10^{-17}$ \\
10
& $8.292\times10^{-3}$
& $2.61\times10^{-17}$
& $4.1917\times10^{-1}$
& $2.31\times10^{-17}$ \\
100
& $1.0901\times10^{-2}$
& $1.07\times10^{-16}$
& $4.1299\times10^{-1}$
& $2.37\times10^{-17}$ \\
\midrule
\multicolumn{5}{c}{Panel C. Analytical-rollout pipeline at $n=100$} \\
\midrule
Estimator & $u$ RMSE & $C$ RMSE & Port. KKT--PG & Cons. KKT--PG \\
\midrule
Analytical-rollout pipeline
& $1.03\times10^{-4}$
& $2.788\times10^{-3}$
& $1.09\times10^{-16}$
& $1.14\times10^{-16}$ \\
\bottomrule
\end{tabular}
\end{table}

The full-shift portfolio recovery lowers RMSE against the analytical optimum
by more than an order of magnitude at every dimension.  At the same time it
solves the \emph{estimated-adjoint} KKT system to numerical tolerance.  These
are different diagnostics: QP-PGDPO is the conditional optimizer of that
estimated quadratic block, so its KKT residual can be essentially zero even
when adjoint and discretization errors leave a nonzero analytical-policy RMSE.
The displayed $10^{-17}$--$10^{-16}$ values are the floating-point residuals
of the QP solve and projected-gradient calculation, rather than substantive
KKT violations.
Consumption is particularly sensitive to the absolute first-adjoint level
through \eqref{eq:proportional_cap_clipping}.  At $n=100$, QP-PGDPO lowers its
RMSE from $1.250\times10^{-1}$ to $1.795\times10^{-3}$.  The
analytical-rollout pipeline gives $2.788\times10^{-3}$ because it still
estimates the adjoints on a finite grid with Monte Carlo continuation; it does
not insert the analytical adjoints.  The small reversal is therefore not
interpreted as DPO-rollout superiority.  Finite-barrier performance is
audited separately below.

\FloatBarrier

\subsection{External Neural Comparison and Computational Scope}
\label{sec:external_comparison}

The external comparison in Table~\ref{tab:hpin_external} asks a narrower question than a general method
ranking: does hard feasibility plus global value and partial differential
equation (PDE) residual training recover the same pointwise local KKT condition
as adjoint-to-control recovery?  The
adapted hPINN-style baseline trains a value network and a pointwise-feasible
control head
with HJB, terminal, and soft projected-KKT residuals.  Its simplex output
enforces the portfolio constraint directly, consistent with the fact that
hPINN and ICPINN constructions can encode hard constraints
\citep{lu2021physics,tan2025improved}.  It does not use a local QP/KKT decoder.

\begingroup
\begin{table}[!htbp]
\centering
\footnotesize
\setlength{\tabcolsep}{3pt}
\renewcommand{\arraystretch}{1.12}
\caption{Illustrative hPINN-style diagnostics at $n=100$.  The methods use the
same model calibrations, but their architectures, objectives, and tuning
budgets differ; the table is therefore not a general method ranking.  The
entry ``n/e'' denotes that B-PGDPO was not evaluated in the archived
supplementary affine-factor run.}
\label{tab:hpin_external}
\begin{tabularx}{\textwidth}{@{}>{\raggedright\arraybackslash}X
*{4}{>{\centering\arraybackslash}X}@{}}
\toprule
\multicolumn{5}{@{}l}{\textbf{Panel A. Mean KKT--PG residuals}} \\
\midrule
Benchmark & DPO & hPINN & B-PGDPO & QP-PGDPO \\
\midrule
Constrained affine predictable-return & $2.470\times10^{-2}$ & $3.78\times10^{-2}$ & $2.006\times10^{-3}$ & $5.64\times10^{-17}$ \\
Supplementary affine factor $(d_Y=3)$ & $4.233\times10^{-2}$ & $1.17\times10^{-1}$ & \text{n/e} & $1.53\times10^{-16}$ \\
\midrule
\multicolumn{5}{@{}l}{\textbf{Panel B. Auxiliary diagnostics}} \\
\midrule
Constrained affine predictable-return
& \multicolumn{2}{c}{hPINN gain $-8.63\times10^{-3}$}
& \multicolumn{2}{c}{QP-PGDPO gain $1.598\times10^{-2}$} \\
Supplementary affine factor $(d_Y=3)$
& \multicolumn{2}{c}{hPINN time $74.9$ min}
& \multicolumn{2}{c}{DPO+\OLBPTT+shift+QP time $5.02$ min} \\
\bottomrule
\end{tabularx}
\end{table}
\endgroup

In the constrained affine predictable-return diagnostic, the pointwise-feasible
hPINN head has mean KKT--PG
$3.78\times10^{-2}$ and negative Hamiltonian gain relative to DPO, whereas
QP-PGDPO solves the supplied-adjoint local QP to numerical tolerance and has
positive gain.  The affine-factor timing row is hardware- and
implementation-specific; it does not establish uniform speed superiority over
PINN methods.  It only records that this research implementation obtains the
reported local KKT accuracy without constructing a value-function grid.  The
architecture and timing decomposition are in
Online Supporting Information~\ref{app:external_comparison_details}
(with timings reported in Table~\ref{tab:timing_profile}).

Deep BSDE methods approximate a BSDE representation of a value or marginal
value and its Brownian integrand along simulated trajectories
\citep{han2018solving,e2017deep}.  They are not intrinsically unable to handle
pointwise control constraints: feasibility can be introduced through a
constrained control parameterization, while reflected or variational
formulations address different obstacle-type structures.  We do not report a
numerical Deep BSDE row because no directly matched constrained implementation
is used here.  A fair comparison would have to specify the constraint-handling
adaptation and match terminal data, time discretization, training and held-out
states, and computational budget on a low-dimensional common benchmark.  Its
omission is therefore not evidence of numerical inferiority.  The
stochastic-maximum-principle learning method of \citet{ji2022solving} is closer
conceptually, but the present architecture is distinguished by using a
reference rollout for fixed-latent first/second-adjoint harvesting and a
separate local KKT/QP or barrier recovery step.

Finally, the $n=100$ rows enlarge the constrained action vector while the state
remains low dimensional ($d_Y\le5$ in the timing audit).  The scalability
claim therefore concerns local control-block recovery, not high-dimensional
state-space PDE or BSDE approximation.

\FloatBarrier

\subsection{Barrier and Switching Audits}

Detailed central-path, barrier-sensitivity, and active-set-transition audits
are reported in Online Supporting Information~\ref{app:barrier_sensitivity}
and Tables~\ref{tab:barrier_final_audit}, \ref{tab:switching_region_audit},
and~\ref{tab:barrier_eps_audit}.
The split-wise audit checks numerical convergence and finite-barrier bias
against the exact QP, whereas the region-stratified audit asks whether the gain
is concentrated near active-set transitions.  At the common
$\varepsilon_{\rm bar}=3\times10^{-6}$ specification, every evaluated cap-model
state satisfies the stopping rule and has positive unregularized gain.  Mean
gaps to the exact local QP are about $2.22\times10^{-4}$ in the
constant-opportunity cap and $2.58$--$2.83\times10^{-4}$ in the affine-factor
cap; the near-switching cells retain positive pointwise gain.  These checks
validate the finite-barrier central path against the exact conditional QP,
not a distinct nonquadratic benchmark.

\section{Conclusion}\label{sec:conclusion}

We developed a solver-neutral adjoint-to-control framework for continuous-time
portfolio choice under smooth pointwise constraints.  A pointwise-feasible
Stage~1 DPO actor generates the default reference rollout; after detachment,
fixed-latent \OLBPTT harvests its chart-reduced adjoints and a local generalized
Hamiltonian recovers the constrained action.  The primitive formulation remains open loop and does not require an HJB solution or a Markov
admissible class.

The second-order adjoint supplies the curvature needed because risky weights
enter the diffusion.  QP-PGDPO solves the resulting quadratic-affine block
exactly conditional on the estimated adjoints, while B-PGDPO follows a smooth
finite-barrier approximation to the same local KKT problem.  The end-to-end
stability theorem converts local reference value loss and numerical
adjoint/QP errors into recovered-policy and local QP-gap bounds.  Analytical-rollout
constant- and predictable-opportunity benchmarks validate the full numerical
adjoint-acquisition pipeline.  Common-input tests show that recovery reduces the
generalized-Hamiltonian/KKT residual left by finite-budget DPO across training
budgets, predictor-dependent hedging inputs, a wealth-dependent consumption cap,
and up to 100 risky assets.  External neural and barrier audits delimit the
comparison and solver scope.

The conclusions remain local and conditional: charts cover regular patches,
QP exactness depends on the supplied adjoints, and the experiments establish
scalability in the action dimension rather than dimension-free state-space
complexity.  Nonsmooth fibers, nonregular switching geometry, pure state
constraints, and genuinely nonquadratic recovery remain for future work.

\section*{Funding}

This work was supported by National Research Foundation of Korea (NRF) grants
funded by the Korea government through the Ministry of Science and ICT (MSIT)
(RS-2025-00562904 and RS-2024-00355646).

\section*{Conflict of Interest}

The authors declare that they have no conflicts of interest.

\section*{Data Availability Statement}

No proprietary or third-party data were used in this study.  All numerical
data were generated synthetically from the models and parameter specifications
described in the article and its Supporting Information.  The source code,
configuration files, and validation scripts accompanying this study are
publicly available in the fixed GitHub release
\texttt{mf-revision-2026-08-29} at
\url{https://github.com/huhjeonggyu/pgdpo-adjoint-control-recovery/releases/tag/mf-revision-2026-08-29}.
The repository documents the scope of self-contained reruns and any additional
assets required for historical paper-scale reruns.

\section*{Ethics Statement}

Ethics approval and consent to participate are not applicable because this
study did not involve human participants, animals, or identifiable personal
data.

\clearpage
\begin{center}
{\Large\bfseries Published Appendix}\par
\medskip
{\small Core mathematical proofs: fixed-latent OL-BPTT correspondence,
reference-policy stability, and estimated-input QP recovery.}
\end{center}
\medskip
\appendix
\numberwithin{equation}{section}
\counterwithin{table}{section}
\counterwithin{figure}{section}
\counterwithin{algorithm}{section}
\section{Notation and Coordinate Conventions}
\label{app:notation_summary}

The state is $\mathbf S=(X,Y)$ and the reduced control is
$\mathbf u=(\boldsymbol\pi^\top,C)^\top$, with
$\pi_0=1-\mathbf1^\top\boldsymbol\pi$.  Superscripts $*$, $\star$,
${\rm ref}$, and ${\rm an}$ denote, respectively, an optimal/target object, a
selected local KKT solution, a frozen-reference object, and an analytical
benchmark.  Tildes and hats denote raw pathwise quantities and numerical
estimates; barred coefficients are chart-composed and differentiated at fixed
latent coordinate.

For $d\ge1$, $\mathbb S_{\rm sym}^d$ is the space of real symmetric
matrices, distinct from the process space $\mathbb S^2$ below.  We use
$\langle A,B\rangle_F=\operatorname{tr}(A^\top B)$ and $\|\cdot\|_F$,
$\DOL$ for fixed-latent differentiation, $\sg$ for stop-gradient,
$\Pi_A$ and $N_A(x)$ for Euclidean projection and the normal cone, and the
standard meanings of $\operatorname{diag}$, $\operatorname{dist}$,
$\lambda_{\max}$, $I$, and $\mathbb I_A$.

For an adapted vector- or matrix-valued process $\chi$ and a
Brownian-integrand process $\mathsf Z$, define
\[
\|\chi\|_{\mathbb S^2}
:=
\left(\mathbb E\sup_{0\le t\le T}\|\chi_t\|^2\right)^{1/2},
\qquad
\|\mathsf Z\|_{\mathbb H^2}
:=
\left(\mathbb E\int_0^T\|\mathsf Z_t\|^2dt\right)^{1/2},
\]
using the Frobenius norm for matrix-valued processes.  For
$1\le p,q\le\infty$, define
\begin{equation}\label{eq:mixed_norm_definition}
\|f\|_{L^{q,p}}
:=
\left(
\int_0^T
\left(\mathbb E\|f_t\|^p\right)^{q/p}dt
\right)^{1/q},
\end{equation}
with the usual essential-supremum modifications.

Tables~\ref{tab:core_symbols} and \ref{tab:discrete_symbols} list the
recurring notation; proof-specific symbols and recovery-error parameters are
defined where introduced.

\begin{table}[!htbp]
\centering
\footnotesize
\renewcommand{\arraystretch}{1.08}
\setlength{\tabcolsep}{4.5pt}
\caption{Model, chart, Hamiltonian, adjoint, and recovery notation.}
\label{tab:core_symbols}
\begin{tabularx}{\textwidth}{@{}p{4.4cm}X@{}}
\toprule
\textbf{Symbol} & \textbf{Meaning} \\
\midrule

$\mathbf S=(X,Y)$, $d_S=1+d_Y$
&
Wealth--factor state and its dimension.  The factor-free specialization has
$d_Y=0$ and $\mathbf S=X$.
\\

$\mathbf W$, $d_W$
&
Brownian motion and the number of Brownian risk drivers.
\\

$\boldsymbol\alpha$, $\boldsymbol v$,
$\boldsymbol\Sigma=\boldsymbol v\boldsymbol v^\top$
&
Risky-asset excess-return vector, return loading, and risky-return covariance
matrix.
\\

$\boldsymbol\Sigma_{RY}
=\boldsymbol v\boldsymbol\sigma_Y^\top$
&
Instantaneous return--factor cross-covariance matrix.
\\

$\boldsymbol b_S$, $\boldsymbol\sigma_S$
&
Joint state drift and
$d_S\times d_W$ state diffusion matrix.
\\

$\boldsymbol\pi$, $\pi_0$, $C$,
$\mathbf u=(\boldsymbol\pi^\top,C)^\top$
&
Risky-asset weights, residual risk-free weight, consumption rate, and the
reduced control vector; $\pi_0=1-\mathbf1^\top\boldsymbol\pi$.
\\

$\Gamma(t,\mathbf s;\mathbf u)$,
$\mathcal U(t,\mathbf s)$, $\nu$
&
Pointwise inequality map, feasible control fiber, and inequality-multiplier
vector.  The convention is $\Gamma\ge\boldsymbol0$.
\\

$\mathbf a\in\mathfrak A$, $\Psi$
&
Nonredundant latent coordinate, its state-independent local domain, and the
state-dependent feasible-fiber chart
$\mathbf u=\Psi(t,\mathbf s,\mathbf a)$.
\\

$\Theta=(\theta,\phi)$, $\varpi$,
$J$, $J_{\rm train}$
&
Policy parameters, initial time--state sampling law, the control objective,
and the initial-law-averaged policy-training objective.
\\

$\ell$, $\mathcal H$, $\bar{\mathcal H}$, $\mathbb H$
&
Running reward, primitive first-order Hamiltonian, fixed-latent
chart-reduced Hamiltonian, and second-order generalized Hamiltonian.
\\

$\lambda\in\mathbb R^{d_S}$,
$\mathbf Z\in\mathbb R^{d_S\times d_W}$
&
First-order adjoint and its Brownian martingale coefficient.
\\

$P\in\mathbb S_{\rm sym}^{d_S}$,
$\mathbf R=(R^1,\ldots,R^{d_W})$
&
Matrix second-order adjoint and its matrix-valued Brownian martingale
coefficients, with $R^\ell\in\mathbb S_{\rm sym}^{d_S}$.
\\

$\boldsymbol\zeta
=\mathbf Z-P\boldsymbol\sigma_S$
&
Shifted martingale adjoint used in the generalized Hamiltonian.
\\

\shortstack[l]{
$e_X$; $\lambda^X,\mathbf Z^X,\boldsymbol\zeta^X$\\
$P^{XX},P^{XY},P^{YX}$}
&
$e_X=(1,0,\ldots,0)^\top$ selects the wealth component or wealth row:
$\lambda^X=e_X^\top\lambda$,
$\mathbf Z^X=e_X^\top\mathbf Z$,
$\boldsymbol\zeta^X=e_X^\top\boldsymbol\zeta$, and
$P^{XX}=e_X^\top Pe_X$.
The mixed blocks satisfy $P^{YX}=(P^{XY})^\top$.
\\

$\vartheta=(\lambda,\boldsymbol\zeta,P)$,
$\widehat\vartheta$
&
Exact stagewise adjoint input and its numerical estimate.  In the portfolio
specialization, only the wealth-relevant components may be passed to the
solver.
\\

$\mathcal R(t,\mathbf s;\vartheta)$,
$\mathbf u^\star(\vartheta)$
&
Local set-valued recovery correspondence on the selected feasible branch and
a selected local KKT/recovery solution.
\\

$\mathbf g(t,\mathbf s;\vartheta)$,
$Q(t,\mathbf s;\vartheta)$,
$\mathcal K(t,\mathbf s)$
&
Linear coefficient, quadratic matrix, and convex risky-weight feasible set of
the local QP.  Under $P^{XX}<0$ and
$\boldsymbol\Sigma\succ0$, $Q\succ0$.
\\

\bottomrule
\end{tabularx}
\end{table}

\begin{table}[!htbp]
\centering
\footnotesize
\renewcommand{\arraystretch}{1.08}
\setlength{\tabcolsep}{4.5pt}
\caption{Discrete OL-BPTT sensitivities, projections, and shifted-input
estimation.}
\label{tab:discrete_symbols}
\begin{tabularx}{\textwidth}{@{}p{4.4cm}X@{}}
\toprule
\textbf{Symbol} & \textbf{Meaning} \\
\midrule

$h$, $t_k$, $\Delta\mathbf W_k$
&
Uniform time step, grid point, and Brownian increment over
$[t_k,t_{k+1}]$.
\\

$\sg$, $\bar{\mathbf a}_k$,
$\mathbf u_k^{\rm ref}$,
$\DOL_{\mathbf S_k}$
&
Stop-gradient operator, detached latent output, corresponding reference
control, and fixed-latent open-loop state derivative.
\\

$\mathcal M_\ell^{(m)}$, $\mathfrak r_k$,
$\mathfrak g_N$, $\mathcal J_k$
&
Activity mask for rollout $m$, one-step reward, terminal reward, and
continuation payoff.
\\

$\widetilde\lambda_k$, $\widetilde P_k$
&
Raw pathwise gradient and Hessian of the fixed-latent continuation payoff.
They depend on future simulated shocks.
\\

$A_k$, $B_k^j$
&
Fixed-latent state-transition Jacobian and Hessian of the $j$th component of
the next state.
\\

$\mathbb E_k[\cdot]$, $\lambda_k$, $P_k$
&
Conditional expectation given $\mathcal F_{t_k}$ and the population adapted
projections
$\lambda_k=\mathbb E_k[\widetilde\lambda_k]$,
$P_k=\mathbb E_k[\widetilde P_k]$.
\\

$\widehat\lambda_k$, $\widehat P_k$
&
Monte Carlo or regression approximations of the population projections
$\lambda_k$ and $P_k$.
\\

$\mathbf Z_k^h$, $R_k^{h,\ell}$
&
Intervalwise first-chaos projection coefficients associated with
$[t_k,t_{k+1}]$ for the first- and second-adjoint processes.
\\

$\xi_{k+1}$, $\xi_{k+1}^P$
&
One-step components with zero conditional mean that are orthogonal to the
span of the Brownian increment over the step.  They are retained explicitly in
the discrete adjoint recursions.
\\

$\mathfrak D_k^{(2),{\rm ol}}$,
$\mathfrak e_k^{(1)}$, $\mathfrak e_k^{(2)}$,
$\mathcal E_h$
&
Discrete second-adjoint driver, first- and second-order remainder arrays, and
their aggregate maximal-error criterion.
\\

$M_{\rm out}$, $M_{\rm in}$, $D_k$, $\kappa_Z$
&
Number of signed outer Brownian paths (an even number, forming
$M_{\rm out}/2$ independent antithetic pairs), number of inner
common-random-number continuations per branch, Brownian design matrix, and
ridge parameter in the shifted-input regression.
\\

$c_k^{X,{\rm anc}}$
&
Columnized wealth-row anchor,
\[
c_k^{X,{\rm anc}}
=
\left[
e_X^\top\widehat P_k
\boldsymbol\sigma_S
(t_k,\mathbf S_k;\mathbf u_k^{\rm ref})
\right]^\top.
\]
\\

$\widehat{\boldsymbol\zeta}_k^X$,
$\widehat{\mathbf Z}_k^X$
&
Estimated shifted wealth row and, when required, the reconstructed wealth-row
Brownian coefficient.
\\

\bottomrule
\end{tabularx}
\end{table}

Uniform chart and recovery statements are understood on the compact
localizations and regular KKT branches stated in the corresponding results.

\section{Technical Results for Fixed-Latent \OLBPTT}
\label{app:proof_bptt_pmp}

This published appendix section states the fixed-latent one-step expansion used in
Section~\ref{sec:olbptt_pmp_correspondence} and proves
Theorem~\ref{thm:bptt_pmp_constrained}.  The proof has two layers.  The
algebraic layer consists of the one-step expansion and the exact first- and
second-order pathwise chain rules.  The probabilistic layer consists of
conditional projection, control of the orthogonal projection residuals and
higher-order remainders, identification of the discrete adjoint drivers, and
stability of the limiting vector/matrix BSDE system.

\subsection{Fixed-Latent Setup and Variable Initial Times}
\label{app:fixed_latent_setup}

Work on the common compact chart patch from
Assumptions~\ref{ass:chart_pmp_regularity} and
\ref{ass:bptt_pmp_reg}.  On the uniform grid
\[
0=t_0<t_1<\cdots<t_N=T,
\qquad
h:=T/N,
\]
the frozen actor generates
\[
\bar{\mathbf a}_k
:=
\sg\!\left(
\mathbf a_\Theta(t_k,\mathbf S_k)
\right),
\qquad
\mathbf u_k^{\rm ref}
:=
\Psi(t_k,\mathbf S_k,\bar{\mathbf a}_k).
\]
For a state derivative initiated at time $t_k$, all later latent values
$\{\bar{\mathbf a}_\ell\}_{\ell\ge k}$ are held fixed.  Thus derivatives
through the actor output vanish, while the state derivatives of the chart
remain.

The experiments use smooth feasible output maps.  For example, full portfolio
weights under no borrowing and no short sales can be generated by
\[
\pi_i
=
\frac{\exp(o_i/\tau_{\rm sm})}
{\sum_{j=0}^n\exp(o_j/\tau_{\rm sm})},
\qquad i=0,\ldots,n,
\]
where $o_0,\ldots,o_n$ are network logits and $\tau_{\rm sm}>0$ is the
softmax temperature.  A gauge-fixed, nonredundant latent coordinate is used
for the geometric chart.  When borrowing is allowed but short sales are not, one may instead set
\[
\pi_i=\operatorname{softplus}(\widetilde o_i),
\quad i=1,\ldots,n,
\qquad
\pi_0=1-\sum_{i=1}^n\pi_i.
\]
Here $\operatorname{softplus}(z):=\log(1+e^z)$.  A smooth consumption box is
obtained from
\[
C
=
C_{\min}
+(C_{\max}-C_{\min})\frac{1+\tanh o_C}{2}.
\]
Here $o_C$ is the consumption logit.  These are implementation maps for
feasible reference generation; none is
assumed to be a global chart of the whole feasible graph.

On the portfolio grid, the wealth component is advanced by exponential--Euler
and the factor component by a model-appropriate transition:
\begin{equation}\label{eq:discrete_state_transition}
\begin{aligned}
X_{k+1}
={}&
X_k\exp\!\Bigg(
\left[
r_k+\boldsymbol\pi_k^\top\boldsymbol\alpha_k
-\frac12\boldsymbol\pi_k^\top\boldsymbol\Sigma_k\boldsymbol\pi_k
-\frac{C_k}{X_k}
\right]h
+\boldsymbol\pi_k^\top\boldsymbol v_k\Delta\mathbf W_k
\Bigg),
\\
Y_{k+1}
={}&
\Phi_Y(t_k,Y_k,\Delta\mathbf W_k).
\end{aligned}
\end{equation}
Here the coefficient subscript denotes evaluation at $(t_k,Y_k)$.
The map $\Phi_Y$ may be Euler--Maruyama or an exact transition when one is
available.  Together these updates define the fixed-latent transition
\[
\mathbf S_{k+1}
=
\Phi_h(t_k,\mathbf S_k,\mathbf u_k^{\rm ref},\Delta\mathbf W_k).
\]
Define its first and second state derivatives by
\[
A_k
:=
D_{\mathbf S_k}^{\rm ol}\mathbf S_{k+1},
\qquad
B_k^j
:=
D_{\mathbf S_k\mathbf S_k}^{2,\rm ol}S_{k+1}^j,
\qquad
j=1,\ldots,d_S.
\]
For completeness, consider a rollout $m$ with sampled initial index
$k_{\rm in}^{(m)}$.  Define the activity mask
\[
\mathcal M_\ell^{(m)}
:=
\mathbb I_{\{\ell\ge k_{\rm in}^{(m)}\}},
\]
the calendar-discounted one-step payoff
\[
\mathfrak r_\ell^{(m)}
:=
\mathcal M_\ell^{(m)}
e^{-\rho t_\ell}
U(C_\ell^{(m)})h,
\]
and the terminal payoff
\[
\mathfrak g_N^{(m)}
:=
K e^{-\rho T}U(X_N^{(m)}).
\]
For $k\ge k_{\rm in}^{(m)}$, the continuation payoff is
\begin{equation}\label{eq:appendix_masked_continuation}
\mathcal J_k^{(m)}
=
\sum_{\ell=k}^{N-1}
\mathfrak r_\ell^{(m)}
+
\mathfrak g_N^{(m)}.
\end{equation}
Leading entries before $k_{\rm in}^{(m)}$ are padding only and do not
contribute to the objective or its derivatives.  We suppress the rollout
index below.

The reported implementation restarts each diagnostic rollout with discounts
relative to its initial time.  For a rollout beginning at $t_{\rm in}$,
\[
\mathcal J_k^{\rm rel}
=
e^{\rho t_{\rm in}}\mathcal J_k,
\qquad
\widetilde\lambda_k^{\rm rel}
=
e^{\rho t_{\rm in}}\widetilde\lambda_k,
\qquad
\widetilde P_k^{\rm rel}
=
e^{\rho t_{\rm in}}\widetilde P_k.
\]
Thus relative-discount estimates are current-value-normalized versions of the
calendar adjoints at the restarted initial state; multiplying by
$e^{-\rho t_{\rm in}}$ converts them back to the calendar normalization used in
the theory.  The deterministic positive scaling does not change the local
portfolio maximizer, while the consumption first-order condition must be
written in the matching normalization, as in Section~\ref{sec:cons_cap}.

At the terminal index,
\[
\widetilde\lambda_N
=
K e^{-\rho T}U'(X_N)e_X,
\qquad
\widetilde P_N
=
K e^{-\rho T}U''(X_N)e_Xe_X^\top.
\]

Writing $\mathcal J_k=\mathfrak r_k+\mathcal J_{k+1}$, the exact
multivariate chain rules are
\begin{align}
\widetilde\lambda_k
&=
D_{\mathbf S_k}^{\rm ol}\mathfrak r_k
+A_k^\top\widetilde\lambda_{k+1},
\label{eq:bptt_chain_rule}
\\
\widetilde P_k
&=
D^2_{\mathbf S_k\mathbf S_k}\mathfrak r_k
+A_k^\top\widetilde P_{k+1}A_k
+\sum_{j=1}^{d_S}\widetilde\lambda_{k+1}^jB_k^j.
\label{eq:bptt_second_order}
\end{align}

\subsection{Fixed-Latent One-Step Expansion}
\label{app:fixed_latent_one_step}

\begin{applemma}[Fixed-latent one-step expansion]
\label{lem:fixed_latent_one_step_expansion}
Suppose Assumptions~\ref{ass:chart_pmp_regularity} and
\ref{ass:bptt_pmp_reg}(i)--(ii) hold.  Consider the transition scheme
\eqref{eq:discrete_state_transition}: exponential--Euler for wealth and either
Euler--Maruyama or the exact OU transition for the factor block.  For each
fixed-latent one-step object
\[
\mathcal X_k
\in
\left\{
\mathbf S_{k+1}-\mathbf S_k,\,
A_k-I,\,
B_k^1,\ldots,B_k^{d_S}
\right\},
\]
there exist uniformly bounded
$\mathcal F_{t_k}$-measurable coefficient arrays
$\mu_k$, $\Sigma_k^\ell$, and $C_k^{\ell m}$ such that
\begin{equation}\label{eq:appendix_one_step_decomposition}
\begin{aligned}
\mathcal X_k
={}&
\mu_k h
+
\sum_{\ell=1}^{d_W}
\Sigma_k^\ell\Delta W_k^\ell
+
\chi_{k+1}
+
\operatorname{Rem}_{k+1},
\\
\chi_{k+1}
:={}&
\sum_{\ell,m=1}^{d_W}
C_k^{\ell m}
\left(
\Delta W_k^\ell\Delta W_k^m
-
\delta_{\ell m}h
\right).
\end{aligned}
\end{equation}
The centered second-chaos term satisfies
\begin{equation}\label{eq:appendix_second_chaos_bounds}
\mathbb E_k[\chi_{k+1}]
=
0,
\qquad
\mathbb E_k[
\chi_{k+1}\Delta\mathbf W_k^\top
]
=
0,
\qquad
\mathbb E_k
\|\chi_{k+1}\|^2
\le
Ch^2.
\end{equation}
The higher-order remainder satisfies
\begin{equation}\label{eq:appendix_one_step_remainder_bounds}
\left\|
\mathbb E_k[
\operatorname{Rem}_{k+1}
]
\right\|
\le
Ch^2,
\qquad
\mathbb E_k
\|
\operatorname{Rem}_{k+1}
\|^2
\le
Ch^3.
\end{equation}
The bounds hold uniformly over the common compact chart patch, with the
Euclidean or Frobenius norm according to the object under consideration.
\end{applemma}

\begin{proof}
For the exponential--Euler wealth update in
\eqref{eq:discrete_state_transition}, write
\[
X_{k+1}
=
X_k e^{q_k},
\]
where
\begin{equation}\label{eq:appendix_wealth_log_increment}
q_k
=
\left[
r_k
+
\boldsymbol\pi_k^\top\boldsymbol\alpha_k
-
\frac12
\boldsymbol\pi_k^\top
\boldsymbol\Sigma_k
\boldsymbol\pi_k
-
\frac{C_k}{X_k}
\right]h
+
\boldsymbol\pi_k^\top
\boldsymbol v_k
\Delta\mathbf W_k.
\end{equation}
Here the risky weights and consumption are evaluated through the
fixed-latent chart.  In particular, their state derivatives retain the
dependence of $\Psi$ on $\mathbf S_k$, while
$\bar{\mathbf a}_k$ is held fixed.

The Taylor expansion
\[
e^{q_k}
=
1+q_k+\frac12q_k^2+\operatorname{Rem}(q_k)
\]
separates the predictable order-$h$ term, the first Brownian chaos, and the
centered quadratic term
\begin{equation}\label{eq:appendix_centered_wealth_chaos}
\left(
\boldsymbol\pi_k^\top
\boldsymbol v_k
\Delta\mathbf W_k
\right)^2
-
\boldsymbol\pi_k^\top
\boldsymbol\Sigma_k
\boldsymbol\pi_k\,h.
\end{equation}
The term in
\eqref{eq:appendix_centered_wealth_chaos}
has conditional mean zero, is orthogonal to each component of
$\Delta\mathbf W_k$, and has conditional second moment $O(h^2)$.

Products of an order-$h$ drift term with a first-chaos term, together with
the terms of order three and higher in the exponential expansion, have
conditional mean $O(h^2)$ and conditional second moment $O(h^3)$.  Uniform
bounds follow from the compact chart localization and the bounded derivatives
of the chart-composed coefficients.

The Euler--Maruyama factor update is affine in
$\Delta\mathbf W_k$.  The exact Ornstein--Uhlenbeck (OU) transition used in
the corresponding experiments is also a predictable term plus a linear Gaussian innovation and
admits the same first-chaos representation, with higher-order coefficient
corrections absorbed into the remainder.  Hence the factor transition
satisfies
\eqref{eq:appendix_one_step_decomposition}--
\eqref{eq:appendix_one_step_remainder_bounds}.

Finally, differentiate these expansions once and twice with respect to
$\mathbf S_k$, holding $\bar{\mathbf a}_k$ fixed.  Assumption
\ref{ass:chart_pmp_regularity} gives uniformly bounded derivatives of the
chart and chart-composed coefficients through order two.  The same chaos
decomposition and moment orders therefore hold for
$A_k-I$ and each $B_k^j$.  This proves the lemma.
\end{proof}

\subsection{Martingale Projections and Orthogonal Residuals}
\label{app:martingale_projection_residuals}

Under the augmented Brownian filtration, martingale representation gives
\begin{equation}\label{eq:appendix_first_interval_density}
\lambda_{k+1}
-
\mathbb E_k[\lambda_{k+1}]
=
\int_{t_k}^{t_{k+1}}
\mathcal Z_s^{\lambda,k+1}\,d\mathbf W_s
\end{equation}
for a matrix-valued interval density
$\mathcal Z^{\lambda,k+1}$, and
\begin{equation}\label{eq:appendix_second_interval_density}
P_{k+1}
-
\mathbb E_k[P_{k+1}]
=
\sum_{\ell=1}^{d_W}
\int_{t_k}^{t_{k+1}}
\mathcal Z_s^{P,k+1,\ell}\,dW_s^\ell
\end{equation}
for matrix-valued densities
$\mathcal Z^{P,k+1,\ell}$.

Their best
$\mathcal F_{t_k}$-measurable, intervalwise constant projections are
\begin{equation}\label{eq:discrete_martingale_coefficients}
\begin{aligned}
\mathbf Z_k^h
&=
\frac1h
\mathbb E_k\!\left[
\int_{t_k}^{t_{k+1}}
\mathcal Z_s^{\lambda,k+1}\,ds
\right]
=
\frac1h
\mathbb E_k\!\left[
\lambda_{k+1}\Delta\mathbf W_k^\top
\right],
\\
R_k^{h,\ell}
&=
\frac1h
\mathbb E_k\!\left[
\int_{t_k}^{t_{k+1}}
\mathcal Z_s^{P,k+1,\ell}\,ds
\right]
=
\frac1h
\mathbb E_k\!\left[
P_{k+1}\Delta W_k^\ell
\right].
\end{aligned}
\end{equation}

These are the standard intervalwise Brownian projections used in
time-discretization schemes for BSDE martingale integrands
\citep{bouchard2004discrete,zhang2004numerical}.  Here they are applied
jointly to the vector and matrix adjoints, and the components orthogonal to
the current Brownian increment are retained explicitly rather than discarded.

The orthogonal residuals are therefore
\begin{equation}\label{eq:appendix_first_projection_residual}
\xi_{k+1}
=
\int_{t_k}^{t_{k+1}}
\left(
\mathcal Z_s^{\lambda,k+1}
-
\mathbf Z_k^h
\right)d\mathbf W_s
\end{equation}
and
\begin{equation}\label{eq:appendix_second_projection_residual}
\xi_{k+1}^P
=
\sum_{\ell=1}^{d_W}
\int_{t_k}^{t_{k+1}}
\left(
\mathcal Z_s^{P,k+1,\ell}
-
R_k^{h,\ell}
\right)dW_s^\ell.
\end{equation}
They satisfy
\[
\mathbb E_k[\xi_{k+1}]
=
0,
\qquad
\mathbb E_k[
\xi_{k+1}\Delta\mathbf W_k^\top
]
=
0,
\]
with the analogous componentwise identities for $\xi_{k+1}^P$.

By It\^o isometry and
Assumption~\ref{ass:bptt_pmp_reg}(iv),
\begin{equation}\label{eq:appendix_projection_energy}
\sum_{k=k_{\rm in}}^{N-1}
\mathbb E
\left[
\|\xi_{k+1}\|^2
+
\|\xi_{k+1}^P\|_F^2
\right]
\le
C_{\rm proj}h.
\end{equation}
Since the residuals form martingale-difference arrays, Doob's inequality gives
\begin{equation}\label{eq:appendix_residual_maximal_bound}
\begin{aligned}
&
\mathbb E
\max_{k_{\rm in}\le m\le N-1}
\left\|
\sum_{k=k_{\rm in}}^m
\xi_{k+1}
\right\|^2
\\
&\quad
+
\mathbb E
\max_{k_{\rm in}\le m\le N-1}
\left\|
\sum_{k=k_{\rm in}}^m
\xi_{k+1}^P
\right\|_F^2
\longrightarrow0.
\end{aligned}
\end{equation}

A centered second-chaos term is orthogonal to constants and first Brownian
increments, but it need not be orthogonal to the residual
$\xi_{k+1}$.  This interaction must therefore be controlled rather than set
equal to zero.  By conditional Cauchy--Schwarz and
\eqref{eq:appendix_second_chaos_bounds},
\begin{equation}\label{eq:appendix_residual_chaos_bound}
\mathbb E
\left\|
\mathbb E_k[
\chi_{k+1}\xi_{k+1}
]
\right\|^2
\le
Ch^2
\mathbb E\|\xi_{k+1}\|^2,
\end{equation}
where products are understood componentwise or through the compatible
matrix contraction arising in the adjoint recursion.  Consequently,
\begin{equation}\label{eq:appendix_aggregate_residual_chaos}
\sum_{k=k_{\rm in}}^{N-1}
\frac1h
\mathbb E
\left\|
\mathbb E_k[
\chi_{k+1}\xi_{k+1}
]
\right\|^2
\le
Ch
\sum_{k=k_{\rm in}}^{N-1}
\mathbb E\|\xi_{k+1}\|^2
\longrightarrow0.
\end{equation}
The same estimate holds for the contractions involving
$\xi_{k+1}^P$.  Thus no pointwise order is assigned to either residual; only
their aggregate energy and their interactions with the one-step expansion are
used below.

For reference, set
\[
\bar\lambda_{k+1}:=\mathbb E_k[\lambda_{k+1}],
\qquad
\bar P_{k+1}:=\mathbb E_k[P_{k+1}],
\]
and evaluate the chart-reduced coefficients and Hamiltonian at
$(t_k,\mathbf S_k,\bar{\mathbf a}_k,
\bar\lambda_{k+1},\mathbf Z_k^h)$.  With all state derivatives understood in
the fixed-latent sense, define
\begin{equation}\label{eq:second_open_loop_driver}
\begin{aligned}
\mathfrak D_k^{(2),\rm ol}
:={}&
D^2_{\mathbf s\mathbf s}\bar{\mathcal H}_k
+(D_{\mathbf s}\bar{\boldsymbol b}_k)^\top\bar P_{k+1}
+\bar P_{k+1}D_{\mathbf s}\bar{\boldsymbol b}_k
\\
&+
\sum_{\ell=1}^{d_W}
(D_{\mathbf s}\bar{\boldsymbol\sigma}_k^\ell)^\top
\bar P_{k+1}D_{\mathbf s}\bar{\boldsymbol\sigma}_k^\ell
\\
&+
\sum_{\ell=1}^{d_W}
\left[
(D_{\mathbf s}\bar{\boldsymbol\sigma}_k^\ell)^\top R_k^{h,\ell}
+R_k^{h,\ell}D_{\mathbf s}\bar{\boldsymbol\sigma}_k^\ell
\right].
\end{aligned}
\end{equation}
Define the remainder arrays
$\{\mathfrak e_k^{(1)},\mathfrak e_k^{(2)}\}$ by the exact adapted
identities
\begin{align}
\lambda_k-\lambda_{k+1}
&=
D_{\mathbf s}\bar{\mathcal H}_k h
-\mathbf Z_k^h\Delta\mathbf W_k
-\xi_{k+1}+\mathfrak e_k^{(1)},
\label{eq:first_projected_recursion}
\\
P_k-P_{k+1}
&=
\mathfrak D_k^{(2),\rm ol}h
-\sum_{\ell=1}^{d_W}R_k^{h,\ell}\Delta W_k^\ell
-\xi_{k+1}^P+\mathfrak e_k^{(2)}.
\label{eq:second_projected_recursion}
\end{align}

\begin{applemma}[Discrete adjoint consistency]
\label{lem:discrete_adjoint_consistency}
Under Assumption~\ref{ass:bptt_pmp_reg} and the one-step bounds in
Lemma~\ref{lem:fixed_latent_one_step_expansion}, the remainders in
\eqref{eq:first_projected_recursion}--
\eqref{eq:second_projected_recursion} satisfy
\begin{equation}\label{eq:appendix_remainder_sufficient_condition}
\sum_{k=k_{\rm in}}^{N-1}
\frac1h
\mathbb E
\left[
\|\mathfrak e_k^{(1)}\|^2
+
\|\mathfrak e_k^{(2)}\|_F^2
\right]
\longrightarrow0.
\end{equation}
Consequently,
\begin{equation}\label{eq:aggregate_remainder_condition}
\begin{aligned}
\mathcal E_h
:={}&
\mathbb E\max_{k_{\rm in}\le m\le N-1}
\left\|\sum_{k=k_{\rm in}}^m\mathfrak e_k^{(1)}\right\|^2
\\
&+
\mathbb E\max_{k_{\rm in}\le m\le N-1}
\left\|\sum_{k=k_{\rm in}}^m\mathfrak e_k^{(2)}\right\|_F^2
\longrightarrow0.
\end{aligned}
\end{equation}
\end{applemma}

\begin{proof}
Insert the decompositions
\eqref{eq:appendix_one_step_decomposition} and
\eqref{eq:orthogonal_one_step_decomposition} into the exact conditional
first- and second-order chain rules.  After the displayed leading driver terms
are removed, each remainder admits a finite decomposition
\[
\mathfrak e_k^{(i)}
=
h\eta_{k,h}^{(i)}
+
\tau_{k,h}^{(i)}
+
\rho_{k,h}^{(i)},
\qquad i=1,2,
\]
with the following three sources.

First, $h\eta_{k,h}^{(i)}$ contains coefficient-consistency errors between the
leading coefficients of the discrete transition derivatives and the limiting
fixed-latent derivatives.  Assumption~\ref{ass:bptt_pmp_reg}(iii) gives
\begin{equation}\label{eq:appendix_consistency_error_class}
\sum_{k=k_{\rm in}}^{N-1}
h\mathbb E
\left[
\|\eta_{k,h}^{(1)}\|^2
+
\|\eta_{k,h}^{(2)}\|_F^2
\right]
\longrightarrow0.
\end{equation}

Second, $\tau_{k,h}^{(i)}$ contains the purely Taylor terms: conditional means
of the higher-order remainders, drift--drift products, and products involving a
centered second-chaos term but no orthogonal projection residual.  Conditional
centering and first-chaos orthogonality remove the nominal order-$h$ terms.
Using
\eqref{eq:appendix_second_chaos_bounds}--
\eqref{eq:appendix_one_step_remainder_bounds}, the mesh-uniform adjoint energy,
and conditional Cauchy--Schwarz gives
\begin{equation}\label{eq:appendix_taylor_error_class}
\sum_{k=k_{\rm in}}^{N-1}
\frac1h
\mathbb E
\left[
\|\tau_{k,h}^{(1)}\|^2
+
\|\tau_{k,h}^{(2)}\|_F^2
\right]
\le
Ch+o(1).
\end{equation}
For example, a conditional mean of a pure one-step remainder is $O(h^2)$,
whereas a term of the form
$\mathbb E_k[\operatorname{Rem}_{k+1}^\top
\mathbf Z_k^h\Delta\mathbf W_k]$ has squared norm bounded by
$Ch^4\|\mathbf Z_k^h\|_F^2$; summation is therefore controlled by the
interval $\mathbb H^2$ energy.  The matrix terms are handled identically with
$P$ and $\mathbf R$.

Third, $\rho_{k,h}^{(i)}$ contains all interactions with
$\xi_{k+1}$ or $\xi_{k+1}^P$.  Constant and first-chaos factors vanish by
conditional orthogonality.  The only potentially leading interactions are the
second-chaos contractions, controlled by
\eqref{eq:appendix_aggregate_residual_chaos}.  Interactions with a higher-order
remainder satisfy, for example,
\[
\mathbb E
\left\|
\mathbb E_k[
\operatorname{Rem}_{k+1}^\top\xi_{k+1}
]
\right\|^2
\le
Ch^3\mathbb E\|\xi_{k+1}\|^2,
\]
and analogously for $\xi_{k+1}^P$.  Together with
\eqref{eq:appendix_projection_energy}, this yields
\begin{equation}\label{eq:appendix_residual_error_class}
\sum_{k=k_{\rm in}}^{N-1}
\frac1h
\mathbb E
\left[
\|\rho_{k,h}^{(1)}\|^2
+
\|\rho_{k,h}^{(2)}\|_F^2
\right]
\longrightarrow0.
\end{equation}
Equations~\eqref{eq:appendix_consistency_error_class}--
\eqref{eq:appendix_residual_error_class} prove
\eqref{eq:appendix_remainder_sufficient_condition}.
Finally, for $i=1,2$,
\[
\max_{k_{\rm in}\le m\le N-1}
\left\|\sum_{k=k_{\rm in}}^m\mathfrak e_k^{(i)}\right\|^2
\le
\frac{T}{h}
\sum_{k=k_{\rm in}}^{N-1}\|\mathfrak e_k^{(i)}\|^2,
\]
with the Frobenius norm for $i=2$.  Hence
\eqref{eq:appendix_remainder_sufficient_condition} implies
\eqref{eq:aggregate_remainder_condition}.
\end{proof}

\subsection{Proof of Theorem~\texorpdfstring{
\ref{thm:bptt_pmp_constrained}}{}}
\label{app:proof_olbptt_pmp_theorem}

\begin{proof}[Proof of Theorem~\ref{thm:bptt_pmp_constrained}]

\textbf{Step 1: exact pathwise identities.}

For $k\ge k_{\rm in}$,
\[
\mathcal J_k
=
\mathfrak r_k
+
\mathcal J_{k+1}.
\]
Under fixed-latent differentiation,
$\mathcal J_{k+1}$ depends on $\mathbf S_k$ only through the transition
$\mathbf S_{k+1}$.  The multivariate chain rule therefore gives
\[
\widetilde\lambda_k
=
D_{\mathbf S_k}^{\rm ol}\mathfrak r_k
+
A_k^\top\widetilde\lambda_{k+1},
\]
and
\[
\widetilde P_k
=
D_{\mathbf S_k\mathbf S_k}^{2,\rm ol}\mathfrak r_k
+
A_k^\top\widetilde P_{k+1}A_k
+
\sum_{j=1}^{d_S}
\widetilde\lambda_{k+1}^j B_k^j.
\]
These are exactly
\eqref{eq:bptt_chain_rule} and
\eqref{eq:bptt_second_order}.  They use neither stationarity nor optimality.

Because $A_k$ and $B_k^j$ are
$\mathcal F_{t_{k+1}}$-measurable, the tower property gives
\[
\mathbb E_k[
A_k^\top\widetilde\lambda_{k+1}
]
=
\mathbb E_k[
A_k^\top\lambda_{k+1}
],
\]
\[
\mathbb E_k[
A_k^\top\widetilde P_{k+1}A_k
]
=
\mathbb E_k[
A_k^\top P_{k+1}A_k
],
\]
and
\[
\mathbb E_k[
\widetilde\lambda_{k+1}^jB_k^j
]
=
\mathbb E_k[
\lambda_{k+1}^jB_k^j
].
\]
Thus conditional projection converts the raw pathwise identities into exact
adapted one-step identities.

\medskip\noindent
\textbf{Step 2: first-adjoint driver.}

Write
\[
\bar\lambda_{k+1}
=
\mathbb E_k[\lambda_{k+1}]
\]
and use the decomposition
\[
\lambda_{k+1}
=
\bar\lambda_{k+1}
+
\mathbf Z_k^h\Delta\mathbf W_k
+
\xi_{k+1}.
\]
Insert this decomposition and the fixed-latent expansion of $A_k$ from
Lemma~\ref{lem:fixed_latent_one_step_expansion} into the adapted first-order
identity.

The predictable order-$h$ terms consist of the state derivative of the
one-step reward, the drift-Jacobian term multiplied by
$\bar\lambda_{k+1}$, and the diffusion-Jacobian term contracted with
$\mathbf Z_k^h$.  They combine into
\[
D_{\mathbf s}\bar{\mathcal H}_k h.
\]
The first-chaos term gives
\[
\mathbf Z_k^h\Delta\mathbf W_k,
\]
which is subtracted when
$\lambda_{k+1}-\bar\lambda_{k+1}$
is moved to the left-hand side.  The orthogonal component remains explicitly
as $\xi_{k+1}$.

Collect all nonleading terms in $\mathfrak e_k^{(1)}$.  They comprise
centered second-chaos and higher-order Taylor terms, coefficient-consistency
errors, and interactions with the projection residual.  Their cumulative
control is precisely the first-order part of
Lemma~\ref{lem:discrete_adjoint_consistency}.  Hence
\[
\lambda_k-\lambda_{k+1}
=
D_{\mathbf s}\bar{\mathcal H}_k h
-
\mathbf Z_k^h\Delta\mathbf W_k
-
\xi_{k+1}
+
\mathfrak e_k^{(1)},
\]
which proves
\eqref{eq:first_projected_recursion}.

\medskip\noindent
\textbf{Step 3: second-adjoint driver.}

Write
\[
\bar P_{k+1}
=
\mathbb E_k[P_{k+1}]
\]
and use
\[
P_{k+1}
=
\bar P_{k+1}
+
\sum_{\ell=1}^{d_W}
R_k^{h,\ell}\Delta W_k^\ell
+
\xi_{k+1}^P.
\]
Insert this decomposition and the one-step expansions of $A_k$ and $B_k^j$
into the adapted second-order identity.

The predictable order-$h$ terms in
\[
\mathbb E_k[
A_k^\top P_{k+1}A_k
]
\]
are
\[
(D_{\mathbf s}\bar{\boldsymbol b}_k)^\top
\bar P_{k+1}
+
\bar P_{k+1}
D_{\mathbf s}\bar{\boldsymbol b}_k,
\]
\[
\sum_{\ell=1}^{d_W}
(D_{\mathbf s}\bar{\boldsymbol\sigma}_k^\ell)^\top
\bar P_{k+1}
D_{\mathbf s}\bar{\boldsymbol\sigma}_k^\ell,
\]
and
\[
\sum_{\ell=1}^{d_W}
\left[
(D_{\mathbf s}\bar{\boldsymbol\sigma}_k^\ell)^\top
R_k^{h,\ell}
+
R_k^{h,\ell}
D_{\mathbf s}\bar{\boldsymbol\sigma}_k^\ell
\right].
\]
The last line arises from the product of the first-chaos component of the
transition Jacobian with the first-chaos projection of $P_{k+1}$.

The reward Hessian and the transition-Hessian contraction produce
\[
\begin{aligned}
D^2_{\mathbf s\mathbf s}\bar\ell_k
&+
\sum_{j=1}^{d_S}
\bar\lambda_{k+1}^j
D^2_{\mathbf s\mathbf s}\bar b_k^j
\\
&+
\sum_{j=1}^{d_S}
\sum_{\ell=1}^{d_W}
(\mathbf Z_k^h)^{j\ell}
D^2_{\mathbf s\mathbf s}
\bar\sigma_k^{j\ell}
=
D^2_{\mathbf s\mathbf s}\bar{\mathcal H}_k.
\end{aligned}
\]
Together, these terms give the discrete second-order driver
$\mathfrak D_k^{(2),\rm ol}$ in
\eqref{eq:second_open_loop_driver}.

All nonleading terms are collected in
$\mathfrak e_k^{(2)}$.  Lemma~\ref{lem:discrete_adjoint_consistency}
shows jointly for the two recursions that
\eqref{eq:appendix_remainder_sufficient_condition} and
\eqref{eq:aggregate_remainder_condition} hold.  Therefore,
\[
P_k-P_{k+1}
=
\mathfrak D_k^{(2),\rm ol}h
-
\sum_{\ell=1}^{d_W}
R_k^{h,\ell}\Delta W_k^\ell
-
\xi_{k+1}^P
+
\mathfrak e_k^{(2)},
\]
which proves
\eqref{eq:second_projected_recursion}.

\medskip\noindent
\textbf{Step 4: continuous-time limit.}

Define the piecewise-constant interpolants, for
$t\in[t_k,t_{k+1})$, by
\[
\lambda_t^h
:=
\lambda_k,
\qquad
P_t^h
:=
P_k,
\qquad
\mathbf Z_t^h
:=
\mathbf Z_k^h,
\qquad
R_t^{h,\ell}
:=
R_k^{h,\ell}.
\]
Summing
\eqref{eq:first_projected_recursion} and
\eqref{eq:second_projected_recursion}
over the grid gives approximate integral forms of the chart-reduced first-
and second-adjoint BSDEs.  By
\eqref{eq:appendix_residual_maximal_bound} and
\eqref{eq:aggregate_remainder_condition}, the cumulative orthogonal residuals
and collected remainders vanish in $\mathbb S^2$.

Assumption~\ref{ass:bptt_pmp_reg}(iii) gives convergence of the frozen latent
processes, charted controls, states, barred coefficients, and their first two
fixed-latent state derivatives.  The terminal data converge because the
terminal state remains in the common compact localization and the terminal
reward is $C^2$ there.

For completeness, apply BSDE stability sequentially.  Let
$\eta_h^{(1)}\to0$ collect the squared terminal error, the integrated
first-driver consistency error, and the squared $\mathbb S^2$ norm of the
cumulative residual and remainder perturbations.  The usual It\^o--Young
estimate for the first vector BSDE gives
\begin{equation}\label{eq:appendix_first_discrete_bsde_stability}
\|\lambda^h-\lambda\|_{\mathbb S^2}^2
+
\|\mathbf Z^h-\mathbf Z\|_{\mathbb H^2}^2
\le
C\eta_h^{(1)}
\longrightarrow0.
\end{equation}
Next let $\eta_h^{(2)}\to0$ denote the analogous terminal, coefficient, and
cumulative perturbation error for the matrix equation.  Its driver is Lipschitz
in $(P,\mathbf R)$ on the common localization and depends continuously on
$(\lambda,\mathbf Z)$.  The matrix-valued BSDE estimate and
\eqref{eq:appendix_first_discrete_bsde_stability} therefore give
\begin{equation}\label{eq:appendix_second_discrete_bsde_stability}
\|P^h-P\|_{\mathbb S^2}^2
+
\|\mathbf R^h-\mathbf R\|_{\mathbb H^2}^2
\le
C\left(
\eta_h^{(2)}
+
\|\lambda^h-\lambda\|_{\mathbb S^2}^2
+
\|\mathbf Z^h-\mathbf Z\|_{\mathbb H^2}^2
\right)
\longrightarrow0.
\end{equation}
Combining the two estimates yields
\[
\begin{aligned}
&
\mathbb E
\sup_{t_{\rm in}\le t\le T}
\left(
\|\lambda_t^h-\lambda_t\|^2
+
\|P_t^h-P_t\|_F^2
\right)
\\
&\quad
+
\mathbb E
\int_{t_{\rm in}}^T
\left(
\|\mathbf Z_t^h-\mathbf Z_t\|_F^2
+
\|\mathbf R_t^h-\mathbf R_t\|_F^2
\right)dt
\longrightarrow0.
\end{aligned}
\]
The limiting processes solve
\eqref{eq:chart_reduced_first_adjoint} and
\eqref{eq:chart_reduced_second_adjoint}
along the limiting fixed-latent reference control.  If that reference control
is optimal, the limiting processes are the corresponding PMP adjoints for the open-loop control problem.  This proves
\eqref{eq:full_adjoint_bsde_convergence}.

\medskip\noindent
\textbf{Step 5: wealth-relevant blocks.}

Conditional projection and the preceding convergence hold for the full vector
and matrix.  Premultiplication by $e_X^\top$, postmultiplication by $e_X$, and
extraction of the wealth--factor block are continuous linear operations.
Consequently,
\[
\widetilde\lambda^X
\longrightarrow
\lambda^X,
\qquad
\widetilde P^{XX}
\longrightarrow
P^{XX},
\qquad
\widetilde P^{XY}
\longrightarrow
P^{XY}
\]
after conditional projection and interpolation.  These are precisely the
blocks entering
\eqref{eq:factor_qp_coefficients}.  No identification of $P^{XY}$ with
$V_{XY}$ is used.

\end{proof}

\begin{appremark}[Filtration and recovery separation]
\label{rem:appendix_filtration_recovery}
The Brownian-filtration assumption ensures the martingale-representation
property used in
\eqref{eq:appendix_first_interval_density} and
\eqref{eq:appendix_second_interval_density}.  Under a larger filtration, the
adjoint systems and discrete projections must retain the corresponding
orthogonal martingale components.  No barrier term enters either adjoint
recursion; constrained recovery is applied only after the adapted adjoints
have been estimated.
\end{appremark}

\section{Reference-Policy Stability}
\label{app:proof_reference_performance}

\subsection{From Reference Performance to Adjoint Proximity}
\label{sec:performance_adjoint_bridge}

Theorem~\ref{thm:bptt_pmp_constrained} identifies the adjoints associated with
a fixed reference control, but it does not compare that reference with the
optimal control.  We provide this link under a local quadratic-growth
condition.  The result below supplies the performance-to-adjoint component of
the end-to-end guarantee stated as
Theorem~\ref{thm:end_to_end_qp_recovery} in the main manuscript.

We use the process norms $\mathbb S^2$, $\mathbb H^2$, and $L^{q,p}$
defined in Appendix~\ref{app:notation_summary}.

For every admissible control $\mathbf u$ in the localization below, let
\[
(\mathbf S^{\mathbf u},
 \lambda^{\mathbf u},
 \mathbf Z^{\mathbf u},
 P^{\mathbf u},
 \mathbf R^{\mathbf u})
\]
denote the corresponding state and exact chart-reduced adjoint systems.
Define
\begin{equation}\label{eq:control_indexed_shift}
\boldsymbol\zeta_t^{\mathbf u}
:=
\mathbf Z_t^{\mathbf u}
-
P_t^{\mathbf u}
\boldsymbol\sigma_S
(t,\mathbf S_t^{\mathbf u};\mathbf u_t),
\qquad
\vartheta^{\mathbf u}
:=
(\lambda^{\mathbf u},
 \boldsymbol\zeta^{\mathbf u},
 P^{\mathbf u}).
\end{equation}

\begin{appassumption}[Reference-policy quadratic growth and coefficient-jet regularity]
\label{ass:reference_qg}
Let $\mathbf u^*$ be a locally optimal admissible control.  There exist a
localized admissible comparison neighborhood
$\mathfrak U_0\subset\mathbb H^2$ of $\mathbf u^*$, understood relative to
a common compact chart localization, and constants
$\kappa_{\rm qg},C_{\rm jet},C_{\rm adj}>0$ such that:

\begin{enumerate}[label=(\roman*),leftmargin=*]

\item \emph{Common localization.}
Every $\mathbf u\in\mathfrak U_0$ and its state remain in the common chart
localization.  On this set, the chart-reduced drift, diffusion, and reward,
their fixed-latent state derivatives through order two, and the first two
derivatives of the terminal reward are bounded and Lipschitz with constant
$C_{\rm jet}$.  The required forward and adjoint moments are uniform over
$\mathfrak U_0$.  In addition, the exact adjoints satisfy the uniform
essential bound
\begin{equation}\label{eq:reference_uniform_adjoint_bound}
\sup_{\mathbf u\in\mathfrak U_0}
\left(
\|\lambda^{\mathbf u}\|_{L^\infty(dt\otimes d\mathbb P)}
+
\|\mathbf Z^{\mathbf u}\|_{L^\infty(dt\otimes d\mathbb P)}
+
\|P^{\mathbf u}\|_{L^\infty(dt\otimes d\mathbb P)}
+
\|\mathbf R^{\mathbf u}\|_{L^\infty(dt\otimes d\mathbb P)}
\right)
\le C_{\rm adj}.
\end{equation}

\item \emph{Local quadratic growth.}
For every $\mathbf u\in\mathfrak U_0$,
\begin{equation}\label{eq:reference_quadratic_growth}
J(\mathbf u^*)-J(\mathbf u)
\ge
\kappa_{\rm qg}
\|\mathbf u-\mathbf u^*\|_{\mathbb H^2}^2.
\end{equation}

\end{enumerate}
\end{appassumption}

\begin{appremark}[A smooth-Markov sufficient condition for the martingale-adjoint bound]
\label{rem:smooth_markov_z_bound}
The explicit bound on $\mathbf Z^{\mathbf u}$ in
\eqref{eq:reference_uniform_adjoint_bound} cannot in general be replaced by
finiteness of arbitrarily high moments: multiplication by an unbounded
$\mathbf Z^{\mathbf u}$ is not a bounded operation on $\mathbb H^2$.
The level adjoint $\lambda^{\mathbf u}$ is already bounded for the present
linear BSDE when its terminal value and coefficients are bounded; it is listed
in \eqref{eq:reference_uniform_adjoint_bound} for a uniform statement.

A standard sufficient setting for the nontrivial $\mathbf Z$ bound is a
smooth Markovian decoupling field for the adjoint BSDE.  More precisely,
suppose that, for every $\mathbf u\in\mathfrak U_0$, there is a deterministic
field $\Lambda^{\mathbf u}$ such that
\[
\lambda_t^{\mathbf u}
=
\Lambda^{\mathbf u}(t,\mathbf S_t^{\mathbf u}),
\qquad
\mathbf Z_t^{\mathbf u}
=
D_{\mathbf s}\Lambda^{\mathbf u}(t,\mathbf S_t^{\mathbf u})
\boldsymbol\sigma_S
(t,\mathbf S_t^{\mathbf u};\mathbf u_t),
\]
and that the fields $\Lambda^{\mathbf u}$ are uniformly $C^1$ in the state on
the compact localization.  Boundedness of the localized diffusion then gives
the required uniform essential bound on $\mathbf Z^{\mathbf u}$.  At a smooth
Markov optimum, the usual PMP--dynamic-programming identification gives
$\Lambda^*=D_{\mathbf s}V$ and hence
$\mathbf Z^*=D_{\mathbf s\mathbf s}^2V\,\boldsymbol\sigma_S$.
For a general nonstationary reference feedback, the decoupling field need not
coincide with the gradient of the closed-loop policy-evaluation function,
because differentiating that function also differentiates the feedback.  The
stated decoupling-field condition avoids this distinction.
\end{appremark}

\begin{apptheorem}[Performance-to-adjoint stability for a near-optimal reference]
\label{thm:reference_performance_adjoint}
Suppose Assumptions~\ref{ass:Gamma},
\ref{ass:chart_pmp_regularity}, and
\ref{ass:reference_qg} hold.  Let
$\mathbf u^{\rm ref}\in\mathfrak U_0$ and define
\[
\varepsilon_{\rm ref}
:=
J(\mathbf u^*)-J(\mathbf u^{\rm ref})
\ge0.
\]
Then
\begin{equation}\label{eq:reference_control_sqrt_rate}
\|\mathbf u^{\rm ref}-\mathbf u^*\|_{\mathbb H^2}
\le
\sqrt{
\frac{\varepsilon_{\rm ref}}
{\kappa_{\rm qg}}
}.
\end{equation}
Moreover, there exists a constant $C_{\rm stab}>0$, independent of
$\varepsilon_{\rm ref}$, such that
\begin{align}
&
\|\mathbf S^{\rm ref}-\mathbf S^*\|_{\mathbb S^2}
+
\|\lambda^{\rm ref}-\lambda^*\|_{\mathbb S^2}
+
\|\mathbf Z^{\rm ref}-\mathbf Z^*\|_{\mathbb H^2}
\notag
\\
&\quad
+
\|P^{\rm ref}-P^*\|_{\mathbb S^2}
+
\|\mathbf R^{\rm ref}-\mathbf R^*\|_{\mathbb H^2}
+
\|\boldsymbol\zeta^{\rm ref}
-\boldsymbol\zeta^*\|_{\mathbb H^2}
\le
C_{\rm stab}
\sqrt{
\frac{\varepsilon_{\rm ref}}
{\kappa_{\rm qg}}
}.
\label{eq:reference_full_adjoint_sqrt_rate}
\end{align}
In particular, for a constant $C_{\rm ref}>0$,
\begin{equation}\label{eq:reference_theta_sqrt_rate}
\|\vartheta^{\rm ref}-\vartheta^*\|_{L^{2,2}}
\le
C_{\rm ref}
\sqrt{
\frac{\varepsilon_{\rm ref}}
{\kappa_{\rm qg}}
}.
\end{equation}
\end{apptheorem}

\begin{proof}
See Appendix~\ref{app:proof_reference_performance}.
\end{proof}

\begin{appremark}[Scope of the quadratic-growth bridge]
\label{rem:reference_qg_scope}
The pointwise curvature conditions
$P^{XX}<0$ and $\boldsymbol\Sigma\succ0$ make the fixed-adjoint risky-weight
block strongly concave, but they do not by themselves imply
\eqref{eq:reference_quadratic_growth} for the full dynamic objective.
Quadratic growth is a coercivity condition for the stochastic second variation
on the admissible process space.  The neighborhood in
Assumption~\ref{ass:reference_qg} is relative to the stated localization; no
ambient $\mathbb H^2$ openness is asserted for a hard pathwise state constraint.

Theorem~\ref{thm:reference_performance_adjoint} is trajectory-wise and
integrated in time.  A uniform comparison of adjoint fields at identical
off-trajectory deployment states additionally requires quadratic growth for
the corresponding restarted problems or an appropriate
state-coverage/concentrability condition.  The same distinction applies when
\eqref{eq:training_objective} averages over a random initial law.
Online Supporting Information~\ref{app:merton_qg_verification} makes the
scope explicit: uniform unweighted $\mathbb H^2$ curvature fails in the
untruncated adapted CRRA problem, while a genuine untruncated
quadratic-growth verification is available on bounded deterministic control
classes.  The latter is a narrower nonvacuous example, not a verification for
the full adapted class.
\end{appremark}


\subsection{Proof of Theorem~\texorpdfstring{
\ref{thm:reference_performance_adjoint}}{}}
\label{app:reference_performance_proof}

The proof is stated for the objective $J$ appearing in
Assumption~\ref{ass:reference_qg}.  If $J$ is instead the objective averaged
over the random initial time--state law in
\eqref{eq:training_objective}, the same argument applies conditionally on the
initial pair and then after integration, provided that the localization and
stability constants below are uniform over its support.

For each
$\mathbf u\in\mathfrak U_0$, let
$\mathbf a^{\mathbf u}$ denote its unique latent representation on the common
chart patch, so that
\[
\mathbf u_t
=
\Psi
(t,\mathbf S_t^{\mathbf u},\mathbf a_t^{\mathbf u}).
\]
The embedding property in
Assumption~\ref{ass:chart_pmp_regularity}
and compact localization imply uniform local Lipschitz control of the latent
coordinates by the corresponding state and control variables.  In particular,
with
$\Delta\mathbf a:=\mathbf a^{\rm ref}-\mathbf a^*$,
\begin{equation}\label{eq:proof_reference_latent_stability}
\|\Delta\mathbf a_t\|
\le
C
\left(
\|\Delta\mathbf S_t\|
+
\|\Delta\mathbf u_t\|
\right)
\qquad
(dt\otimes d\mathbb P)\text{-a.e.}
\end{equation}

For compactness, write
\[
\Delta\mathbf u
:=
\mathbf u^{\rm ref}-\mathbf u^*,
\qquad
\Delta\mathbf S
:=
\mathbf S^{\rm ref}-\mathbf S^*,
\]
and use analogous notation for the adjoint processes.  All constants below
depend only on the common localization in
Assumption~\ref{ass:reference_qg} and may change from line to line.

\begin{proof}[Proof of Theorem~\ref{thm:reference_performance_adjoint}]

\textbf{Step 1: performance loss to control distance.}

Evaluating the quadratic-growth condition
\eqref{eq:reference_quadratic_growth}
at $\mathbf u^{\rm ref}$ gives
\[
\kappa_{\rm qg}
\|\Delta\mathbf u\|_{\mathbb H^2}^2
\le
J(\mathbf u^*)-J(\mathbf u^{\rm ref})
=
\varepsilon_{\rm ref}.
\]
Hence
\[
\|\Delta\mathbf u\|_{\mathbb H^2}
\le
\sqrt{
\frac{\varepsilon_{\rm ref}}
{\kappa_{\rm qg}}
},
\]
which proves
\eqref{eq:reference_control_sqrt_rate}.

\medskip\noindent
\textbf{Step 2: forward-state stability.}

Under a generic admissible control $\mathbf u\in\mathfrak U_0$, the state
satisfies
\[
d\mathbf S_t^{\mathbf u}
=
\boldsymbol b_S
(t,\mathbf S_t^{\mathbf u};\mathbf u_t)\,dt
+
\boldsymbol\sigma_S
(t,\mathbf S_t^{\mathbf u};\mathbf u_t)\,d\mathbf W_t.
\]
Subtracting the equations for
$\mathbf u^{\rm ref}$ and $\mathbf u^*$ and using the uniform Lipschitz bounds
on the common localization gives
\[
\begin{aligned}
&
\left\|
\boldsymbol b_S
(t,\mathbf S_t^{\rm ref};\mathbf u_t^{\rm ref})
-
\boldsymbol b_S
(t,\mathbf S_t^*;\mathbf u_t^*)
\right\|
\\
&\qquad
+
\left\|
\boldsymbol\sigma_S
(t,\mathbf S_t^{\rm ref};\mathbf u_t^{\rm ref})
-
\boldsymbol\sigma_S
(t,\mathbf S_t^*;\mathbf u_t^*)
\right\|_F
\\
&\le
C
\left(
\|\Delta\mathbf S_t\|
+
\|\Delta\mathbf u_t\|
\right).
\end{aligned}
\]
The Burkholder--Davis--Gundy inequality and Gronwall's lemma therefore yield
\begin{equation}\label{eq:proof_reference_state_stability}
\|\Delta\mathbf S\|_{\mathbb S^2}
\le
C
\|\Delta\mathbf u\|_{\mathbb H^2}.
\end{equation}

\medskip\noindent
\textbf{Step 3: first-adjoint stability.}

Let
\[
g(\mathbf s)
:=
K e^{-\rho T}U(x),
\qquad
\mathbf s=(x,y),
\]
denote the terminal reward.  For a generic control
$\mathbf u\in\mathfrak U_0$, define the chart-reduced first-adjoint driver by
\[
F_1^{\mathbf u}
(t,\lambda,\mathbf Z)
:=
D_{\mathbf s}\bar{\mathcal H}
\left(
t,
\mathbf S_t^{\mathbf u},
\mathbf a_t^{\mathbf u},
\lambda,
\mathbf Z
\right).
\]
The first-adjoint BSDE is
\[
-d\lambda_t^{\mathbf u}
=
F_1^{\mathbf u}
(t,\lambda_t^{\mathbf u},\mathbf Z_t^{\mathbf u})\,dt
-
\mathbf Z_t^{\mathbf u}\,d\mathbf W_t,
\qquad
\lambda_T^{\mathbf u}
=
D_{\mathbf s}g(\mathbf S_T^{\mathbf u}).
\]

For compact notation, set
\[
\begin{aligned}
c_t^{\mathbf u}
&:=
D_{\mathbf s}\bar\ell
(t,\mathbf S_t^{\mathbf u},\mathbf a_t^{\mathbf u}),
\\
A_t^{\mathbf u}
&:=
D_{\mathbf s}\bar{\boldsymbol b}
(t,\mathbf S_t^{\mathbf u},\mathbf a_t^{\mathbf u}),
\\
B_t^{\mathbf u,\ell}
&:=
D_{\mathbf s}\bar{\boldsymbol\sigma}^{\ell}
(t,\mathbf S_t^{\mathbf u},\mathbf a_t^{\mathbf u}).
\end{aligned}
\]
Then
\[
F_1^{\mathbf u}(t,\lambda,\mathbf Z)
=
c_t^{\mathbf u}
+
(A_t^{\mathbf u})^\top\lambda
+
\sum_{\ell=1}^{d_W}
(B_t^{\mathbf u,\ell})^\top\mathbf Z^\ell.
\]
Consequently, the driver difference has the exact decomposition
\begin{equation}\label{eq:proof_reference_first_driver_decomposition}
\begin{aligned}
&F_1^{\rm ref}
(t,\lambda_t^{\rm ref},\mathbf Z_t^{\rm ref})
-
F_1^*
(t,\lambda_t^*,\mathbf Z_t^*)
\\
={}&
\Delta c_t
+
(\Delta A_t)^\top\lambda_t^{\rm ref}
+
\sum_{\ell=1}^{d_W}
(\Delta B_t^\ell)^\top
(\mathbf Z_t^{\rm ref})^\ell
\\
&+
(A_t^*)^\top\Delta\lambda_t
+
\sum_{\ell=1}^{d_W}
(B_t^{*,\ell})^\top(\Delta\mathbf Z_t)^\ell.
\end{aligned}
\end{equation}
The coefficient-jet Lipschitz bounds and
\eqref{eq:proof_reference_latent_stability} give
\[
\|\Delta c_t\|
+
\|\Delta A_t\|_F
+
\sum_{\ell=1}^{d_W}\|\Delta B_t^\ell\|_F
\le
C
\left(
\|\Delta\mathbf S_t\|
+
\|\Delta\mathbf u_t\|
\right).
\]
Using the uniform essential bounds on
$\lambda^{\rm ref}$ and $\mathbf Z^{\rm ref}$ from
\eqref{eq:reference_uniform_adjoint_bound} in
\eqref{eq:proof_reference_first_driver_decomposition} therefore yields
\[
\begin{aligned}
&
\left\|
F_1^{\rm ref}
(t,\lambda_t^{\rm ref},\mathbf Z_t^{\rm ref})
-
F_1^*
(t,\lambda_t^*,\mathbf Z_t^*)
\right\|
\\
&\le
C
\left(
\|\Delta\mathbf S_t\|
+
\|\Delta\mathbf u_t\|
+
\|\Delta\lambda_t\|
+
\|\Delta\mathbf Z_t\|_F
\right).
\end{aligned}
\]
The terminal-gradient Lipschitz condition similarly gives
\[
\left\|
D_{\mathbf s}g(\mathbf S_T^{\rm ref})
-
D_{\mathbf s}g(\mathbf S_T^*)
\right\|
\le
C\|\Delta\mathbf S_T\|.
\]

Applying It\^o's formula to
$e^{\beta t}\|\Delta\lambda_t\|^2$, choosing $\beta$ sufficiently large, and
using Young's inequality and the standard BSDE estimate yield
\begin{equation}\label{eq:proof_reference_first_bsde_stability}
\|\Delta\lambda\|_{\mathbb S^2}
+
\|\Delta\mathbf Z\|_{\mathbb H^2}
\le
C
\left(
\|\Delta\mathbf S\|_{\mathbb S^2}
+
\|\Delta\mathbf u\|_{\mathbb H^2}
\right).
\end{equation}

\medskip\noindent
\textbf{Step 4: second-adjoint stability.}

For a generic control $\mathbf u$, let
\[
F_2^{\mathbf u}
(t,\lambda,\mathbf Z,P,\mathbf R)
\]
denote the matrix driver in
\eqref{eq:chart_reduced_second_adjoint},
evaluated at
\[
\left(
t,
\mathbf S_t^{\mathbf u},
\mathbf a_t^{\mathbf u},
\lambda,
\mathbf Z,
P,
\mathbf R
\right).
\]
Then
\[
-dP_t^{\mathbf u}
=
F_2^{\mathbf u}
\left(
t,
\lambda_t^{\mathbf u},
\mathbf Z_t^{\mathbf u},
P_t^{\mathbf u},
\mathbf R_t^{\mathbf u}
\right)dt
-
\sum_{\ell=1}^{d_W}
R_t^{\mathbf u,\ell}\,dW_t^\ell,
\]
with terminal condition
\[
P_T^{\mathbf u}
=
D_{\mathbf s\mathbf s}^2g
(\mathbf S_T^{\mathbf u}).
\]

Write
\[
\Xi_t^{\mathbf u}
:=
(\lambda_t^{\mathbf u},\mathbf Z_t^{\mathbf u},
 P_t^{\mathbf u},\mathbf R_t^{\mathbf u}).
\]
The difference of the two drivers is split into a coefficient perturbation and
an adjoint-variable perturbation:
\begin{equation}\label{eq:proof_reference_second_driver_split}
\begin{aligned}
F_2^{\rm ref}(t,\Xi_t^{\rm ref})
-
F_2^*(t,\Xi_t^*)
={}&
\left[
F_2^{\rm ref}(t,\Xi_t^{\rm ref})
-
F_2^*(t,\Xi_t^{\rm ref})
\right]
\\
&+
\left[
F_2^*(t,\Xi_t^{\rm ref})
-
F_2^*(t,\Xi_t^*)
\right].
\end{aligned}
\end{equation}
In the first bracket, every coefficient-jet difference is bounded by
\[
C
\left(
\|\Delta\mathbf S_t\|
+
\|\Delta\mathbf a_t\|
\right)
\le
C
\left(
\|\Delta\mathbf S_t\|
+
\|\Delta\mathbf u_t\|
\right)
\]
by \eqref{eq:proof_reference_latent_stability}.  These differences multiply
only bounded coefficient jets and the reference variables
$\lambda^{\rm ref}$, $\mathbf Z^{\rm ref}$,
$P^{\rm ref}$, and $\mathbf R^{\rm ref}$; the required bounds are exactly
\eqref{eq:reference_uniform_adjoint_bound}.  This includes the potentially
problematic terms
\[
\sum_{j,\ell}
(\mathbf Z_t^{\rm ref})^{j\ell}
\left(
D_{\mathbf s\mathbf s}^2\bar\sigma_{t,\rm ref}^{j\ell}
-
D_{\mathbf s\mathbf s}^2\bar\sigma_{t,*}^{j\ell}
\right).
\]
The second bracket in
\eqref{eq:proof_reference_second_driver_split} is Lipschitz in
$(\lambda,\mathbf Z,P,\mathbf R)$ because the second-adjoint driver is affine
in these variables on the common localization.  Hence
\[
\begin{aligned}
&
\left\|
F_2^{\rm ref}(t,\Xi_t^{\rm ref})
-
F_2^*(t,\Xi_t^*)
\right\|_F
\\
&\le
C
\Big(
\|\Delta\mathbf S_t\|
+
\|\Delta\mathbf u_t\|
+
\|\Delta\lambda_t\|
+
\|\Delta\mathbf Z_t\|_F
+
\|\Delta P_t\|_F
+
\|\Delta\mathbf R_t\|_F
\Big).
\end{aligned}
\]
The terminal-Hessian Lipschitz condition gives
\[
\left\|
D_{\mathbf s\mathbf s}^2g(\mathbf S_T^{\rm ref})
-
D_{\mathbf s\mathbf s}^2g(\mathbf S_T^*)
\right\|_F
\le
C\|\Delta\mathbf S_T\|.
\]

Applying the corresponding matrix-valued BSDE stability estimate and using
\eqref{eq:proof_reference_first_bsde_stability} yield
\begin{equation}\label{eq:proof_reference_second_bsde_stability}
\begin{aligned}
\|\Delta P\|_{\mathbb S^2}
+
\|\Delta\mathbf R\|_{\mathbb H^2}
\le
C
\Big(
&
\|\Delta\mathbf S\|_{\mathbb S^2}
+
\|\Delta\mathbf u\|_{\mathbb H^2}
\\
&
+
\|\Delta\lambda\|_{\mathbb S^2}
+
\|\Delta\mathbf Z\|_{\mathbb H^2}
\Big).
\end{aligned}
\end{equation}

\medskip\noindent
\textbf{Step 5: shifted-adjoint stability.}

Define
\[
\boldsymbol\sigma_t^{\rm ref}
:=
\boldsymbol\sigma_S
(t,\mathbf S_t^{\rm ref};\mathbf u_t^{\rm ref}),
\qquad
\boldsymbol\sigma_t^*
:=
\boldsymbol\sigma_S
(t,\mathbf S_t^*;\mathbf u_t^*).
\]
From
\eqref{eq:control_indexed_shift},
\[
\begin{aligned}
\Delta\boldsymbol\zeta_t
={}&
\Delta\mathbf Z_t
-
\Delta P_t\,
\boldsymbol\sigma_t^{\rm ref}
-
P_t^*
\left(
\boldsymbol\sigma_t^{\rm ref}
-
\boldsymbol\sigma_t^*
\right).
\end{aligned}
\]
The diffusion coefficient is bounded and Lipschitz on the common
localization, and $P^*$ is uniformly essentially bounded.  Hence
\[
\left\|
\boldsymbol\sigma^{\rm ref}
-
\boldsymbol\sigma^*
\right\|_{\mathbb H^2}
\le
C
\left(
\|\Delta\mathbf S\|_{\mathbb S^2}
+
\|\Delta\mathbf u\|_{\mathbb H^2}
\right),
\]
and therefore
\begin{equation}\label{eq:proof_reference_shift_stability}
\begin{aligned}
\|\Delta\boldsymbol\zeta\|_{\mathbb H^2}
\le
C
\Big(
&
\|\Delta\mathbf Z\|_{\mathbb H^2}
+
\|\Delta P\|_{\mathbb S^2}
\\
&
+
\|\Delta\mathbf S\|_{\mathbb S^2}
+
\|\Delta\mathbf u\|_{\mathbb H^2}
\Big).
\end{aligned}
\end{equation}

\medskip\noindent
\textbf{Step 6: combination of the estimates.}

Combining
\eqref{eq:proof_reference_state_stability}--
\eqref{eq:proof_reference_shift_stability}
gives
\[
\begin{aligned}
&
\|\Delta\mathbf S\|_{\mathbb S^2}
+
\|\Delta\lambda\|_{\mathbb S^2}
+
\|\Delta\mathbf Z\|_{\mathbb H^2}
+
\|\Delta P\|_{\mathbb S^2}
\\
&\qquad
+
\|\Delta\mathbf R\|_{\mathbb H^2}
+
\|\Delta\boldsymbol\zeta\|_{\mathbb H^2}
\le
C_{\rm stab}
\|\Delta\mathbf u\|_{\mathbb H^2}.
\end{aligned}
\]
Using
\eqref{eq:reference_control_sqrt_rate}
proves
\eqref{eq:reference_full_adjoint_sqrt_rate}.

Finally, for any level process $\chi$,
\[
\|\chi\|_{L^{2,2}}
\le
\sqrt T\,\|\chi\|_{\mathbb S^2},
\]
while
\[
\|\boldsymbol\zeta\|_{L^{2,2}}
=
\|\boldsymbol\zeta\|_{\mathbb H^2}.
\]
Since
\[
\vartheta
=
(\lambda,\boldsymbol\zeta,P),
\]
we obtain
\[
\|\vartheta^{\rm ref}-\vartheta^*\|_{L^{2,2}}
\le
C_{\rm ref}
\sqrt{
\frac{\varepsilon_{\rm ref}}
{\kappa_{\rm qg}}
},
\]
which proves
\eqref{eq:reference_theta_sqrt_rate}.

\end{proof}

\section{Estimated-Input QP Recovery}
\label{app:end_to_end_recovery_proofs}

This appendix isolates the QP branch of the recovery analysis.  Appendix~\ref{app:proof_reference_performance}
controls the reference and target adjoints along their paired trajectories.
Here all inputs are evaluated instead at a common deployment-state process,
and the resulting adjoint-input perturbation is propagated through the
strongly concave local QP.

\subsection{Common-State Adjoint Inputs and End-to-End QP Bounds}
\label{sec:end_to_end_recovery}

Fix a common deployment-state process $\mathbf S$ taking values in a compact
deployment region.  Let
\[
\widehat\vartheta_t
=
(\widehat\lambda_t,
 \widehat{\boldsymbol\zeta}_t,
 \widehat P_t)
\]
denote the estimated input,
$\vartheta_t^{\rm ref}$ the exact input associated with the frozen reference
control, and $\vartheta_t^*$ the target input.  In the portfolio
specialization, these symbols may be read as their wealth-relevant
restrictions.

\begin{appassumption}[Common-state adjoint-input errors]
\label{ass:qp_adjoint_input_stability}
Let
\[
\delta_{\rm ol}
:=
\delta_{\rm time}(h)
+
\delta_{\rm MC}(N_{\rm MC})
+
\delta_{\rm reg},
\]
where each component converges to zero under consistent refinement.  Assume
\[
\left\|
(\widehat\lambda,\widehat P)
-
(\lambda^{\rm ref},P^{\rm ref})
\right\|_{L^{q,p}}
\le
C_{\rm ol}\delta_{\rm ol},
\]
\[
\left\|
\widehat{\boldsymbol\zeta}
-
\boldsymbol\zeta^{\rm ref}
\right\|_{L^{q,p}}
\le
\delta_\zeta,
\qquad
\left\|
\vartheta^{\rm ref}
-
\vartheta^*
\right\|_{L^{q,p}}
\le
\delta_{\rm ref}.
\]
Define
\[
\delta_{\rm adj}
:=
\delta_{\rm ol}
+
\delta_\zeta
+
\delta_{\rm ref}.
\]
The estimated input and the exact and numerical QP solutions remain in common
compact localizations.
\end{appassumption}

By the triangle inequality, there is a constant $C_{\rm in}>0$ such that
\begin{equation}\label{eq:aggregate_adjoint_input_bound}
\|\widehat\vartheta-\vartheta^*\|_{L^{q,p}}
\le
C_{\rm in}\delta_{\rm adj}.
\end{equation}

Theorem~\ref{thm:reference_performance_adjoint} gives a trajectory-wise
$O(\sqrt{\varepsilon_{\rm ref}/\kappa_{\rm qg}})$ bound.  Since the QP
comparison is made at identical deployment states, the additional bridge is
stated explicitly.

\begin{appassumption}[Common-state reference-to-target bridge]
\label{ass:common_state_reference_bridge}
For the common deployment-state process $\mathbf S$ above, there exists a
constant $C_{\rm cs}<\infty$ such that, for $p=q=2$,
\begin{equation}\label{eq:common_state_reference_bridge}
\left\|
\vartheta^{\rm ref}(\cdot,\mathbf S_\cdot)
-
\vartheta^*(\cdot,\mathbf S_\cdot)
\right\|_{L^{2,2}}
\le
C_{\rm cs}
\sqrt{
\frac{\varepsilon_{\rm ref}}
{\kappa_{\rm qg}}
}.
\end{equation}
Equivalently,
$\delta_{\rm ref}\le
C_{\rm cs}\sqrt{\varepsilon_{\rm ref}/\kappa_{\rm qg}}$.
\end{appassumption}

Possible sufficient mechanisms are discussed in
Remark~\ref{rem:reference_qg_scope}: restarted-problem stability or a
suitable state-coverage/concentrability condition.  No such implication is
asserted without the corresponding additional regularity.

\subsection{QP-PGDPO Stability and End-to-End Rate}
\label{sec:qp_recovery}

Let
\[
\widehat{\mathbf g}_t
:=
\mathbf g(t,\mathbf S_t;\widehat\vartheta_t),
\qquad
\widehat Q_t
:=
Q(t,\mathbf S_t;\widehat\vartheta_t),
\]
where $\mathbf g$ and $Q$ are defined in
\eqref{eq:factor_qp_coefficients}.  Let
\[
\mathcal K_t
:=
\mathcal K(t,\mathbf S_t)
\]
be the convex risky-weight feasible set.

For a closed convex set $\mathcal K$, define the normal cone by
\[
N_{\mathcal K}(\boldsymbol\pi)
:=
\left\{
v:
\langle v,w-\boldsymbol\pi\rangle\le0
\text{ for all }w\in\mathcal K
\right\},
\qquad
\boldsymbol\pi\in\mathcal K.
\]
Let $\widehat{\boldsymbol\pi}_t^{\rm QP}$ be a numerical solution of the
estimated QP, and define its stationarity residual by
\[
r_t^{\rm qp}
:=
\operatorname{dist}\!\left(
0,
\widehat Q_t
\widehat{\boldsymbol\pi}_t^{\rm QP}
-
\widehat{\mathbf g}_t
+
N_{\mathcal K_t}
(\widehat{\boldsymbol\pi}_t^{\rm QP})
\right),
\qquad
\|r^{\rm qp}\|_{L^{q,p}}
\le
\delta_{\rm qp}.
\]

Let
$\boldsymbol\pi_t^*$
denote the exact QP solution obtained from
$\vartheta_t^*$ and define
\[
G_t^{\rm qp}
:=
\mathcal Q_{\rm loc}
(\boldsymbol\pi_t^*;\vartheta_t^*)
-
\mathcal Q_{\rm loc}
(\widehat{\boldsymbol\pi}_t^{\rm QP};\vartheta_t^*),
\]
where
\[
\mathcal Q_{\rm loc}
(\boldsymbol\pi;\vartheta)
:=
\mathbf g(\vartheta)^\top\boldsymbol\pi
-
\frac12
\boldsymbol\pi^\top Q(\vartheta)\boldsymbol\pi.
\]

\begin{appproposition}[Estimated-input stability of QP-PGDPO]
\label{prop:qp_policy_gap}
Suppose that $\mathcal K_t$ is a nonempty closed convex set and that, on the
input localization in
Assumption~\ref{ass:qp_adjoint_input_stability},
\[
Q(\vartheta)
\succeq
\underline\lambda_QI
\]
uniformly for some $\underline\lambda_Q>0$.  Suppose also that
$(\mathbf g,Q)$ are locally Lipschitz in $\vartheta$ and that the exact and
numerical QP solutions remain in a common bounded localization.  Then,
$dt\otimes d\mathbb P$-almost everywhere,
\begin{equation}\label{eq:qp_pointwise_policy_gap}
\|
\widehat{\boldsymbol\pi}_t^{\rm QP}
-
\boldsymbol\pi_t^*
\|
\le
C_{\rm qp}
\left(
e_t^{\rm adj}
+
r_t^{\rm qp}
\right),
\end{equation}
and therefore
\begin{equation}\label{eq:qp_policy_gap}
\|
\widehat{\boldsymbol\pi}^{\rm QP}
-
\boldsymbol\pi^*
\|_{L^{q,p}}
\le
C_{\rm qp}'
\left(
\delta_{\rm adj}
+
\delta_{\rm qp}
\right).
\end{equation}
Moreover,
\begin{equation}\label{eq:qp_pointwise_hamiltonian_gap}
0
\le
G_t^{\rm qp}
\le
C_Q
\left(
e_t^{\rm adj}
+
r_t^{\rm qp}
\right)^2.
\end{equation}
If $p,q\ge2$, then
\begin{equation}\label{eq:qp_hamiltonian_gap}
\|G^{\rm qp}\|_{L^{q/2,p/2}}
\le
C_Q'
\left(
\delta_{\rm adj}
+
\delta_{\rm qp}
\right)^2.
\end{equation}
\end{appproposition}

\begin{proof}
See Appendix~\ref{app:proof_qp_stability}.
\end{proof}

\begin{appcorollary}[Consistency and end-to-end rate of QP-PGDPO]
\label{cor:qp_consistency}
For an exact QP solve,
$\delta_{\rm qp}=0$.  Hence
\[
\delta_{\rm ol},
\quad
\delta_\zeta,
\quad
\delta_{\rm ref}
\longrightarrow0
\]
implies
\[
\widehat{\boldsymbol\pi}^{\rm QP}
\longrightarrow
\boldsymbol\pi^*
\qquad
\text{in }L^{q,p},
\]
and, for $p,q\ge2$, the true local QP-gap process converges to zero in
$L^{q/2,p/2}$.

If Assumptions~\ref{ass:reference_qg} and
\ref{ass:common_state_reference_bridge} hold, then for $p=q=2$,
\begin{equation}\label{eq:qp_end_to_end_control_rate}
\|
\widehat{\boldsymbol\pi}^{\rm QP}
-
\boldsymbol\pi^*
\|_{L^{2,2}}
\le
\widetilde C_{\rm qp}
\left(
\delta_{\rm ol}
+
\delta_\zeta
+
\delta_{\rm qp}
+
\sqrt{
\frac{\varepsilon_{\rm ref}}
{\kappa_{\rm qg}}
}
\right),
\end{equation}
and
\begin{equation}\label{eq:qp_end_to_end_gap_rate}
\|G^{\rm qp}\|_{L^{1,1}}
\le
\widetilde C_Q
\left(
\delta_{\rm ol}
+
\delta_\zeta
+
\delta_{\rm qp}
+
\sqrt{
\frac{\varepsilon_{\rm ref}}
{\kappa_{\rm qg}}
}
\right)^2.
\end{equation}
Thus, when adjoint acquisition, shifted-input estimation, and the QP solve are
exact, the recovered risky-weight error is
$O(\sqrt{\varepsilon_{\rm ref}})$ and the integrated true local QP gap is
$O(\varepsilon_{\rm ref})$.
\end{appcorollary}

\begin{appremark}[Interpretation of the QP end-to-end rate]
\label{rem:qp_rate_interpretation}
The recovery step does not generally improve the policy-distance order:
both the reference control and the recovered risky weights remain
$O(\sqrt{\varepsilon_{\rm ref}})$-close to their targets.  The
$O(\varepsilon_{\rm ref})$ conclusion concerns the integrated true local QP
value loss and follows from strong concavity and the squared stationarity
error.  It is not a pointwise Hamiltonian statement or a matching rate for the
unrecovered DPO action.
\end{appremark}


\subsection{Proof of Proposition~\texorpdfstring{
\ref{prop:qp_policy_gap}}{}}
\label{app:proof_qp_stability}

Fix a deployment point and suppress $(t,\omega)$ throughout.  Write
\[
\mathbf g^*
:=
\mathbf g(\vartheta^*),
\qquad
Q^*
:=
Q(\vartheta^*),
\]
and define the true-input optimality operator
\begin{equation}\label{eq:appendix_true_qp_operator}
\mathcal T^*(\boldsymbol\pi)
:=
Q^*\boldsymbol\pi
-
\mathbf g^*
+
N_{\mathcal K}(\boldsymbol\pi).
\end{equation}
The exact true-input solution satisfies
\[
\boldsymbol0
\in
\mathcal T^*(\boldsymbol\pi^*).
\]

\begin{proof}[Proof of Proposition~\ref{prop:qp_policy_gap}]

\textbf{Step 1: transfer of the estimated-input residual.}

Let
\[
\widehat{\mathbf g}
:=
\mathbf g(\widehat\vartheta),
\qquad
\widehat Q
:=
Q(\widehat\vartheta),
\]
and let
$\widehat{\boldsymbol\pi}^{\rm QP}\in\mathcal K$
be the numerical QP output.  By the definition of
$r^{\rm qp}$, for every $\epsilon>0$ there exists
$\widehat v\in N_{\mathcal K}
(\widehat{\boldsymbol\pi}^{\rm QP})$
such that
\[
\left\|
\widehat Q\widehat{\boldsymbol\pi}^{\rm QP}
-
\widehat{\mathbf g}
+
\widehat v
\right\|
\le
r^{\rm qp}+\epsilon.
\]
Using the same normal vector in the true-input operator gives
\[
\begin{aligned}
\operatorname{dist}
\left(
0,
\mathcal T^*
(\widehat{\boldsymbol\pi}^{\rm QP})
\right)
\le{}&
 r^{\rm qp}+\epsilon
+
\left\|
(Q^*-\widehat Q)
\widehat{\boldsymbol\pi}^{\rm QP}
\right\|
+
\left\|
\widehat{\mathbf g}-\mathbf g^*
\right\|.
\end{aligned}
\]
The common bounded localization and the local Lipschitz dependence of
$(\mathbf g,Q)$ on $\vartheta$ imply
\[
\left\|
(Q^*-\widehat Q)
\widehat{\boldsymbol\pi}^{\rm QP}
\right\|
+
\left\|
\widehat{\mathbf g}-\mathbf g^*
\right\|
\le
C e^{\rm adj}.
\]
Letting $\epsilon\downarrow0$ yields
\begin{equation}\label{eq:appendix_qp_true_residual_bound}
\operatorname{dist}
\left(
0,
\mathcal T^*
(\widehat{\boldsymbol\pi}^{\rm QP})
\right)
\le
C
\left(
 e^{\rm adj}
+
r^{\rm qp}
\right).
\end{equation}

\medskip\noindent
\textbf{Step 2: pointwise and process-level policy bounds.}

Because
$Q^*\succeq\underline\lambda_Q I$
and the normal-cone mapping is monotone,
$\mathcal T^*$ is
$\underline\lambda_Q$-strongly monotone.  Thus, for any
$w\in\mathcal T^*
(\widehat{\boldsymbol\pi}^{\rm QP})$,
\[
\underline\lambda_Q
\left\|
\widehat{\boldsymbol\pi}^{\rm QP}
-
\boldsymbol\pi^*
\right\|^2
\le
\left\langle
w,
\widehat{\boldsymbol\pi}^{\rm QP}
-
\boldsymbol\pi^*
\right\rangle.
\]
Taking the infimum over $w$ and using Cauchy--Schwarz give
\[
\left\|
\widehat{\boldsymbol\pi}^{\rm QP}
-
\boldsymbol\pi^*
\right\|
\le
\frac{1}{\underline\lambda_Q}
\operatorname{dist}
\left(
0,
\mathcal T^*
(\widehat{\boldsymbol\pi}^{\rm QP})
\right).
\]
Combining this estimate with
\eqref{eq:appendix_qp_true_residual_bound}
proves
\eqref{eq:qp_pointwise_policy_gap}.
Taking the $L^{q,p}$ norm and using
\eqref{eq:aggregate_adjoint_input_bound}
and
$\|r^{\rm qp}\|_{L^{q,p}}\le\delta_{\rm qp}$
yield
\eqref{eq:qp_policy_gap}.

\medskip\noindent
\textbf{Step 3: true-input QP gap.}

Let $I_{\mathcal K}$ be the convex indicator of $\mathcal K$ and define
\begin{equation}\label{eq:appendix_true_qp_convex_objective}
\phi^*(\boldsymbol\pi)
:=
\frac12
\boldsymbol\pi^\top
Q^*
\boldsymbol\pi
-
\mathbf g^{*\top}
\boldsymbol\pi
+
I_{\mathcal K}(\boldsymbol\pi).
\end{equation}
Then
\[
\partial\phi^*(\boldsymbol\pi)
=
\mathcal T^*(\boldsymbol\pi),
\]
and $\phi^*$ is
$\underline\lambda_Q$-strongly convex.  For any
$w\in\partial\phi^*
(\widehat{\boldsymbol\pi}^{\rm QP})$,
strong convexity gives
\[
\phi^*(\boldsymbol\pi^*)
\ge
\phi^*(\widehat{\boldsymbol\pi}^{\rm QP})
+
\left\langle
w,
\boldsymbol\pi^*
-
\widehat{\boldsymbol\pi}^{\rm QP}
\right\rangle
+
\frac{\underline\lambda_Q}{2}
\left\|
\widehat{\boldsymbol\pi}^{\rm QP}
-
\boldsymbol\pi^*
\right\|^2.
\]
Maximizing the resulting upper bound over the distance to
$\boldsymbol\pi^*$ yields the standard subgradient estimate
\[
0
\le
\phi^*(\widehat{\boldsymbol\pi}^{\rm QP})
-
\phi^*(\boldsymbol\pi^*)
\le
\frac{1}{2\underline\lambda_Q}
\operatorname{dist}
\left(
0,
\partial\phi^*
(\widehat{\boldsymbol\pi}^{\rm QP})
\right)^2.
\]
Both controls are feasible, so the left-hand side equals the true-input
concave-QP loss $G^{\rm qp}$.  Applying
\eqref{eq:appendix_qp_true_residual_bound}
therefore gives
\[
0
\le
G^{\rm qp}
\le
C_Q
\left(
 e^{\rm adj}
+
r^{\rm qp}
\right)^2,
\]
which proves
\eqref{eq:qp_pointwise_hamiltonian_gap}.

If $p,q\ge2$, then
\[
\bigl\||f|^2\bigr\|_{L^{q/2,p/2}}
=
\|f\|_{L^{q,p}}^2.
\]
Using the aggregate input and numerical-residual bounds gives
\[
\|G^{\rm qp}\|_{L^{q/2,p/2}}
\le
C_Q'
\left(
\delta_{\rm adj}
+
\delta_{\rm qp}
\right)^2,
\]
which is
\eqref{eq:qp_hamiltonian_gap}.

\end{proof}

\begin{appremark}[Uniformity and dimension dependence]
\label{rem:uniformity_norm}
The QP estimates are uniform only on the specified compact deployment
localization.  They require a uniform lower eigenvalue bound for $Q$ and
bounded local Lipschitz constants for $(\mathbf g,Q)$.  The constants may
depend on portfolio dimension and covariance conditioning.  The framework
avoids a state grid but does not imply dimension-free linear algebra.
\end{appremark}

\clearpage
\begin{center}
{\Large\bfseries Online Appendix / Supporting Information}\par
\medskip
{\small Classical open-loop/feedback specializations, quadratic-growth scope
results, barrier/KKT theory, implementation details, and supplementary
numerical audits.}
\end{center}
\medskip
\setcounter{section}{0}
\renewcommand{\thesection}{OA.\arabic{section}}
\renewcommand{\theHsection}{OA.\arabic{section}}
\renewcommand{\thesubsection}{\thesection.\arabic{subsection}}
\renewcommand{\theHsubsection}{\theHsection.\arabic{subsection}}
\renewcommand{\thesubsubsection}{\thesubsection.\arabic{subsubsection}}
\renewcommand{\theHsubsubsection}{\theHsubsection.\arabic{subsubsection}}
\numberwithin{equation}{section}
\counterwithin{table}{section}
\counterwithin{figure}{section}
\counterwithin{algorithm}{section}
\section{First-Order, Feedback, and Smooth-Markov Specializations}
\label{app:smooth_markov_details}

\subsection{First-Order Variation on a Regular Chart Patch}
\label{app:first_order_variation}

The first-order adjoint is meaningful beyond the special cases in which it
fully characterizes the control.  The following localized statement records
the exact role of convex-amplitude variations in the chart coordinates; see
\citet{bensoussan1982lectures}, \citet{cadenillas1995stochastic}, and
\citet[Chapter~3]{yong2012stochastic} for classical first-order stochastic
maximum principles.

\begin{appproposition}[Localized first-order variational condition]
\label{prop:first_order_chart_variation}
Suppose Assumption~\ref{ass:chart_pmp_regularity} holds.  Let
$\mathbf a^*$ be optimal for the chart-reduced open-loop problem on the
selected regular patch, and let $\mathbf v$ be a bounded progressively
measurable process.  Put
$\mathbf h:=\mathbf v-\mathbf a^*$ and suppose that there exist
$\varepsilon_0>0$ and a compact subset
$\mathfrak A_0\subset\mathfrak A$ such that
\[
\mathbf a_t^\varepsilon
:=
\mathbf a_t^*+\varepsilon\mathbf h_t
\in\mathfrak A_0,
\qquad
0\leq\varepsilon\leq\varepsilon_0,
\]
for $dt\otimes d\mathbb P$-almost every $(t,\omega)$.  Suppose each
$\mathbf a^\varepsilon$ generates an admissible state trajectory that remains
in the common compact chart patch.  Then
\begin{equation}\label{eq:first_order_integral_variation}
\mathbb E\!\left[
\int_0^T
\left\langle
D_{\mathbf a}\bar{\mathcal H}
(t,\mathbf S_t^*,\mathbf a_t^*,\lambda_t^*,\mathbf Z_t^*),
\mathbf v_t-\mathbf a_t^*
\right\rangle dt
\right]
\leq 0.
\end{equation}
If, for this fixed comparator $\mathbf v$, the localized admissible class is
also stable when $\mathbf h$ is multiplied by indicators of predictable
time--event sets, then
\begin{equation}\label{eq:first_order_pointwise_variation}
\left\langle
D_{\mathbf a}\bar{\mathcal H}
(t,\mathbf S_t^*,\mathbf a_t^*,\lambda_t^*,\mathbf Z_t^*),
\mathbf v_t-\mathbf a_t^*
\right\rangle
\leq0
\end{equation}
for $dt\otimes d\mathbb P$-almost every $(t,\omega)$.  Thus the exceptional
null set may depend on the chosen comparator $\mathbf v$.
\end{appproposition}

\begin{proof}
This is the standard convex-amplitude first-variation argument for a
controlled diffusion, applied to the chart-reduced coefficients; see, for
example, \citet[Chapter~3]{yong2012stochastic}.  We give the short argument to
make the scope explicit.  Let
$\mathbf Y_t:=\left.\frac{d}{d\varepsilon}
\mathbf S_t^\varepsilon\right|_{\varepsilon=0+}$.
The $C^1$ coefficient regularity and bounded derivatives on the common
compact localization imply the standard state differentiability, and
$\mathbf Y$ solves
\begin{equation}\label{eq:first_order_state_variation}
\begin{aligned}
d\mathbf Y_t
={}&
\left[
D_{\mathbf s}\bar{\boldsymbol b}_t^*\,\mathbf Y_t
+
D_{\mathbf a}\bar{\boldsymbol b}_t^*\,\mathbf h_t
\right]dt
\\
&+
\sum_{\ell=1}^{d_W}
\left[
D_{\mathbf s}\bar{\boldsymbol\sigma}_t^{*,\ell}\,\mathbf Y_t
+
D_{\mathbf a}\bar{\boldsymbol\sigma}_t^{*,\ell}\,\mathbf h_t
\right]dW_t^\ell,
\qquad
\mathbf Y_0=0,
\end{aligned}
\end{equation}
where a star means evaluation at
$(t,\mathbf S_t^*,\mathbf a_t^*)$.
Differentiating the objective at $0+$ gives
\begin{equation}\label{eq:first_order_objective_derivative}
\begin{aligned}
\left.\frac{d}{d\varepsilon}J(\mathbf a^\varepsilon)
\right|_{\varepsilon=0+}
={}&
\mathbb E\!\left[
K e^{-\rho T}U'(X_T^*)e_X^\top\mathbf Y_T
\right]
\\
&+
\mathbb E\!\left[
\int_0^T
\left[
\left\langle
D_{\mathbf s}\bar\ell_t^*,\mathbf Y_t
\right\rangle
+
\left\langle
D_{\mathbf a}\bar\ell_t^*,\mathbf h_t
\right\rangle
\right]dt
\right].
\end{aligned}
\end{equation}
It\^o's formula for
$\lambda_t^{*\top}\mathbf Y_t$, together with
\eqref{eq:chart_reduced_first_adjoint},
\eqref{eq:first_order_state_variation}, and the terminal condition for
$\lambda^*$, cancels the terms containing $\mathbf Y$.  Hence
\[
\left.\frac{d}{d\varepsilon}J(\mathbf a^\varepsilon)
\right|_{\varepsilon=0+}
=
\mathbb E\!\left[
\int_0^T
\left\langle
D_{\mathbf a}\bar{\mathcal H}_t^*,\mathbf h_t
\right\rangle dt
\right].
\]
Optimality of $\mathbf a^*$ makes this one-sided derivative nonpositive,
which proves \eqref{eq:first_order_integral_variation}.

For the pointwise statement, fix the comparator $\mathbf v$ and choose a
predictable version, up to $dt\otimes d\mathbb P$-null sets, of the integrable
process
\[
G_t
:=
\left\langle
D_{\mathbf a}\bar{\mathcal H}_t^*,\mathbf h_t
\right\rangle.
\]
Let $E$ be a predictable time--event set.  By the localization hypothesis, the
same argument applies to the direction $\mathbf 1_E\mathbf h$ and gives
\[
\mathbb E\!\left[
\int_0^T
\mathbf 1_E G_t\,dt
\right]
\leq0.
\]
For $n\ge1$, the set
\[
E_n
:=
\{(t,\omega):G_t>1/n\}
\]
is predictable.  Taking $E=E_n$ shows that each $E_n$ has
$dt\otimes d\mathbb P$-measure zero.  Since
$\{G>0\}=\bigcup_{n\ge1}E_n$, the process $G$ is nonpositive
$dt\otimes d\mathbb P$-almost everywhere for this fixed $\mathbf v$.
\end{proof}

Proposition~\ref{prop:first_order_chart_variation} allows the diffusion to
depend on the control.  Its conclusion is nevertheless infinitesimal: it
does not assert that the primitive first-order Hamiltonian is maximized over
finite control changes, that the condition is sufficient for dynamic
optimality, or that a variation may cross a nonregular intersection or leave
the selected chart patch.  Standard sufficient stochastic maximum principles
require additional concavity-type hypotheses; see, for example,
\citet{zhou1996sufficient}.  Those hypotheses are not imposed by the core
recovery results of this paper.

\subsection{When the Second-Order Correction Drops Out}
\label{app:second_order_dropout}

There is a simple reduction when the compared controls do not change the
diffusion.  Fix the optimal reference control $\mathbf u^*$ and its adjoints,
and let $\mathcal C$ be a set of feasible controls in the selected comparison
region.  If
\[
\boldsymbol\sigma_S(t,\mathbf s;\mathbf u)
=
\boldsymbol\sigma_S(t,\mathbf s;\mathbf u^*)
\qquad\text{for every }\mathbf u\in\mathcal C,
\]
then the shifted linear-diffusion term and the quadratic $P$-term in
\eqref{eq:generalized_hamiltonian} are constant over $\mathcal C$.  Hence
\begin{equation}\label{eq:control_independent_diffusion_reduction}
\begin{aligned}
&\mathbb H(t,\mathbf s;\mathbf u,\lambda^*,\boldsymbol\zeta^*,P^*)
-\mathbb H(t,\mathbf s;\mathbf u^*,\lambda^*,\boldsymbol\zeta^*,P^*)
\\
&\qquad=
\mathcal H(t,\mathbf s,\mathbf u,\lambda^*,\mathbf Z^*)
-\mathcal H(t,\mathbf s,\mathbf u^*,\lambda^*,\mathbf Z^*).
\end{aligned}
\end{equation}
Thus the second-order correction has no effect on that control comparison.
This applies blockwise to consumption in the present model because consumption
does not enter the diffusion.  When the feasible consumption section is a
convex interval---as under the separable consumption constraints used in the
numerical models---concavity of $U$ therefore makes the scalar first-order KKT
condition sufficient for the conditional scalar stage problem.

The reduction does not apply to the risky weights.  The control-dependent
component of the wealth-diffusion row is
\[
x\boldsymbol\pi^\top\boldsymbol v(t,y).
\]
Finite changes in $\boldsymbol\pi$ therefore produce the quadratic and hedging terms in
\eqref{eq:open_loop_portfolio_block}.  Although the first-order variational
condition remains necessary along feasible convex directions, it does not by
itself provide this finite-comparison curvature.  A first-order condition can
instead become sufficient for the full dynamic problem under a separate
sufficient-maximum-principle theorem with the required joint concavity and
admissibility hypotheses; no such global concavity claim is made here.

\subsection{Feedback Verification against the Open-Loop Class}
\label{app:feedback_verification}

The next result states the clean classical setting in which using a Markov
feedback is without loss at the level of the value.  It is a verification
result, not a claim that all optimal open-loop controls are Markov.

\begin{appproposition}[Feedback realization of the open-loop value]
\label{prop:feedback_realizes_open_loop_value}
Suppose the model data and feasible correspondence are Markov in
$(t,\mathbf s)$ and that $\mathcal U$ has a measurable graph.  Suppose further
that $V\in C^{1,2}$ solves \eqref{eq:HJB_constrained} with the terminal
condition in Section~\ref{sec:hjb_motivation} and with growth and integrability
conditions sufficient for the verification argument.  Assume that the HJB
supremum over the full feasible fiber is attained by a Borel measurable
selector
\[
\widehat{\mathbf u}(t,\mathbf s)
\in
\mathcal U(t,\mathbf s),
\]
and that the closed-loop state equation induced by this selector has an
admissible solution on the underlying filtered probability space.  Then, for
each initial pair in the classical verification domain,
\(
\widehat{\mathbf u}_t=\widehat{\mathbf u}(t,\mathbf S_t)
\)
attains the value in \eqref{eq:value_function} among all admissible open-loop
controls.
\end{appproposition}

\begin{proof}
For any admissible open-loop control, the HJB equation makes the drift of the
value process plus the running reward nonpositive.  It\^o's formula,
localization, and the stated integrability therefore give an objective no
larger than $V$.  For the measurable maximizing selector the HJB inequality
holds with equality, and the closed-loop control attains $V$.
\end{proof}

The viscosity characterization mentioned in
Section~\ref{sec:hjb_motivation} is not used in this proposition.  Moreover,
if a selector maximizes only over
$\mathcal U_{\rm ch}(t,\mathbf s)$ while $V$ solves the full-fiber HJB, the
verification equality need not hold for that selector.  A genuine branchwise
verification statement requires the branch-restricted value function and the
HJB equation whose supremum is taken over
$\mathcal U_{\rm ch}(t,\mathbf s)$.  No branchwise optimality is inferred here
merely from a local selector for the full-value HJB integrand.

Under weaker nonsmooth hypotheses, existence of a Markov control requires
additional measurable-selection, compactness, and controlled-martingale-
problem arguments; such results may produce weak or relaxed Markov controls
rather than a strong deterministic feedback
\citep{kurtz1998existence}.  If relevant path history admits a
finite-dimensional adapted statistic with closed controlled dynamics, the
proposition may be applied to the augmented state.  Without such a statistic,
no finite-dimensional feedback reduction is asserted.

\subsection{Smooth PMP--Dynamic-Programming Portfolio Identities}
\label{app:smooth_pmp_dp_identities}

Under the optional smooth Markov hypotheses of
Section~\ref{sec:markov_reduction}, premultiplication by $e_X^\top$ extracts
the wealth component of a vector and the wealth row of a matrix.  Applying it
to \eqref{eq:markov_adjoint_identification} gives the following identities
\citep{zhou1991unified,yong2012stochastic}:
\begin{equation}\label{eq:markov_wealth_row_identification}
\begin{aligned}
\lambda_t^{X,*}
:=
e_X^\top\lambda_t^*
&=
V_X(t,\mathbf S_t^*),
\\
\mathbf Z_t^{X,*}
:=
e_X^\top\mathbf Z_t^*
&=
e_X^\top
D^2_{\mathbf s\mathbf s}V(t,\mathbf S_t^*)
\boldsymbol\sigma_t^*,
\\
\boldsymbol\zeta_t^{X,*}
:=
e_X^\top\boldsymbol\zeta_t^*
&=
e_X^\top
\left[
D^2_{\mathbf s\mathbf s}V(t,\mathbf S_t^*)
-
P_t^*
\right]
\boldsymbol\sigma_t^*.
\end{aligned}
\end{equation}
Thus, the wealth row of $\mathbf Z_t^*$ is determined by the value Hessian in
the smooth Markov representation, whereas the wealth row of $P_t^*$ is
obtained from the second-order PMP adjoint system.  Their difference in the
state-diffusion directions gives the shifted wealth-row input used in
\eqref{eq:open_loop_portfolio_block}.

At a Markov optimum that is interior with respect to the risky-weight
constraints, the first-order condition for
\eqref{eq:open_loop_portfolio_block}, together with
\eqref{eq:markov_adjoint_identification}, reduces to
\[
0
=
x\left[
V_X\boldsymbol\alpha(t,y)
+
\boldsymbol\Sigma_{RY}(t,y)\nabla_yV_X
\right]
+
x^2V_{XX}\boldsymbol\Sigma(t,y)\boldsymbol\pi^*,
\]
where all derivatives of $V$ in this display are evaluated at $(t,x,y)$.
Hence, when $V_{XX}(t,x,y)<0$,
\begin{equation}\label{eq:interior_markov_portfolio}
\boldsymbol\pi^*(t,x,y)
=
-\frac{1}{xV_{XX}(t,x,y)}
\boldsymbol\Sigma(t,y)^{-1}
\left[
V_X(t,x,y)\boldsymbol\alpha(t,y)
+
\boldsymbol\Sigma_{RY}(t,y)\nabla_yV_X(t,x,y)
\right].
\end{equation}
Substituting \eqref{eq:markov_adjoint_identification} into
\eqref{eq:open_loop_portfolio_block} shows that the terms involving $P_t^*$
cancel.  Thus, \eqref{eq:interior_markov_portfolio} does not require
$P_t^*=D^2_{\mathbf s\mathbf s}V(t,\mathbf S_t^*)$.  When $d_Y=0$, the
intertemporal-hedging term disappears and the formula reduces to the
corresponding no-factor optimal portfolio rule.

\section{Quadratic-Growth Scope in the Terminal-Wealth Merton Model}
\label{app:merton_qg_verification}

The quadratic-growth condition used by the end-to-end theorem is a substantive
local coercivity assumption on the admissible process space; it is not implied
merely by pointwise negativity of $P^{XX}$ or positive definiteness of the local
portfolio covariance.  This section records two complementary benchmark facts.
Uniform curvature in the unweighted adapted $\mathbb H^2$ metric generally
fails in the untruncated geometric-Brownian Merton model, whereas genuine
untruncated quadratic growth holds on bounded convex classes of deterministic
controls.  The calculations are specific to the terminal-wealth benchmark and
are not asserted for the joint portfolio--consumption problem or for general
state-dependent feasible fibers.

Consider the constant-coefficient Merton model without intermediate
consumption, with CRRA utility
\[
U(x)=\frac{x^{1-\gamma}}{1-\gamma},\qquad \gamma>1,
\]
wealth dynamics
\[
\frac{dX_t^{\boldsymbol\pi}}{X_t^{\boldsymbol\pi}}
=
\left(r+\boldsymbol\pi_t^\top\boldsymbol\alpha\right)dt
+
\boldsymbol\pi_t^\top\boldsymbol v\,d\mathbf W_t,
\qquad X_0=x_0>0,
\]
and objective
\[
J_{\rm M}(\boldsymbol\pi)
:=
\mathbb E\!\left[K e^{-\rho T}U(X_T^{\boldsymbol\pi})\right],
\qquad K>0.
\]
Here $r$, $\boldsymbol\alpha$, and $\boldsymbol v$ are constant, and
\[
\boldsymbol\Sigma:=\boldsymbol v\boldsymbol v^\top
\succeq \underline\sigma I
\]
for some $\underline\sigma>0$.

\begin{appproposition}[Failure of uniform unweighted adapted-$\mathbb H^2$ curvature]
\label{prop:merton_no_adapted_h2_qg}
Let $\boldsymbol\pi^*$ be a bounded deterministic stationary point of the
unconstrained constant-coefficient Merton problem.  Suppose that for some
$t_0\in(0,T)$ the baseline wealth $X_{t_0}^{\boldsymbol\pi^*}$ has unbounded
upper support.  Then
\[
\inf_{\substack{h\text{ bounded predictable}\\
                 \|h\|_{\mathbb H^2}=1}}
\left\{
-D^2J_{\rm M}(\boldsymbol\pi^*)[h,h]
\right\}
=
0.
\]
Consequently, no ambient adapted $\mathbb H^2$ neighborhood of
$\boldsymbol\pi^*$ can satisfy a uniform strong-concavity bound, and the local
quadratic-growth condition in the unweighted $\mathbb H^2$ norm cannot hold at
this stationary point.
\end{appproposition}

\begin{proof}
For a bounded predictable direction $h$, define
\[
\begin{aligned}
A_{\boldsymbol\pi^*}(h)
&:=
\int_0^T
h_t^\top
\left(
\boldsymbol\alpha-\boldsymbol\Sigma\boldsymbol\pi_t^*
\right)dt
+
\int_0^T h_t^\top\boldsymbol v\,d\mathbf W_t,
\\
B(h)
&:=
\int_0^T h_t^\top\boldsymbol\Sigma h_t\,dt.
\end{aligned}
\]
Differentiating the exact log-wealth representation twice gives
\begin{equation}\label{eq:appendix_merton_adapted_second_variation}
\begin{aligned}
-D^2J_{\rm M}(\boldsymbol\pi^*)[h,h]
={}&
\frac{K e^{-\rho T}}{\gamma-1}
\mathbb E\!\left[
(X_T^{\boldsymbol\pi^*})^{1-\gamma}
\right.
\\
&\left.\hspace{13mm}\times
\left\{
(\gamma-1)^2
A_{\boldsymbol\pi^*}(h)^2
+
(\gamma-1)B(h)
\right\}
\right].
\end{aligned}
\end{equation}

Choose $\delta>0$ with $t_0+\delta\le T$ and a deterministic unit vector $e$.
For $M>0$, set
\[
E_M
:=
\{X_{t_0}^{\boldsymbol\pi^*}>M\},
\qquad
p_M:=\mathbb P(E_M),
\]
and, whenever $p_M>0$, define the predictable direction
\[
h_t^M
:=
\frac{\mathbf1_{E_M}\mathbf1_{(t_0,t_0+\delta]}(t)}
{\sqrt{\delta p_M}}e.
\]
Then $\|h^M\|_{\mathbb H^2}=1$, and
\[
\begin{aligned}
A_{\boldsymbol\pi^*}(h^M)
={}&
\frac{\mathbf1_{E_M}}{\sqrt{\delta p_M}}
\left[
\int_{t_0}^{t_0+\delta}
 e^\top
(\boldsymbol\alpha-\boldsymbol\Sigma\boldsymbol\pi_t^*)\,dt
+
\int_{t_0}^{t_0+\delta}e^\top\boldsymbol v\,d\mathbf W_t
\right],
\\
B(h^M)
={}&
\frac{\mathbf1_{E_M}}{\delta p_M}
\int_{t_0}^{t_0+\delta}e^\top\boldsymbol\Sigma e\,dt.
\end{aligned}
\]
Because $\boldsymbol\pi^*$ is deterministic, the multiplicative continuation
factor from $t_0$ to $T$ and the Gaussian polynomial factors in the preceding
display have conditional moments bounded independently of $M$.  Conditioning
on $\mathcal F_{t_0}$ in
\eqref{eq:appendix_merton_adapted_second_variation} therefore yields
\[
-D^2J_{\rm M}(\boldsymbol\pi^*)[h^M,h^M]
\le
C
\frac{
\mathbb E\!\left[
\mathbf1_{E_M}
(X_{t_0}^{\boldsymbol\pi^*})^{1-\gamma}
\right]
}{p_M}
\le
C M^{1-\gamma}.
\]
Since $\gamma>1$ and $p_M>0$ for arbitrarily large $M$, the right-hand side
tends to zero.  Finally, suppose local quadratic growth with constant
$\kappa>0$ held at the stationary point.  Apply it along
$\boldsymbol\pi^*+\theta h$ for a bounded predictable direction $h$, divide by
$\theta^2$, and let $\theta\downarrow0$.  The first directional derivative
vanishes by stationarity, while the second directional derivative is given by
\eqref{eq:appendix_merton_adapted_second_variation}; hence
$-D^2J_{\rm M}(\boldsymbol\pi^*)[h,h]\ge
2\kappa\|h\|_{\mathbb H^2}^2$, a contradiction.
\end{proof}

The failure above uses event-dependent directions.  It does not arise on a
deterministic control class, where an exact untruncated verification is
available.

\begin{appproposition}[Untruncated quadratic growth on deterministic control classes]
\label{prop:merton_deterministic_qg}
Let
$\mathfrak U_{\rm M}^{\rm det}\subset
L^2([0,T];\mathbb R^n)$ be a bounded convex set of deterministic risky-weight
controls.  Then $J_{\rm M}$ is strongly concave on
$\mathfrak U_{\rm M}^{\rm det}$.  In particular, every maximizer
$\boldsymbol\pi^*$ on a convex local neighborhood satisfies
\[
J_{\rm M}(\boldsymbol\pi^*)-J_{\rm M}(\boldsymbol\pi)
\ge
\kappa_{\rm M}^{\rm det}
\|\boldsymbol\pi-\boldsymbol\pi^*\|_{L^2(0,T)}^2
\]
for some $\kappa_{\rm M}^{\rm det}>0$.
\end{appproposition}

\begin{proof}
For deterministic $\boldsymbol\pi$,
\begin{equation}\label{eq:appendix_merton_deterministic_closed_form}
\begin{aligned}
J_{\rm M}(\boldsymbol\pi)
={}&
\frac{K e^{-\rho T}x_0^{1-\gamma}}{1-\gamma}
\exp\!\Bigg(
(1-\gamma)rT
+
(1-\gamma)
\int_0^T
\boldsymbol\alpha^\top\boldsymbol\pi_t\,dt
\\
&\hspace{34mm}
+
\frac{\gamma(\gamma-1)}2
\int_0^T
\boldsymbol\pi_t^\top
\boldsymbol\Sigma
\boldsymbol\pi_t\,dt
\Bigg).
\end{aligned}
\end{equation}
Let $\mathcal A(\boldsymbol\pi)$ denote the exponent in
\eqref{eq:appendix_merton_deterministic_closed_form}.  Direct differentiation
gives
\begin{equation}\label{eq:appendix_merton_deterministic_hessian}
D^2J_{\rm M}(\boldsymbol\pi)[h,h]
=
J_{\rm M}(\boldsymbol\pi)
\left[
\left(D\mathcal A(\boldsymbol\pi)[h]\right)^2
+
\gamma(\gamma-1)
\int_0^T h_t^\top\boldsymbol\Sigma h_t\,dt
\right].
\end{equation}
The objective is negative for $\gamma>1$.  Since
$\mathfrak U_{\rm M}^{\rm det}$ is bounded in $L^2$, the exponent is bounded
below there, and hence
\[
m_J
:=
\inf_{\boldsymbol\pi\in\mathfrak U_{\rm M}^{\rm det}}
\left(-J_{\rm M}(\boldsymbol\pi)\right)
>
0.
\]
Using
$\boldsymbol\Sigma\succeq\underline\sigma I$ in
\eqref{eq:appendix_merton_deterministic_hessian} yields
\[
D^2J_{\rm M}(\boldsymbol\pi)[h,h]
\le
-m_J\gamma(\gamma-1)\underline\sigma
\|h\|_{L^2(0,T)}^2.
\]
Thus $J_{\rm M}$ is strongly concave on the deterministic class.  Integrating
the Hessian bound along the segment from a maximizer to
$\boldsymbol\pi$ gives the stated quadratic-growth inequality, for example
with
$\kappa_{\rm M}^{\rm det}
=\tfrac12m_J\gamma(\gamma-1)\underline\sigma$.
\end{proof}

Proposition~\ref{prop:merton_no_adapted_h2_qg} shows that the local
quadratic-growth assumption used in the end-to-end theorem cannot be inferred
from the standard untruncated Merton model on the full adapted, unweighted
$\mathbb H^2$ space.  Proposition~\ref{prop:merton_deterministic_qg} supplies
a genuine untruncated, but narrower, nonvacuous example.  Neither calculation
establishes quadratic growth for intermediate consumption, control-coupled
factor dynamics, or general state-dependent feasible fibers.

\section{Barrier--KKT Theory and Estimated-Input B-PGDPO Stability}
\label{app:barrier_theory}

This section collects the theory for the secondary log-barrier realization.
The main manuscript retains the solver-neutral recovery map and the statement
of the local barrier--KKT approximation; the central-path and estimated-input
proofs are recorded here so that the Published Appendix can focus on the two
flagship QP/adjoint theorems.

\subsection{Local Barrier--KKT Approximation}
\label{app:proof_barrier_policy}

This subsection proves Proposition~\ref{prop:barrier_policy}.  Numerical
implementation details are given separately in
Online Supporting Information~\ref{app:solver_details}: the Newton--conjugate-gradient
(Newton--CG) system, fraction-to-the-boundary rule, line search, and stopping
criterion.

Fix the stagewise state and adjoint input
$(t,\mathbf s,\vartheta)$, where
\[
\vartheta=(\lambda,\boldsymbol\zeta,P),
\]
and suppress these arguments throughout.  Write
\[
\mathbb H_\vartheta(\mathbf u)
:=
\mathbb H(t,\mathbf s;\mathbf u,
\lambda,\boldsymbol\zeta,P),
\qquad
\Gamma(\mathbf u)
:=
\Gamma(t,\mathbf s;\mathbf u).
\]
The control is represented in the reduced coordinates
\[
\mathbf u=(\boldsymbol\pi^\top,C)^\top.
\]
The risk-free weight is
$\pi_0=1-\mathbf1^\top\boldsymbol\pi$, and all pointwise restrictions,
including the no-borrowing condition when imposed, are included among the
components of $\Gamma(\mathbf u)\ge\boldsymbol0$.  Hence no separate
full-investment equality or equality multiplier is needed.

Although the feasible fiber may depend on the state, the state is fixed in
this stagewise problem.  Accordingly, the KKT map below contains only control
derivatives.  State derivatives of the feasible chart enter the adjoint
system, not the fixed-state recovery problem.

For
\[
z:=(\mathbf u,\nu)
\]
and $\varepsilon\ge0$, define the perturbed KKT map
\begin{equation}\label{eq:appendix_central_path_map}
\mathcal F_{\rm cp}(z;\varepsilon)
:=
\begin{pmatrix}
\nabla_{\mathbf u}\mathbb H_\vartheta(\mathbf u)
+
D_{\mathbf u}\Gamma(\mathbf u)^\top\nu
\\[1mm]
\operatorname{diag}(\Gamma(\mathbf u))\nu
-
\varepsilon\mathbf1
\end{pmatrix}.
\end{equation}
For $\varepsilon>0$, a solution satisfying
\[
\Gamma(\mathbf u)>\boldsymbol0,
\qquad
\nu>\boldsymbol0
\]
is equivalent to a stationary point of the logarithmic-barrier problem in
\eqref{eq:barrier_stage_problem}, with
\[
\nu_j
=
\frac{\varepsilon}{\Gamma_j(\mathbf u)}.
\]
At $\varepsilon=0$,
\eqref{eq:appendix_central_path_map} reduces to the KKT stationarity and
complementarity system.

For later use, define the stagewise Lagrangian
\begin{equation}\label{eq:appendix_stagewise_lagrangian}
\mathcal L_\vartheta(\mathbf u,\nu)
:=
\mathbb H_\vartheta(\mathbf u)
+
\nu^\top\Gamma(\mathbf u).
\end{equation}

\begin{proof}[Proof of Proposition~\ref{prop:barrier_policy}]

Let
\[
z^\star
:=
(\mathbf u^\star,\nu^\star)
\]
denote the selected KKT solution from
Assumption~\ref{ass:barrier_kkt}.  Set
\[
\varepsilon:=\varepsilon_{\rm bar}
\]
for notational brevity.

\medskip\noindent
\textbf{Step 1: local central-path branch.}

At $(z^\star,0)$,
\[
\mathcal F_{\rm cp}(z^\star;0)=0.
\]
The Jacobian of
\eqref{eq:appendix_central_path_map}
with respect to $z=(\mathbf u,\nu)$ is
\begin{equation}\label{eq:appendix_central_path_jacobian}
D_z\mathcal F_{\rm cp}(z;\varepsilon)
=
\begin{pmatrix}
\nabla_{\mathbf u\mathbf u}^2
\mathcal L_\vartheta(\mathbf u,\nu)
&
D_{\mathbf u}\Gamma(\mathbf u)^\top
\\[1mm]
\operatorname{diag}(\nu)
D_{\mathbf u}\Gamma(\mathbf u)
&
\operatorname{diag}(\Gamma(\mathbf u))
\end{pmatrix}.
\end{equation}

Partition the constraints into the active and inactive sets at
$\mathbf u^\star$:
\[
\mathcal I^\star
:=
\{j:\Gamma_j(\mathbf u^\star)=0\},
\qquad
\mathcal I^{\star c}
:=
\{1,\ldots,m_\Gamma\}\setminus\mathcal I^\star.
\]
For $j\in\mathcal I^{\star c}$,
\[
\Gamma_j(\mathbf u^\star)>0,
\qquad
\nu_j^\star=0,
\]
so the inactive multiplier block of
\eqref{eq:appendix_central_path_jacobian}
is nonsingular.  For $j\in\mathcal I^\star$, strict complementarity gives
\[
\nu_j^\star>0.
\]
Eliminate the inactive block, and rescale the active complementarity rows by
$\operatorname{diag}(\nu_{\mathcal I^\star}^\star)^{-1}$.  The remaining
Jacobian then has the same nonsingularity structure as the bordered KKT matrix
\[
\begin{pmatrix}
\nabla_{\mathbf u\mathbf u}^2
\mathcal L_\vartheta(\mathbf u^\star,\nu^\star)
&
D_{\mathbf u}\Gamma_{\mathcal I^\star}
(\mathbf u^\star)^\top
\\[1mm]
D_{\mathbf u}\Gamma_{\mathcal I^\star}
(\mathbf u^\star)
&
0
\end{pmatrix}.
\]
LICQ and the strong second-order sufficient condition therefore imply that
\[
D_z\mathcal F_{\rm cp}(z^\star;0)
\]
is nonsingular.

The Implicit Function Theorem gives a unique local $C^1$ branch
\[
\varepsilon
\longmapsto
z_\varepsilon
=
(\mathbf u_\varepsilon,\nu_\varepsilon)
\]
satisfying
\[
\mathcal F_{\rm cp}(z_\varepsilon;\varepsilon)=0,
\qquad
z_0=z^\star,
\]
for all sufficiently small $\varepsilon\ge0$.  Hence, after reducing
$\varepsilon_0$ if necessary,
\begin{equation}\label{eq:appendix_central_path_rate}
\|z_\varepsilon-z^\star\|
\le
C_{\rm cp}\varepsilon,
\end{equation}
and in particular,
\[
\|\mathbf u_\varepsilon-\mathbf u^\star\|
\le
C_{\rm cp}\varepsilon.
\]
This proves
\eqref{eq:barrier_control_error}.

\medskip\noindent
\textbf{Step 2: strict local maximality and slack bounds.}

The Hessian of the barrier objective at
$\mathbf u_\varepsilon$ can be written as
\begin{equation}\label{eq:appendix_barrier_hessian}
\begin{aligned}
\nabla_{\mathbf u\mathbf u}^2
\Bigg[
\mathbb H_\vartheta(\mathbf u)
+
\varepsilon
\sum_{j=1}^{m_\Gamma}
\log\Gamma_j(\mathbf u)
\Bigg]_{\mathbf u=\mathbf u_\varepsilon}
={}&
\nabla_{\mathbf u\mathbf u}^2
\mathcal L_\vartheta
(\mathbf u_\varepsilon,\nu_\varepsilon)
\\
&-
\sum_{j=1}^{m_\Gamma}
\frac{\varepsilon}
{\Gamma_j(\mathbf u_\varepsilon)^2}
\nabla_{\mathbf u}\Gamma_j(\mathbf u_\varepsilon)
\nabla_{\mathbf u}\Gamma_j(\mathbf u_\varepsilon)^\top.
\end{aligned}
\end{equation}
To make the definiteness argument explicit, let
\[
G_*
:=
D_{\mathbf u}\Gamma_{\mathcal I^\star}(\mathbf u^\star),
\qquad
G_\varepsilon
:=
D_{\mathbf u}\Gamma_{\mathcal I^\star}(\mathbf u_\varepsilon).
\]
By LICQ, every direction $d$ admits a uniformly stable decomposition
$d=d_T+d_N$, where
$d_T\in\ker G_*$ and
$d_N\in\operatorname{range}(G_*^\top)$.  The strong second-order sufficient
condition gives, for some $\kappa_T>0$,
\[
d_T^\top
\nabla_{\mathbf u\mathbf u}^2
\mathcal L_\vartheta(\mathbf u^\star,\nu^\star)
d_T
\le
-\kappa_T\|d_T\|^2.
\]
Since
$(\mathbf u_\varepsilon,\nu_\varepsilon)
\to(\mathbf u^\star,\nu^\star)$, continuity of the Lagrangian Hessian transfers
this estimate, with $\kappa_T/2$, to
$\nabla_{\mathbf u\mathbf u}^2
\mathcal L_\vartheta(\mathbf u_\varepsilon,\nu_\varepsilon)$ on the fixed
subspace $\ker G_*$ for all sufficiently small $\varepsilon$.
Along the central path,
\[
\frac{\varepsilon}{\Gamma_j(\mathbf u_\varepsilon)^2}
=
\frac{\nu_{\varepsilon,j}^2}{\varepsilon}
\ge
\frac{c_\nu}{\varepsilon},
\qquad j\in\mathcal I^\star,
\]
for a uniform $c_\nu>0$.  Continuity of the active-gradient matrix and LICQ
therefore give negative curvature of order $\varepsilon^{-1}$ on the normal
component.  The Lagrangian Hessian is uniformly bounded; its tangent--normal
cross terms are absorbed by Young's inequality, while
$G_\varepsilon d_T=O(\varepsilon)\|d_T\|$ because
$\mathbf u_\varepsilon-\mathbf u^\star=O(\varepsilon)$.  Consequently,
after reducing $\varepsilon_0$, there exist $c_T,c_N>0$ such that
\[
d^\top
\nabla_{\mathbf u\mathbf u}^2
\Bigg[
\mathbb H_\vartheta
+
\varepsilon\sum_j\log\Gamma_j
\Bigg](\mathbf u_\varepsilon)d
\le
-c_T\|d_T\|^2
-
\frac{c_N}{\varepsilon}\|d_N\|^2.
\]
Thus $\mathbf u_\varepsilon$ is a strict local maximizer of the barrier
problem for every sufficiently small $\varepsilon>0$.

For every active constraint $j\in\mathcal I^\star$,
\[
\Gamma_j(\mathbf u_\varepsilon)
=
\frac{\varepsilon}{\nu_{\varepsilon,j}},
\qquad
\nu_{\varepsilon,j}
\longrightarrow
\nu_j^\star>0.
\]
Thus there exists a constant $c_{\rm act}>0$ such that
\[
\Gamma_j(\mathbf u_\varepsilon)
\ge
c_{\rm act}\varepsilon,
\qquad
j\in\mathcal I^\star.
\]
For every inactive constraint $j\in\mathcal I^{\star c}$,
continuity and
$\Gamma_j(\mathbf u^\star)>0$
give a constant $c_{\rm inact}>0$ such that
\[
\Gamma_j(\mathbf u_\varepsilon)
\ge
c_{\rm inact}
\]
for all sufficiently small $\varepsilon$.  Reducing $\varepsilon_0$ and
choosing a common constant $c_{\rm int}>0$ therefore yields
\[
\Gamma_j(\mathbf u_\varepsilon)
\ge
c_{\rm int}\varepsilon,
\qquad
j=1,\ldots,m_\Gamma.
\]
This proves
\eqref{eq:barrier_slack_bound}.

\medskip\noindent
\textbf{Step 3: Lagrangian and Hamiltonian gaps.}

At the true multipliers,
\[
\nabla_{\mathbf u}
\mathcal L_\vartheta
(\mathbf u^\star,\nu^\star)
=
0.
\]
Because the Lagrangian Hessian is locally bounded, Taylor's theorem and
\eqref{eq:appendix_central_path_rate}
give
\[
\left|
\mathcal L_\vartheta
(\mathbf u^\star,\nu^\star)
-
\mathcal L_\vartheta
(\mathbf u_\varepsilon,\nu^\star)
\right|
\le
C_{\mathcal L}\varepsilon^2.
\]
This proves
\eqref{eq:barrier_lagrangian_gap}.

Both $\mathbf u^\star$ and $\mathbf u_\varepsilon$ are feasible for the
original inequalities, and $\mathbf u^\star$ is a strict local constrained
maximizer.  Since $\mathbf u_\varepsilon\to\mathbf u^\star$,
\[
0
\le
\mathbb H_\vartheta(\mathbf u^\star)
-
\mathbb H_\vartheta(\mathbf u_\varepsilon)
\]
for all sufficiently small $\varepsilon$.

KKT complementarity gives
\[
\nu^{\star\top}\Gamma(\mathbf u^\star)=0.
\]
Using
\eqref{eq:appendix_stagewise_lagrangian},
\[
\begin{aligned}
\mathbb H_\vartheta(\mathbf u^\star)
-
\mathbb H_\vartheta(\mathbf u_\varepsilon)
={}&
\mathcal L_\vartheta
(\mathbf u^\star,\nu^\star)
-
\mathcal L_\vartheta
(\mathbf u_\varepsilon,\nu^\star)
\\
&+
\nu^{\star\top}
\Gamma(\mathbf u_\varepsilon).
\end{aligned}
\]
Only active multipliers contribute to the final term.  Therefore,
\[
\nu^{\star\top}\Gamma(\mathbf u_\varepsilon)
=
\sum_{j\in\mathcal I^\star}
\nu_j^\star
\frac{\varepsilon}{\nu_{\varepsilon,j}}
\le
C\varepsilon.
\]
Combining this estimate with the quadratic Lagrangian bound gives
\[
0
\le
\mathbb H_\vartheta(\mathbf u^\star)
-
\mathbb H_\vartheta(\mathbf u_\varepsilon)
\le
C_{H,1}\varepsilon
+
C_{H,2}\varepsilon^2.
\]
This proves
\eqref{eq:barrier_hamiltonian_gap}.

If no inequality is active at $\mathbf u^\star$, then
\[
\nu^\star=\boldsymbol0.
\]
The complementarity term vanishes, and the Hamiltonian gap is bounded solely
by the quadratic Lagrangian term:
\[
0
\le
\mathbb H_\vartheta(\mathbf u^\star)
-
\mathbb H_\vartheta(\mathbf u_\varepsilon)
\le
C\varepsilon^2.
\]
Hence the linear term in
\eqref{eq:barrier_hamiltonian_gap}
vanishes.

\end{proof}

\begin{appremark}[Degeneracy and active-set switching]
\label{rem:barrier_degeneracy}
Strict complementarity is used to obtain the smooth central-path branch and
the $O(\varepsilon_{\rm bar})$ control rate.  Without strict complementarity,
the linear rate need not hold and can deteriorate to
$O(\sqrt{\varepsilon_{\rm bar}})$ in degenerate cases.  At an active-set
switch, a constraint may satisfy
\[
\Gamma_j(\mathbf u^\star)=0,
\qquad
\nu_j^\star=0,
\]
which can destroy strong regularity of a single KKT branch.  The local
estimates in Proposition~\ref{prop:barrier_policy} are therefore not asserted
uniformly across such switching points.
\end{appremark}

\subsection{Estimated-Input B-PGDPO Stability}
\label{app:barrier_estimated_input}
\label{sec:ppgdpo}

Fix a common deployment-state process $\mathbf S$ in a compact deployment
region and a selected regular barrier branch.  Let
$\widehat\vartheta$, $\vartheta^{\rm ref}$, and $\vartheta^*$ denote,
respectively, the estimated, exact-reference, and target adjoint inputs,
all evaluated at the same $(t,\mathbf S_t)$.  Assume
\[
\left\|
(\widehat\lambda,\widehat P)
-
(\lambda^{\rm ref},P^{\rm ref})
\right\|_{L^{q,p}}
\le C_{\rm ol}\delta_{\rm ol},
\qquad
\left\|
\widehat{\boldsymbol\zeta}
-
\boldsymbol\zeta^{\rm ref}
\right\|_{L^{q,p}}
\le\delta_\zeta,
\]
\[
\left\|
\vartheta^{\rm ref}-\vartheta^*
\right\|_{L^{q,p}}
\le\delta_{\rm ref},
\qquad
\delta_{\rm adj}
:=
\delta_{\rm ol}+\delta_\zeta+\delta_{\rm ref}.
\]
Thus, for a constant $C_{\rm in}$,
\begin{equation}\label{eq:barrier_aggregate_adjoint_input_bound}
\|\widehat\vartheta-\vartheta^*\|_{L^{q,p}}
\le C_{\rm in}\delta_{\rm adj}.
\end{equation}
Under the common-state bridge in
Assumption~\ref{ass:common_state_reference_bridge},
$\delta_{\rm ref}$ may be replaced by the square-root
reference-performance rate.

\begin{appassumption}[Uniform barrier-branch regularity]
\label{ass:barrier_recovery_stability}
There exist an open neighborhood $\mathcal N$ of the compact range of
$\vartheta^*$, a closed convex set $\mathcal N_0$ satisfying
\[
\operatorname{range}(\vartheta^*)
\subset
\operatorname{int}(\mathcal N_0),
\qquad
\mathcal N_0\Subset\mathcal N,
\]
a number $\varepsilon_0>0$, a uniform radius $r_{\rm KKT}>0$, and
constants $\underline\nu,\underline\kappa,C_{\rm KKT},C_{\rm res}>0$
such that, for every $\vartheta\in\mathcal N_0$ and
$0<\varepsilon_{\rm bar}\le\varepsilon_0$, the selected KKT solution and
local central path remain on a common regular branch.  The active multipliers
are bounded below by $\underline\nu$, the reduced Hessian is bounded above by
$-\underline\kappa I$ on the feasible tangent space, and the KKT and
perturbed central-path Jacobians have inverse norm at most $C_{\rm KKT}$
in the uniform tube.  Moreover,
\begin{equation}\label{eq:uniform_barrier_residual_error_bound}
\|z-z_{\rm cp}(\varepsilon_{\rm bar};\vartheta)\|
\le
C_{\rm res}
\|\mathcal F_{\rm bar}(z;\vartheta,\varepsilon_{\rm bar})\|
\end{equation}
throughout that tube.  The estimated input is projected onto $\mathcal N_0$,
and the numerical outputs remain in the tube.
\end{appassumption}

For notational compactness, define
\[
\mathbb H_t(\mathbf u;\vartheta)
:=
\mathbb H(t,\mathbf S_t;\mathbf u,\vartheta),
\qquad
\Gamma_t(\mathbf u)
:=
\Gamma(t,\mathbf S_t;\mathbf u).
\]
The inequality-only perturbed KKT map is
\[
\mathcal F_{\rm bar}
(\mathbf u,\nu;\vartheta,\varepsilon)
:=
\begin{pmatrix}
\nabla_{\mathbf u}
\mathbb H_t(\mathbf u;\vartheta)
+
D_{\mathbf u}\Gamma_t(\mathbf u)^\top\nu
\\[1mm]
\operatorname{diag}
(\Gamma_t(\mathbf u))\nu
-
\varepsilon\mathbf1
\end{pmatrix}.
\]
Let
$(\widehat{\mathbf u}_t^{\rm bar},
  \widehat\nu_t^{\rm bar})$
be a strictly feasible numerical approximation to the estimated-input central
path from \eqref{eq:barrier_stage_problem}.  Define
\[
r_t^{\rm bar}
:=
\left\|
\mathcal F_{\rm bar}
(\widehat{\mathbf u}_t^{\rm bar},
 \widehat\nu_t^{\rm bar};
 \widehat\vartheta_t,
 \varepsilon_{\rm bar})
\right\|,
\qquad
\|r^{\rm bar}\|_{L^{q,p}}
\le
\delta_{\rm bar}.
\]
The numerical output satisfies
\[
\Gamma_t(\widehat{\mathbf u}_t^{\rm bar})
>
\boldsymbol0,
\qquad
\widehat\nu_t^{\rm bar}
>
\boldsymbol0.
\]

Let
\[
e_t^{\rm adj}
:=
\|\widehat\vartheta_t-\vartheta_t^*\|,
\]
and let $\mathbf u_t^*$ denote the selected local maximizer corresponding to
$\vartheta_t^*$ at the common deployment state.  Define
\[
\mathbb H_t^*(\mathbf u)
:=
\mathbb H_t(\mathbf u;\vartheta_t^*),
\qquad
G_t^{\rm bar}
:=
\mathbb H_t^*(\mathbf u_t^*)
-
\mathbb H_t^*(\widehat{\mathbf u}_t^{\rm bar}).
\]

\begin{apptheorem}[Estimated-input policy-gap bound for B-PGDPO]
\label{thm:ppgdpo_gap_barrier}
Suppose Assumptions~\ref{ass:barrier_kkt} and
\ref{ass:barrier_recovery_stability} hold uniformly on the deployment region.
For sufficiently small errors, the numerical output remains in the local
maximality neighborhood of $\mathbf u_t^*$.  There exist constants, uniform on
the deployment region, such that
$dt\otimes d\mathbb P$-almost everywhere,
\begin{equation}\label{eq:barrier_pointwise_control_gap}
\|
\widehat{\mathbf u}_t^{\rm bar}
-
\mathbf u_t^*
\|
\le
C_{\rm bar}
\left(
\varepsilon_{\rm bar}
+
e_t^{\rm adj}
+
r_t^{\rm bar}
\right),
\end{equation}
and
\begin{equation}\label{eq:barrier_pointwise_hamiltonian_gap}
\begin{aligned}
0
\le
G_t^{\rm bar}
\le{}&
C_{H,1}\varepsilon_{\rm bar}
+
C_{H,2}
\left(
e_t^{\rm adj}+r_t^{\rm bar}
\right)
\\
&+
C_{H,3}
\left(
\varepsilon_{\rm bar}
+
e_t^{\rm adj}
+
r_t^{\rm bar}
\right)^2.
\end{aligned}
\end{equation}
Consequently,
\begin{equation}\label{eq:barrier_process_control_gap}
\|
\widehat{\mathbf u}^{\rm bar}
-
\mathbf u^*
\|_{L^{q,p}}
\le
C_{\rm bar}'
\left(
\varepsilon_{\rm bar}
+
\delta_{\rm adj}
+
\delta_{\rm bar}
\right).
\end{equation}
If $p,q\ge2$, then
\begin{equation}\label{eq:barrier_process_hamiltonian_gap}
\|G^{\rm bar}\|_{L^{q/2,p/2}}
\le
C_H
\left[
\varepsilon_{\rm bar}
+
\delta_{\rm adj}
+
\delta_{\rm bar}
+
\left(
\varepsilon_{\rm bar}
+
\delta_{\rm adj}
+
\delta_{\rm bar}
\right)^2
\right].
\end{equation}
If no inequality is active at $\mathbf u_t^*$, the linear terms in the
Hamiltonian-gap estimate vanish, and the pointwise Hamiltonian gap is quadratic
in $\varepsilon_{\rm bar}+e_t^{\rm adj}+r_t^{\rm bar}$.  The control-distance
bound \eqref{eq:barrier_pointwise_control_gap} remains first order.
\end{apptheorem}

\begin{appcorollary}[Consistency of B-PGDPO]
\label{cor:ppgdpo_consistency}
If
\[
\delta_{\rm ol},
\quad
\delta_\zeta,
\quad
\delta_{\rm ref},
\quad
\delta_{\rm bar},
\quad
\varepsilon_{\rm bar}
\longrightarrow0,
\]
then
\[
\widehat{\mathbf u}^{\rm bar}
\longrightarrow
\mathbf u^*
\qquad
\text{in }L^{q,p}.
\]
For $p,q\ge2$, the true local Hamiltonian-gap process converges to zero in
$L^{q/2,p/2}$.  For $p=q=2$, under
Assumption~\ref{ass:common_state_reference_bridge}, the term $\delta_{\rm ref}$
may be replaced by the square-root reference-performance rate.
\end{appcorollary}

\label{app:proof_ppgdpo}

Fix a deployment point $(t,\omega)$ and suppress it throughout.  Write
\[
\mathbb H_\vartheta(\mathbf u)
:=
\mathbb H
(t,\mathbf S_t;\mathbf u,\vartheta),
\qquad
\Gamma(\mathbf u)
:=
\Gamma(t,\mathbf S_t;\mathbf u).
\]
For
\[
z
=
(\mathbf u,\nu),
\]
define the inequality-only perturbed KKT map
\begin{equation}\label{eq:appendix_estimated_barrier_map}
\mathcal F_{\rm bar}
(z;\vartheta,\varepsilon)
:=
\begin{pmatrix}
\nabla_{\mathbf u}
\mathbb H_\vartheta(\mathbf u)
+
D_{\mathbf u}\Gamma(\mathbf u)^\top\nu
\\[1mm]
\operatorname{diag}
(\Gamma(\mathbf u))\nu
-
\varepsilon\mathbf1
\end{pmatrix}.
\end{equation}
Let
\[
z_{\rm cp}(\varepsilon;\vartheta)
=
\left(
\mathbf u_{\rm cp}(\varepsilon;\vartheta),
\nu_{\rm cp}(\varepsilon;\vartheta)
\right)
\]
denote the exact central-path solution on the selected regular branch, and
let
\[
\widehat z^{\rm bar}
=
\left(
\widehat{\mathbf u}^{\rm bar},
\widehat\nu^{\rm bar}
\right)
\]
denote the numerical primal--dual output.  The true KKT pair is denoted by
\[
z^*
=
(\mathbf u^*,\nu^*).
\]

\begin{proof}[Proof of Theorem~\ref{thm:ppgdpo_gap_barrier}]

\textbf{Step 1: numerical residual and central-path sensitivity.}

By definition,
\[
\left\|
\mathcal F_{\rm bar}
\left(
\widehat z^{\rm bar};
\widehat\vartheta,
\varepsilon_{\rm bar}
\right)
\right\|
=
r^{\rm bar}.
\]
By construction, the projected input $\widehat\vartheta$ belongs to
$\mathcal N_0$, and the numerical output lies in the uniform tube from
Assumption~\ref{ass:barrier_recovery_stability}.  Applying the residual
error bound \eqref{eq:uniform_barrier_residual_error_bound} therefore gives
\begin{equation}\label{eq:appendix_barrier_residual_bound}
\left\|
\widehat z^{\rm bar}
-
z_{\rm cp}
(\varepsilon_{\rm bar};\widehat\vartheta)
\right\|
\le
C_{\rm res}r^{\rm bar}.
\end{equation}

For the input sensitivity, set
\[
\vartheta_s
:=
(1-s)\vartheta^*+s\widehat\vartheta,
\qquad 0\le s\le1.
\]
The inclusion
$\operatorname{range}(\vartheta^*)\subset\operatorname{int}(\mathcal N_0)$,
convexity of $\mathcal N_0$, and projection of $\widehat\vartheta$ ensure that
the entire segment lies in $\mathcal N_0$.  Differentiating
\[
\mathcal F_{\rm bar}
\left(
z_{\rm cp}(\varepsilon_{\rm bar};\vartheta_s);
\vartheta_s,
\varepsilon_{\rm bar}
\right)
=
0
\]
with respect to $s$ gives
\[
\frac{d}{ds}z_{\rm cp}
(\varepsilon_{\rm bar};\vartheta_s)
=
-
\left[D_z\mathcal F_{\rm bar}\right]^{-1}
D_\vartheta\mathcal F_{\rm bar}
(\widehat\vartheta-\vartheta^*).
\]
The inverse-Jacobian and input-derivative bounds are uniform on the tube, so
integration over $s\in[0,1]$ yields
\begin{equation}\label{eq:appendix_barrier_input_sensitivity}
\left\|
z_{\rm cp}
(\varepsilon_{\rm bar};\widehat\vartheta)
-
z_{\rm cp}
(\varepsilon_{\rm bar};\vartheta^*)
\right\|
\le
C_{\rm sens}
\|\widehat\vartheta-\vartheta^*\|
=
C_{\rm sens}e^{\rm adj}.
\end{equation}
Finally, Proposition~\ref{prop:barrier_policy} gives
\begin{equation}\label{eq:appendix_barrier_central_path_bias}
\left\|
\mathbf u_{\rm cp}
(\varepsilon_{\rm bar};\vartheta^*)
-
\mathbf u^*
\right\|
\le
C_{\rm cp}\varepsilon_{\rm bar}.
\end{equation}

Combining
\eqref{eq:appendix_barrier_residual_bound}--
\eqref{eq:appendix_barrier_central_path_bias}
and retaining the primal components yields
\[
\left\|
\widehat{\mathbf u}^{\rm bar}
-
\mathbf u^*
\right\|
\le
C_{\rm bar}
\left(
\varepsilon_{\rm bar}
+
e^{\rm adj}
+
r^{\rm bar}
\right).
\]
This proves
\eqref{eq:barrier_pointwise_control_gap}.

Taking the $L^{q,p}$ norm and using
\eqref{eq:barrier_aggregate_adjoint_input_bound}
together with
\[
\|r^{\rm bar}\|_{L^{q,p}}
\le
\delta_{\rm bar}
\]
gives
\[
\left\|
\widehat{\mathbf u}^{\rm bar}
-
\mathbf u^*
\right\|_{L^{q,p}}
\le
C_{\rm bar}'
\left(
\varepsilon_{\rm bar}
+
\delta_{\rm adj}
+
\delta_{\rm bar}
\right),
\]
which is
\eqref{eq:barrier_process_control_gap}.

\medskip\noindent
\textbf{Step 2: true-input Hamiltonian gap.}

Define the true-input Hamiltonian and Lagrangian by
\[
\mathbb H^*(\mathbf u)
:=
\mathbb H_{\vartheta^*}(\mathbf u),
\qquad
\mathcal L^*(\mathbf u)
:=
\mathbb H^*(\mathbf u)
+
\nu^{*\top}\Gamma(\mathbf u).
\]
The KKT stationarity condition gives
\[
\nabla_{\mathbf u}\mathcal L^*(\mathbf u^*)
=
0.
\]
Since
$\widehat{\mathbf u}^{\rm bar}$
is feasible and remains in the local maximality neighborhood,
\[
G^{\rm bar}
=
\mathbb H^*(\mathbf u^*)
-
\mathbb H^*
(\widehat{\mathbf u}^{\rm bar})
\ge
0.
\]

Local boundedness of the Lagrangian Hessian and Taylor's theorem give
\begin{equation}\label{eq:appendix_true_lagrangian_gap}
\left|
\mathcal L^*(\mathbf u^*)
-
\mathcal L^*
(\widehat{\mathbf u}^{\rm bar})
\right|
\le
C
\left\|
\widehat{\mathbf u}^{\rm bar}
-
\mathbf u^*
\right\|^2.
\end{equation}
By complementarity at the true KKT point,
\[
\nu^{*\top}\Gamma(\mathbf u^*)=0,
\]
and therefore
\begin{equation}\label{eq:appendix_true_hamiltonian_decomposition}
\begin{aligned}
G^{\rm bar}
={}&
\mathcal L^*(\mathbf u^*)
-
\mathcal L^*
(\widehat{\mathbf u}^{\rm bar})
\\
&+
\nu^{*\top}
\Gamma(\widehat{\mathbf u}^{\rm bar}).
\end{aligned}
\end{equation}

Set
\[
\mathbf u_{\rm cp}^*
:=
\mathbf u_{\rm cp}
(\varepsilon_{\rm bar};\vartheta^*).
\]
For every active constraint
$j\in\mathcal I^*$,
the true-input central path satisfies
\[
\Gamma_j(\mathbf u_{\rm cp}^*)
=
\frac{\varepsilon_{\rm bar}}
{\nu_{{\rm cp},j}
(\varepsilon_{\rm bar};\vartheta^*)}.
\]
Strict complementarity and the uniform localization imply that the
denominators are bounded away from zero.  Since inactive true multipliers
vanish,
\begin{equation}\label{eq:appendix_true_complementarity_bias}
\nu^{*\top}\Gamma(\mathbf u_{\rm cp}^*)
\le
C\varepsilon_{\rm bar}.
\end{equation}

The residual and input-sensitivity estimates give
\[
\left\|
\widehat{\mathbf u}^{\rm bar}
-
\mathbf u_{\rm cp}^*
\right\|
\le
C
\left(
e^{\rm adj}
+
r^{\rm bar}
\right).
\]
Using local Lipschitz continuity of $\Gamma$ and boundedness of $\nu^*$,
\begin{equation}\label{eq:appendix_numerical_complementarity_bound}
\nu^{*\top}
\Gamma(\widehat{\mathbf u}^{\rm bar})
\le
C
\left(
\varepsilon_{\rm bar}
+
e^{\rm adj}
+
r^{\rm bar}
\right).
\end{equation}

Combining
\eqref{eq:appendix_true_lagrangian_gap},
\eqref{eq:appendix_true_hamiltonian_decomposition},
\eqref{eq:appendix_numerical_complementarity_bound}, and
\eqref{eq:barrier_pointwise_control_gap}
yields
\[
\begin{aligned}
0
\le
G^{\rm bar}
\le{}&
C_{H,1}\varepsilon_{\rm bar}
+
C_{H,2}
\left(
e^{\rm adj}
+
r^{\rm bar}
\right)
\\
&+
C_{H,3}
\left(
\varepsilon_{\rm bar}
+
e^{\rm adj}
+
r^{\rm bar}
\right)^2.
\end{aligned}
\]
This proves
\eqref{eq:barrier_pointwise_hamiltonian_gap}.

\medskip\noindent
\textbf{Step 3: process-level Hamiltonian bound.}

Suppose $p,q\ge2$.  On the finite measure space
$[0,T]\times\Omega$,
\[
L^{q,p}
\hookrightarrow
L^{q/2,p/2}.
\]
Moreover,
\[
\bigl\||f|^2\bigr\|_{L^{q/2,p/2}}
=
\|f\|_{L^{q,p}}^2.
\]
Minkowski's inequality,
\eqref{eq:barrier_aggregate_adjoint_input_bound}, and the numerical-residual bound
therefore give
\[
\begin{aligned}
\|G^{\rm bar}\|_{L^{q/2,p/2}}
\le
C_H
\Big[
&
\varepsilon_{\rm bar}
+
\delta_{\rm adj}
+
\delta_{\rm bar}
\\
&
+
\left(
\varepsilon_{\rm bar}
+
\delta_{\rm adj}
+
\delta_{\rm bar}
\right)^2
\Big],
\end{aligned}
\]
which is
\eqref{eq:barrier_process_hamiltonian_gap}.

\medskip\noindent
\textbf{Step 4: inactive-constraint case.}

If no inequality is active at $\mathbf u^*$, then
\[
\nu^*
=
\boldsymbol0.
\]
The point $\mathbf u^*$ is an interior stationary point of
$\mathbb H^*$, so
\[
\nabla_{\mathbf u}\mathbb H^*(\mathbf u^*)
=
0.
\]
Consequently,
\[
0
\le
G^{\rm bar}
\le
C
\left\|
\widehat{\mathbf u}^{\rm bar}
-
\mathbf u^*
\right\|^2
\le
C
\left(
\varepsilon_{\rm bar}
+
e^{\rm adj}
+
r^{\rm bar}
\right)^2.
\]
Thus all linear terms vanish, and the corresponding process-level bound is
quadratic.

\end{proof}

\section{Numerical Solver and Implementation Details}
\label{app:solver_details}

This appendix records the numerical realizations used for B-PGDPO and
QP-PGDPO.  The analysis in Appendices~\ref{app:proof_barrier_policy} and
\ref{app:end_to_end_recovery_proofs} is stated in reduced control coordinates
\[
\mathbf u
=
(\boldsymbol\pi^\top,C)^\top,
\qquad
\pi_0
=
1-\mathbf1^\top\boldsymbol\pi.
\]
Thus no explicit full-investment equality or equality multiplier is used
below.  When no borrowing is imposed, the inequality
$1-\mathbf1^\top\boldsymbol\pi\ge0$
is included as a component of $\Gamma$.

\subsection{Anchored Shifted Wealth-Row Estimation}
\label{app:shifted_row_estimator}

This subsection records the nested antithetic construction summarized in
Section~\ref{sec:shifted_input_estimation}.  Fix a diagnostic state
$\mathbf S_k=\mathbf s$ and reference action $\mathbf u_k^{\rm ref}$.  The
implementation records an even number $M_{\rm out}$ of signed outer paths.
For $i=1,\ldots,M_{\rm out}/2$, draw
\[
d_i\sim N(0,hI_{d_W})
\]
and set $d_i^+:=d_i$ and $d_i^-:=-d_i$.  The index $i$ therefore labels one
of the $M_{\rm out}/2$ independent antithetic pairs, while
$\varsigma\in\{+,-\}$ labels one of its two current-increment branches.  For
each $j=1,\ldots,M_{\rm in}$, draw a future
Brownian continuation after $t_{k+1}$ and reuse that same continuation for
the $+$ and $-$ branches.  We write
$\widetilde\lambda_{k+1}^{X,[i,\varsigma,j]}$ for the raw fixed-latent wealth
adjoint obtained by OL-BPTT from the next state generated with
$d_i^\varsigma$ and the $j$th future continuation.  The bracketed tuple is
only a simulation label, not an additional state component or derivative
index.  The two branchwise nested averages and their odd response are
\begin{equation}\label{eq:nested_next_adjoint_average}
\begin{aligned}
\bar\lambda_{i,\varsigma}^X
&:=
\frac1{M_{\rm in}}
\sum_{j=1}^{M_{\rm in}}
\widetilde\lambda_{k+1}^{X,[i,\varsigma,j]},
&&
\varsigma\in\{+,-\},
\\
y_i
&:=
\frac{\bar\lambda_{i,+}^X-\bar\lambda_{i,-}^X}{2}.
\end{aligned}
\end{equation}
Let $D_k$ be the matrix whose $i$th row is $d_i^\top$, and let $y$ collect
the responses $y_i$.  The columnized structural anchor is
\begin{equation}\label{eq:shift_anchor_definition}
c_k^{X,\rm anc}
:=
\left[
e_X^\top\widehat P_k
\boldsymbol\sigma_S
(t_k,\mathbf S_k;\mathbf u_k^{\rm ref})
\right]^\top
\in\mathbb R^{d_W}.
\end{equation}
Thus $y-D_kc_k^{X,\rm anc}$ is the response used in the residual regression
\eqref{eq:anchored_shift_regression}.  When $\kappa_Z=0$, $D_k$ has full
column rank, and $\widehat P_k$ is fixed, this regression is algebraically
equivalent to estimating $\mathbf Z^X$ by ordinary least squares (OLS) and
then subtracting the anchor.  Nested
conditional averaging reduces future-shock noise, and common random numbers
promote cancellation of the paired continuation errors.  Antithetic
differencing cancels the constant and even-order responses to the current
Brownian increment, leaving its antisymmetric response, whose leading term
under the local smooth expansion is the linear Brownian slope.

\subsection{Barrier Newton--CG Realization}
\label{app:barrier_newton_cg}

Fix $(t,\mathbf s,\vartheta)$ and define the barrier objective
\begin{equation}\label{eq:appendix_barrier_objective}
\mathcal B_\varepsilon(\mathbf u;\vartheta)
:=
\mathbb H(t,\mathbf s;\mathbf u,\vartheta)
+
\varepsilon
\sum_{j=1}^{m_\Gamma}
\log\Gamma_j(t,\mathbf s;\mathbf u).
\end{equation}
At a strictly feasible iterate, let
\[
\mathbf g_{\rm bar}
:=
\nabla_{\mathbf u}
\mathcal B_\varepsilon(\mathbf u;\vartheta).
\]
The Hessian--vector operator is
\begin{equation}\label{eq:hvp_general}
\begin{aligned}
\mathcal B_{\rm bar}[v]
:={}&
\nabla_{\mathbf u\mathbf u}^2
\mathbb H(t,\mathbf s;\mathbf u,\vartheta)v
\\
&+
\varepsilon
\sum_{j=1}^{m_\Gamma}
\left[
\frac{
\nabla_{\mathbf u\mathbf u}^2
\Gamma_j(t,\mathbf s;\mathbf u)v
}{
\Gamma_j(t,\mathbf s;\mathbf u)
}
-
\frac{
\nabla_{\mathbf u}\Gamma_j(t,\mathbf s;\mathbf u)
\nabla_{\mathbf u}\Gamma_j(t,\mathbf s;\mathbf u)^\top v
}{
\Gamma_j(t,\mathbf s;\mathbf u)^2
}
\right].
\end{aligned}
\end{equation}
A Newton ascent direction satisfies
\begin{equation}\label{eq:appendix_newton_direction}
\mathcal B_{\rm bar}[d]
=
-\mathbf g_{\rm bar}.
\end{equation}
On the selected local maximal branch,
$-\mathcal B_{\rm bar}$ is positive definite for sufficiently small barrier
parameter.  We therefore solve the equivalent positive-definite system
\[
-\mathcal B_{\rm bar}[d]
=
\mathbf g_{\rm bar}
\]
by conjugate gradients, using Hessian--vector products rather than forming a
dense Hessian.

For a nonnegativity component
$\Gamma_i(\mathbf u)=\pi_i$, the corresponding portfolio contribution is
\begin{equation}\label{eq:hvp_nonneg}
-\frac{\varepsilon}{\pi_i^2}v_i.
\end{equation}
Affine constraints have no second-derivative term; their barrier curvature is
therefore the negative rank-one term in
\eqref{eq:hvp_general}.  If portfolio and consumption constraints are
coupled, the same system is formed in the complete control coordinates
$(\boldsymbol\pi,C)$.

\subsection{Strict Feasibility and Globalization}
\label{app:barrier_globalization}

Let $d$ be an ascent direction.  For affine inequalities, the
fraction-to-the-boundary predictor is
\begin{equation}\label{eq:general_fraction_boundary}
\alpha_{\max}
=
\min\left\{
1,
\min_{j:\,D\Gamma_j(\mathbf u)[d]<0}
-
\tau_{\rm fb}
\frac{\Gamma_j(\mathbf u)}{D\Gamma_j(\mathbf u)[d]}
\right\},
\qquad
\tau_{\rm fb}\in(0,1).
\end{equation}
For nonlinear constraints, this predictor is shortened until the exact slacks
at $\mathbf u+\alpha d$ are all positive.  Starting from $\alpha_{\max}$,
backtracking decreases the step until the Armijo sufficient-increase condition
\begin{equation}\label{eq:appendix_barrier_armijo}
\mathcal B_\varepsilon
(\mathbf u+\alpha d;\vartheta)
\ge
\mathcal B_\varepsilon
(\mathbf u;\vartheta)
+
c_{\rm A}\alpha
\mathbf g_{\rm bar}^\top d,
\qquad
c_{\rm A}\in(0,1),
\end{equation}
holds.  The accepted step therefore preserves strict interiority and yields
monotone ascent of the barrier objective.

The reported barrier runs start from a model-specific strict interior point
obtained by the analytic-center or phase-I initialization used for that
constraint set.  They terminate according to the squared Newton decrement
\begin{equation}\label{eq:appendix_squared_newton_decrement}
\delta_{\rm N}^2
:=
\mathbf g_{\rm bar}^\top d,
\end{equation}
with a prescribed numerical tolerance.  The perturbed KKT residual
$r^{\rm bar}$ from Section~\ref{sec:ppgdpo} is evaluated separately as a
solver audit.

When the consumption and portfolio blocks are separable, the unregularized
consumption maximizer is the solution of
\[
\max_{C\in[C_{\min}(t,\mathbf s),C_{\max}(t,\mathbf s)]}
\left\{e^{-\rho t}U(C)-\lambda^X C\right\}.
\]
When $e^{\rho t}\lambda^X$ lies in the range of $U'$, this solution has the
clipping form
\begin{equation}\label{eq:appendix_consumption_clipping}
C^*
=
\min\left\{
\max\left\{
(U')^{-1}(e^{\rho t}\lambda^X),
C_{\min}(t,\mathbf s)
\right\},
C_{\max}(t,\mathbf s)
\right\}.
\end{equation}
Otherwise the scalar KKT condition selects the appropriate endpoint; no
positivity assumption on $\lambda^X$ is needed for concavity or KKT
sufficiency itself.
Under a barrier realization, the scalar stationarity equation is
\[
e^{-\rho t}U'(C)
-
\lambda^X
+
\frac{\varepsilon_{
\rm bar}}{C-C_{\min}}
-
\frac{\varepsilon_{
\rm bar}}{C_{\max}-C}
=
0,
\]
and may be solved by the same damped Newton and line-search logic.  These
displays use the calendar normalization of Sections~\ref{sec:model_setup}--
\ref{sec:PMP_multiasset}.  In the reported $\rho>0$ implementation, the same
stage problem is written in current-value units: $e^{-\rho t}U$ and
$\lambda^X$ are replaced by $U$ and the current-value adjoint, and
$\varepsilon_{\rm bar}$ is quoted in the correspondingly scaled current-value
units.

\subsection{QP Realization and Residual Audit}
\label{app:qp_solver_details}

For fixed estimated input, QP-PGDPO solves the strongly concave risky-weight
problem in \eqref{eq:exact_qp_recovery}.  Equivalently, it minimizes
\[
\frac12
\boldsymbol\pi^\top
\widehat Q
\boldsymbol\pi
-
\widehat{\mathbf g}^\top
\boldsymbol\pi
\]
over the convex feasible set $\mathcal K$.  When
$\widehat Q\succ0$, the solution is unique and does not depend on a supplied
initialization.  Numerical accuracy is measured by the normal-cone residual
$r^{\rm qp}$ defined in Section~\ref{sec:qp_recovery}.  The implementation is
a direct constrained QP solve; it is not batched or otherwise optimized in the
reported timing comparison.

\section{Robustness and Supplementary Numerical Diagnostics}
\label{app:numerical_diagnostics}

Unless otherwise stated, all adjoint-dependent diagnostics below use the
fixed-latent \OLBPTT graph and the anchored shifted-input estimator.  The
B-PGDPO results use the barrier globalization and squared-Newton-decrement
stopping rule described in Online Supporting Information~\ref{app:solver_details}.  DPO,
B-PGDPO, and QP-PGDPO are compared at common diagnostic states and, when
simulation is required, with common indexed Brownian paths.

\subsection{Calibration, Diagnostic Definitions, and Sampling Protocol}
\label{app:diagnostic_protocol}

Table~\ref{tab:market_calibration_summary} records the high-dimensional market
and constraint calibrations.  Each coefficient generator is deterministic at
the displayed seed, and the same realized coefficients are reused across
methods, budgets, and held-out-state splits.  In particular, expected returns,
volatilities, and correlations vary across assets, so equal weights are not an
analytical benchmark in the $n=100$ experiments.

\begin{table}[!htbp]
\centering
\scriptsize
\setlength{\tabcolsep}{3pt}
\renewcommand{\arraystretch}{1.12}
\caption{High-dimensional calibration and coefficient generation.  The
one-factor analytical validation has its separate closed-form calibration in
Online Supporting Information~\ref{app:unconstrained_predictable_return_validation}.}
\label{tab:market_calibration_summary}
\begin{tabularx}{\textwidth}{@{}>{\raggedright\arraybackslash}p{0.19\textwidth}
>{\raggedright\arraybackslash}p{0.29\textwidth}X@{}}
\toprule
Block & Dynamics and constraints & Coefficient construction \\
\midrule
Constant opportunity, no short
& $n=10,100$; $r=0.03$, $\rho=0$, $\gamma=2$, $K=1$,
  $T=1.5$, $N=20$; $\boldsymbol\pi\ge\boldsymbol0$ with borrowing allowed;
  $X\in[0.1,3]$.
& Seed 42.  The generator starts with volatilities drawn from
  $[0.18,0.28]$, a $0.50$ equicorrelation matrix plus symmetric Gaussian
  jitter of standard deviation $0.06$ projected to the positive-definite cone,
  and excess returns drawn from $[0.05,0.22]$, with one quarter replaced by
  draws from $[-0.03,0.02]$.  It then adjusts the excess-return vector to
  produce a nonsymmetric CRRA target with at least one negative unconstrained
  component and total risky mass in $[0.40,0.75]$; a covariance ridge of
  $2\times10^{-3}$ times the mean variance is added. \\
\addlinespace
Constant opportunity with consumption cap
& $n=2,10,100$; $r=0.03$, $\rho=0.1$, $\gamma=2$, $K=1$,
  $T=1$; $N=16$ for $n=2,10$ and $N=12$ for $n=100$;
  $\boldsymbol\pi\ge\boldsymbol0$, $\mathbf1^\top\boldsymbol\pi\le1$, and
  $10^{-8}X\le C\le0.7X$; $X\in[0.3,3]$.
& Seed 42.  Asset volatilities and excess returns are evenly spaced over
  $[0.18,0.30]$ and $[0.015,0.075]$, respectively.  The correlation matrix
  has off-diagonal baseline $0.25$ plus symmetric Gaussian jitter of standard
  deviation $0.035$, followed by positive-definite projection; the covariance
  ridge is $10^{-4}$ times the mean variance. \\
\addlinespace
Affine factor, simplex and cap
& $n=10,50,100$, one factor; $r=0.03$, $\gamma=2$, $K=1$,
  $T=1$; $N=16$ for $n=10,50$ and $N=12$ for $n=100$;
  $\boldsymbol\pi\ge\boldsymbol0$, $\mathbf1^\top\boldsymbol\pi\le1$.
  The no-consumption runs use $\rho=0$; cap runs use $\rho=0.1$ and
  $10^{-8}X\le C\le0.7X$.
& Seed 11.  Asset volatilities and return intercepts are evenly spaced over
  $[0.18,0.30]$ and $[0.015,0.075]$.  The factor loading of expected returns
  alternates in sign with magnitude evenly spaced over $[0.05,0.22]$.
  The factor parameters are $\kappa=0.7$, $\theta=0$, and volatility $0.12$.
  The joint asset--factor correlation starts from asset baseline $0.25$ with
  Gaussian jitter $0.035$ and asset--factor Gaussian draws of standard
  deviation $0.20$, followed by positive-definite projection and the same
  $10^{-4}$ covariance ridge. \\
\bottomrule
\end{tabularx}
\end{table}

\paragraph{Covariance/loading and curvature audit.}
The simulator and the local QP use the same covariance/loading convention.
For the constant-opportunity generators, the Brownian loading is a factor of
the ridged covariance.  For the affine-factor generators, the
factor-orthogonal asset-noise component is recolored so that
$\boldsymbol\Sigma=\boldsymbol v\boldsymbol v^\top$ while the
asset--factor cross-covariance $\boldsymbol\Sigma_{RY}$ is preserved.
In a 23-job covariance/loading-consistent revalidation spanning the
reported primary and supplementary recovery configurations, the smallest
pointwise values were
$\min(-\widehat P^{XX})=6.18\times10^{-2}$ and
$\min\lambda_{\min}(\widehat Q)=9.07\times10^{-3}$.  The configured
$10^{-10}$ curvature correction was never activated; every full-input QP
converged and no fallback was used.

The lower bound $10^{-8}X$ in the cap rows is a numerical floor used to avoid
evaluating CRRA utility at zero; the theoretical feasible lower bound remains
$C\ge0$.  The learned actors use two hidden layers of width 200 with LeakyReLU
activations.  As stated in Section~\ref{sec:feedback_fixed_latent}, the training
implementation uses 500 PyTorch-Adam epochs at learning rate $10^{-4}$, except
for the explicit budget experiment; one epoch is one freshly resampled Monte
Carlo minibatch followed by one Adam update, with the policy-gradient norm
clipped at one.

The default evaluation seed is 12345.  For a diagnostic state
$\mathbf s_i$, graph $g$, and Brownian continuation $\omega_{im}$, the
conditional Monte Carlo estimate of a raw pathwise adjoint quantity is
\begin{equation}\label{eq:appendix_conditional_mc_average}
\widehat A_i^g
=
\frac1M
\sum_{m=1}^M
\widetilde A^g(\mathbf s_i,\omega_{im}).
\end{equation}
Method comparisons use common diagnostic states and indexed Brownian paths.
Shifted-input estimation and its replication audit use independent
outer--inner Brownian banks.  The graph audit reports
\begin{equation}\label{eq:appendix_graph_nrmse}
d_{\rm nRMSE}(A)
=
\frac{
\left[
N_s^{-1}
\sum_i
\|\widehat A_i^{\rm ol}-\widehat A_i^{\rm fb}\|_2^2
\right]^{1/2}
}{
\left[
N_s^{-1}
\sum_i
\|\widehat A_i^{\rm ol}\|_2^2
\right]^{1/2}
}.
\end{equation}
This quantity measures a discrepancy between two conditional estimates, not
an error relative to analytical truth.  State and continuation sub-batches
only partition the same fixed estimator.

For a $d_\pi$-dimensional risky-weight vector, define
\[
{\rm RMSE}_{\pi}
=
\left[
\frac1{N_{\rm test}d_\pi}
\sum_{i=1}^{N_{\rm test}}
\|\widehat{\boldsymbol\pi}_i
-\boldsymbol\pi_i^{\rm ref}\|_2^2
\right]^{1/2}.
\]
At state $i$, let
\[
\mathcal Q_i(\boldsymbol\pi)
=
\mathbf g_i^\top\boldsymbol\pi
-
\frac12
\boldsymbol\pi^\top Q_i\boldsymbol\pi,
\qquad
s_i^{\rm pg}
=
\lambda_{\max}(Q_i)^{-1}.
\]
Here $\lambda_{\max}(Q_i)$ is the largest eigenvalue of $Q_i$, and
$\Pi_{\mathcal K_i}$ denotes Euclidean projection onto $\mathcal K_i$.
The projected-gradient residual is
\[
\mathcal G_i^{\rm pg}(\boldsymbol\pi)
=
\frac{
\boldsymbol\pi
-
\Pi_{\mathcal K_i}
\left\{
\boldsymbol\pi
+s_i^{\rm pg}
(\mathbf g_i-Q_i\boldsymbol\pi)
\right\}
}{s_i^{\rm pg}},
\qquad
{\rm KKT\mbox{-}PG}
=
\frac1{N_{\rm test}}
\sum_i
\|\mathcal G_i^{\rm pg}(\boldsymbol\pi_i)\|_\infty.
\]
Here
$\mathcal K_i=\mathbb R_+^n$
for borrowing-allowed no-short-sale recovery, whereas no borrowing adds
$\mathbf1^\top\boldsymbol\pi\le1$.  B-PGDPO is evaluated with the
unregularized generalized Hamiltonian.  For consumption, the numerical
diagnostics use the current-value normalization described in
Section~\ref{sec:num_test}, so the corresponding scalar gradient is
\[
g_i^C
=
U'(C_i)
-
\widehat\lambda_i^X,
\]
with projection onto the admissible consumption interval.  Multiplying both
terms by $e^{-\rho t_i}$ gives the equivalent calendar-value gradient.

The statewise shifted-input signal-to-noise ratio is the norm of
$\widehat{\boldsymbol\zeta}_i^X$ divided by the norm of its estimated Monte
Carlo standard-error vector.  The reported SNR is the average of these
statewise ratios, not the ratio of the corresponding sample means.  It is a
Monte Carlo resolution diagnostic rather than a formal hypothesis-test
statistic.  Replication error is the $L^2$ distance between estimates obtained
from the estimation and independent holdout banks.

Independent and identically distributed (IID) states are independent draws
from the training law.  The no-short-sale
run uses $X\in[0.1,3.0]$ and $T-t\in[0,1.5]$; the cap and factor runs use
$X\in[0.3,3.0]$ and $T-t\in[0,1]$.  The default scalar-factor range is
\[
Y
\in
\left[
\bar y
-2.5\frac{\sigma_Y}{\sqrt{2\kappa}},
\bar y
+2.5\frac{\sigma_Y}{\sqrt{2\kappa}}
\right].
\]
For a wealth-edge split, half of the observations (up to one observation
when $N_s$ is odd) are drawn uniformly from the bottom $5\%$ of the configured
wealth interval and the remainder from the top $5\%$, while the other state
coordinates retain their canonical draws.  A factor-edge split retains the
canonical wealth and time draws and assigns the factor coordinates in the
alternating pattern $-1.5,+1.5$.  In a stress split, wealth alternates between
the two configured endpoints and $\tau$ alternates in two-observation blocks
between $\tau_{\min}=\max\{T/100,10^{-4}\}$ and $T$; when factors are present
they again alternate between $-1.5$ and $+1.5$.  The factor-edge and stress
samples are deliberate out-of-distribution diagnostics, not a uniform
extrapolation guarantee.

\subsection{Shifted-Input Estimator Ablation}
\label{app:z_estimator_ablation}

Table~\ref{tab:z_estimator_ablation} compares the candidate
Brownian-coefficient estimators under analytical policies.

\begin{table}[!htbp]
\centering
\small
\caption{Brownian-coefficient estimation under analytical policies.  The
selected estimator is the anchored nested OLS control-variate estimator used
throughout the reported recovery results.  The entry ``n/a'' indicates that
no analytical shifted-adjoint RMSE target is reported for that diagnostic.}
\label{tab:z_estimator_ablation}
\begin{tabularx}{\textwidth}{@{}Xccc@{}}
\toprule
Model & \shortstack{Raw-moment\\$\mathbf Z^X$ nRMSE}
& \shortstack{Anchored OLS\\$\mathbf Z^X$ nRMSE}
& \shortstack{Anchored\\$\boldsymbol\zeta^X$ RMSE} \\
\midrule
Merton exact & $59.95\%$ & $0.326\%$ & $2.30\times10^{-4}$ \\
Constant-opportunity cap exact & $43.87\%$ & $1.313\%$ & $6.16\times10^{-4}$ \\
Affine exact canonical & $8.812\%$ & $0.260\%$ & \text{n/a} \\
\bottomrule
\end{tabularx}
\end{table}

Raw and merely centered covariance estimates remain noisy because the
next-step pathwise adjoint contains substantial future-shock variation.
Nested conditional averaging, antithetic common random numbers, and the
anchor
\[
e_X^\top P\,
\boldsymbol\sigma_S
\]
reduce $\mathbf Z^X$ error by one to two orders of magnitude at the same
analytical targets.  This is why the anchored estimator is used in all
reported recovery tables.

\subsection{Unconstrained Predictable-Return Adjoint Validation}
\label{app:unconstrained_predictable_return_validation}

The main-text predictable-return experiment emphasizes the constrained
recovery contribution.  For completeness, this subsection retains the
unconstrained analytical factor validation.  The calculation is also a
continuity check with the factor-model experiments of
\citet{huh2025breaking}.  We use a one-factor affine special case of
\citet{liu2007portfolio}:
\[
 dY_t=\kappa(\theta-Y_t)dt+\nu\,dB_t^Y,
 \qquad
 \frac{d\mathcal P_t}{\mathcal P_t}
 =\{r+a_0+a_1Y_t\}dt+\sigma\,dB_t^R,
 \qquad
 d\langle B^R,B^Y\rangle_t=\varrho\,dt.
\]
Here $\mathcal P_t$ denotes the risky-asset price.  With terminal CRRA utility
and $\tau=T-t$, the exact value is
\[
 V(t,x,y)=\frac{x^{1-\gamma}}{1-\gamma}
 \exp\!\left\{A(\tau)+B(\tau)y+\frac12D(\tau)y^2\right\},
\]
where $(A,B,D)$ solve the associated Riccati system.  Writing
$\Phi=A+By+\tfrac12Dy^2$, the analytical targets are
\begin{equation}\label{eq:liu_analytic_mixed_targets}
 V_X=x^{-\gamma}e^\Phi,
 \qquad
 V_{XX}=-\gamma x^{-\gamma-1}e^\Phi,
 \qquad
 V_{XY}=V_X\{B(\tau)+D(\tau)y\},
\end{equation}
and the exact risky weight is
\begin{equation}\label{eq:liu_analytic_policy}
 \pi^*(t,y)
 =\frac{a_0+a_1y+\varrho\sigma\nu\{B(\tau)+D(\tau)y\}}
        {\gamma\sigma^2}.
\end{equation}
The smooth Markov identity also supplies
\begin{equation}\label{eq:liu_analytic_Z_target}
 \mathbf Z^{X,*}
 =e_X^\top
D^2_{\mathbf s\mathbf s}V(t,X_t,Y_t)\,\sigma_S^*(t,X_t,Y_t).
\end{equation}
The fixed-latent estimator harvests $P^{XX}$ and $P^{XY}$ directly.  The CRRA
wealth scaling gives $P^{XX}=-(\gamma/X)\lambda^X=V_{XX}$ in this benchmark,
whereas the $P^{XY}$--$V_{XY}$ comparison is retained only as a numerical
diagnostic.  The shifted input is estimated by
\eqref{eq:anchored_shift_regression}, and $\mathbf Z^X$ is reconstructed through
\eqref{eq:anchored_Z_reconstruction}.

Both calibrations use $r=0.03$, $\theta=0.2$, $\nu=0.3$, $\sigma=0.2$, and
$a_1=0.2$.  The canonical tuple $(\gamma,T,\kappa,\varrho,a_0)$ is
$(2,1,2,0.1,0)$; the predeclared hedging-relevant tuple is
$(3,1.5,0.5,-0.5,0.02)$.  Each row uses 128 held-out states, 8,192 antithetic
continuations, 128 rollout steps, and an independent shifted-input replication
bank.

\begin{table}[!htbp]
\centering
\scriptsize
\setlength{\tabcolsep}{2pt}
\caption{Unconstrained predictable-return analytical diagnostics.  Entries
are normalized $L^2$ errors over 128 held-out states.  The
$P^{XY}$--$V_{XY}$ column is a numerical comparison rather than a general
identification.}
\label{tab:jun_liu_mixed_adjoint_validation}
\begin{tabularx}{\textwidth}{@{}l*{6}{>{\centering\arraybackslash}X}@{}}
\toprule
Calibration
& \shortstack{$\lambda^X$\\nRMSE}
& \shortstack{$P^{XX}$\\nRMSE}
& \shortstack{$P^{XY}$--$V_{XY}$\\nRMSE}
& \shortstack{$\mathbf Z^X$\\nRMSE}
& \shortstack{Shift\\SNR}
& \shortstack{Decoded-policy\\RMSE} \\
\midrule
Canonical
& $0.1161\%$
& $0.1345\%$
& $0.3179\%$
& $0.2601\%$
& $0.92$
& $2.465\times10^{-3}$ \\
Hedging-relevant
& $0.1826\%$
& $0.1693\%$
& $0.3502\%$
& $0.9863\%$
& $0.79$
& $8.429\times10^{-3}$ \\
\bottomrule
\end{tabularx}
\end{table}

The residual shift is estimated rather than imposed.  Its mean norm and mean
standard error are respectively $6.25\times10^{-4}$ and
$7.53\times10^{-4}$ in the canonical calibration, and
$4.33\times10^{-3}$ and $5.65\times10^{-3}$ in the hedging-relevant
calibration.  Thus the residual is not resolved above its Monte Carlo
uncertainty in either exact-policy experiment, while $\mathbf Z^X$ is
recovered below $1\%$ nRMSE.  The reported SNR is the average of the statewise
norm ratios defined in Section~\ref{sec:benchmark_scope}, rather than the ratio
of the two displayed means.  This benchmark-specific observation does not
imply a general zero-shift identity or a monotone relation between policy
optimality and shift magnitude.  The hedging-relevant calibration has a
$V_{XY}$ signal about seven times larger than the canonical one; omitting its
hedging component changes the policy by RMSE $5.465\times10^{-2}$,
substantially exceeding the full-input decoded-policy RMSE
$8.429\times10^{-3}$.

\subsection{Frozen-Graph and Shift Diagnostics for Predictable-Return Recovery}
\label{app:lko_recovery_graph_audit}

Section~\ref{sec:state-dependent} reports the constrained affine-factor setup,
budgets, and primary recovery table.  Two supplementary diagnostics clarify
how those results should be read.  First, at $n=100$, three of 128 states have
small negative B-PGDPO gains, with minimum $-5.54\times10^{-5}$, even though
the mean barrier gain is positive.  This is finite central-path bias; the exact
QP gain is positive at every state.  Second, the estimated shift is small
relative to $\mathbf Z^X$ but reproducible: the mean independent-bank
replication errors are $5.81\times10^{-4}$, $8.07\times10^{-4}$, and
$6.64\times10^{-4}$ for $n=10,50,100$, respectively.

\paragraph{Alternative frozen-feedback graph.}
Using common states and Brownian banks, we compare the fixed-latent estimates
with an alternative graph that freezes the network parameters but retains the
actor's state derivatives.  The nRMSE discrepancies in
\eqref{eq:appendix_graph_nrmse} are $0.108\%$ for $\lambda^X$,
$0.168\%$ for $P^{XX}$, $40.9\%$ for $P^{XY}$, $0.127\%$ for
$\mathbf Z^X$, and $62.4\%$ for $\boldsymbol\zeta^X$.  Thus the two
graphs produce similar scalar wealth levels in this calibration but materially
different mixed and shifted blocks.  These are graph-to-graph discrepancies,
not errors relative to analytical truth.

\subsection{Decoder Initialization}
\label{app:local_diagnostics}

We report only constraint-matched decoder comparisons.  Once the estimated
adjoint input determines the local QP coefficients, the exact QP output is
independent of a supplied initialization.  Table~\ref{tab:initialization_diagnostic}
also reports a separate 2,000-iteration projected-gradient-descent (PGD)
check from four
initializations.

\begin{table}[!htbp]
\centering
\caption{Decoder initialization diagnostic under the anchored shifted-input
adjoints.}
\label{tab:initialization_diagnostic}
\begin{tabular}{lcccc}
\toprule
Model & $N$ & Exact QP uses init & Exact max diff & PGD mean diff \\
\midrule
Affine predictable-return & 64 & no & $0$ & $6.20\times10^{-16}$ \\
Constant-opportunity cap & 64 & no & $0$ & $5.97\times10^{-16}$ \\
Affine-factor cap & 32 & no & $0$ & $7.10\times10^{-16}$ \\
\bottomrule
\end{tabular}
\end{table}

The four initializations are the DPO output, equal risky weights, a
cash-only/zero-risky allocation, and a random feasible allocation.  The
maximum PGD differences from exact QP are below $1.77\times10^{-15}$ in all
three models.

\subsection{Finite-Budget and Seed Robustness of the DPO Reference}\label{sec:training_budget}\label{app:training_diagnostics}

This is a Stage~1 training-budget diagnostic for the DPO reference,
not a training comparison among the three deployment rules.  B-PGDPO and
QP-PGDPO are deterministic Stage~2 recovery maps applied after each DPO
checkpoint is fixed; neither recovery solver is trained.  Three training seeds
define three DPO training paths while the market calibration and coefficient
realization are held fixed.  For each seed, the 250-, 500-, and 1,000-epoch
entries are checkpoints along one continued optimization path, rather than
three separately restarted trainings.  All nine checkpoints are evaluated on
the same 128 held-out states and indexed Brownian continuation arrays.
\begin{table}[!htbp]
\centering
\footnotesize
\setlength{\tabcolsep}{4pt}
\caption{Full three-seed training-budget results for the $n=100$
borrowing-allowed no-short-sale benchmark.}
\label{tab:training_budget_multiseed}
\begin{tabular}{lrrrr}
\toprule
\multicolumn{5}{l}{\textbf{Panel A. DPO across training budgets}}\\
\midrule
Epochs & Policy RMSE & KKT--PG & $\mathbb H$ gap to QP & Across-seed RMSE s.d. \\
\midrule
250 & $0.013561$ & $0.043879$ & $0.022969$ & $0.000851$ \\
500 & $0.011251$ & $0.024108$ & $0.004467$ & $0.000171$ \\
1,000 & $0.011190$ & $0.023945$ & $0.004501$ & $0.000118$ \\
\midrule
\multicolumn{5}{l}{\textbf{Panel B. Deployment rules after 1,000 epochs}}\\
\midrule
Method & Policy RMSE & KKT--PG & $\mathbb H$ gain vs DPO & $\mathbb H$ gap to QP \\
\midrule
DPO & $1.1190\times10^{-2}$ & $2.3945\times10^{-2}$ & $0$ & $4.501\times10^{-3}$ \\
B-PGDPO & $9.291\times10^{-3}$ & $8.062\times10^{-4}$ & $4.484\times10^{-3}$ & $1.729\times10^{-5}$ \\
QP-PGDPO & $1.374\times10^{-4}$ & $2.31\times10^{-13}$ & $4.501\times10^{-3}$ & $0$ \\
\bottomrule
\end{tabular}
\end{table}

Most DPO improvement occurs by 500 epochs; thereafter the policy RMSE, KKT--PG,
and QP gap largely plateau.  At 1,000 epochs, B-PGDPO has coordinatewise policy
RMSE $9.29\times10^{-3}$, KKT--PG $8.06\times10^{-4}$, and QP gap
$1.73\times10^{-5}$.  QP-PGDPO has coordinatewise RMSE
$1.37\times10^{-4}$ and solves the estimated QP to numerical tolerance.  The
result is a finite-budget robustness statement, not a claim that no other
architecture or substantially larger budget could improve DPO.

The fixed-latent CRRA identity is a graph-contract check rather than an
optimization-budget diagnostic.  Recovery removes nearly all of the local QP
gap left by the finite DPO checkpoints.

\subsection{Barrier Central-Path, Parameter-Sensitivity, and Switching-Region Audit}
\label{sec:barrier_audit}
\label{app:barrier_sensitivity}

For the quadratic-affine blocks above, QP-PGDPO is the preferred exact local
solver conditional on the estimated adjoints.  Here B-PGDPO audits the
central-path approximation and solver modularity; its intended role is more
general smooth geometry, not greater efficiency than QP.  We audit it in the
$n=100$
constant-opportunity and affine-factor cap models using 64 states per split
and 8,192 continuations per state.  The reported
$\varepsilon_{\rm bar}=3\times10^{-6}$ runs use double precision, an analytic
center, and at most 120 Newton iterations.  If $\mathbf g_{\rm bar}$ is the
barrier-objective gradient and $d$ the Newton ascent direction, the squared
Newton decrement is
$\delta_{\rm N}^2=\mathbf g_{\rm bar}^{\top}d$ as in
\eqref{eq:appendix_squared_newton_decrement}.  Its half is the improvement
predicted by the local quadratic Newton model; the run stops when
$\delta_{\rm N}^2/2<10^{-8}$.  This criterion audits convergence to the
finite-$\varepsilon_{\rm bar}$ barrier solution, not the separate barrier bias
relative to the unregularized QP.  All barrier values use the same states,
adjoint inputs, and Brownian banks as their exact local-QP benchmarks.

\begin{table}[!htbp]
\centering
\scriptsize
\setlength{\tabcolsep}{2.5pt}
\caption{Final $n=100$ B-PGDPO cap-model barrier audit at
$\varepsilon_{\rm bar}=3\times10^{-6}$.  Gain and QP gap use the
unregularized full estimated generalized Hamiltonian.}
\label{tab:barrier_final_audit}
\begin{tabular}{llrrrr}
\toprule
Model & Split
& Portfolio KKT--PG
& $\mathbb H$ gain vs DPO
& $\mathbb H$ gap to QP
& Converged \\
\midrule
Constant-opportunity cap & IID
& $2.610\times10^{-3}$ & $2.5730\times10^{-2}$ & $2.22\times10^{-4}$ & $1.000$ \\
Constant-opportunity cap & Wealth edge
& $3.543\times10^{-3}$ & $5.0067\times10^{-2}$ & $2.22\times10^{-4}$ & $1.000$ \\
Constant-opportunity cap & Stress
& $4.027\times10^{-3}$ & $2.0401\times10^{-1}$ & $2.22\times10^{-4}$ & $1.000$ \\
\addlinespace
Affine-factor cap & IID
& $2.597\times10^{-3}$ & $5.3598\times10^{-2}$ & $2.58\times10^{-4}$ & $1.000$ \\
Affine-factor cap & Wealth edge
& $3.207\times10^{-3}$ & $1.1618\times10^{-1}$ & $2.58\times10^{-4}$ & $1.000$ \\
Affine-factor cap & Factor edge
& $1.504\times10^{-3}$ & $2.5182\times10^{-1}$ & $2.83\times10^{-4}$ & $1.000$ \\
Affine-factor cap & Stress
& $1.293\times10^{-3}$ & $6.9836\times10^{-1}$ & $2.83\times10^{-4}$ & $1.000$ \\
\bottomrule
\end{tabular}
\end{table}

Every evaluated cap-model state satisfies the stopping rule and has positive
unregularized joint gain.  Mean gaps to exact QP are approximately
$2.22\times10^{-4}$ in the constant-opportunity-cap splits and
$2.58$--$2.83\times10^{-4}$ in the factor-cap splits.  The barrier-parameter sensitivity reported below shows their approximately linear decline with $\varepsilon_{\rm bar}$.  This validates the central path against an exact
local-QP benchmark, not a separate nonquadratic benchmark.

\paragraph{Pointwise diagnostics near active-set transitions.}
\label{sec:switching_audit}

Using the common full-input QP benchmark, define the normalized upper slack
\[
 s_{\rm up}:=\frac{C_{\max}-C^{\rm QP}}{C_{\max}-C_{\min}}.
\]
We classify states as upper-active, near-switching
($0<s_{\rm up}\le0.05$), or far-interior.  Because every method uses the same partition, comparisons do not relabel states
endogenously.  The near-switching cells are a pointwise stress test, not a
trajectory-through-switch experiment or evidence of uniform regularity on the
switching surface.

\begin{table}[!htbp]
\centering
\scriptsize
\setlength{\tabcolsep}{3pt}
\caption{Region-stratified B-PGDPO recovery at
$\varepsilon_{\rm bar}=3\times10^{-6}$.  Counts aggregate nonempty IID, edge,
and stress cells.  Gains and QP gaps use the unregularized generalized
Hamiltonian under the common estimated adjoints; ``Negative frac.'' is the
fraction of states with negative gain.}
\label{tab:switching_region_audit}
\begin{tabular}{llrrrrr}
\toprule
Model & QP-benchmark region & $N$
& Mean $\mathbb H$ gain
& Minimum gain
& Mean gap to QP
& Negative frac. \\
\midrule
Constant-opportunity cap & Upper-active
& 95 & $6.7480\times10^{-2}$ & $2.19\times10^{-4}$ & $2.22\times10^{-4}$ & $0$ \\
Constant-opportunity cap & Near-switching
& 8 & $2.913\times10^{-3}$ & $5.21\times10^{-4}$ & $2.21\times10^{-4}$ & $0$ \\
Constant-opportunity cap & Far-interior
& 89 & $1.2892\times10^{-1}$ & $2.255\times10^{-3}$ & $2.22\times10^{-4}$ & $0$ \\
\addlinespace
Affine-factor cap & Upper-active
& 132 & $1.5740\times10^{-1}$ & $4.57\times10^{-4}$ & $2.72\times10^{-4}$ & $0$ \\
Affine-factor cap & Near-switching
& 20 & $1.2240\times10^{-1}$ & $4.494\times10^{-3}$ & $2.68\times10^{-4}$ & $0$ \\
Affine-factor cap & Far-interior
& 104 & $4.6589\times10^{-1}$ & $6.405\times10^{-3}$ & $2.69\times10^{-4}$ & $0$ \\
\bottomrule
\end{tabular}
\end{table}

All regions have positive pointwise gain and similar mean QP gaps.

\paragraph{Barrier-parameter sensitivity.}
Table~\ref{tab:barrier_eps_audit} reports the finite-barrier sensitivity on
common IID states and fixed adjoint inputs.
\begin{table}[!htbp]
\centering
\scriptsize
\setlength{\tabcolsep}{3pt}
\caption{IID barrier-parameter sensitivity at $n=100$ under the anchored
shifted-input adjoints (64 held-out states, 8,192 continuations per state).
Gains and QP gaps use the unregularized generalized Hamiltonian; ``Neg. share''
is the fraction of states with negative gain.}
\label{tab:barrier_eps_audit}
\begin{tabular}{lrrrrrr}
\toprule
Model & $\varepsilon_{\rm bar}$ & Mean gain & Min. gain & Neg. share & Gap to QP & Converged \\
\midrule
Constant-opportunity cap & $10^{-5}$ & $2.5235\times10^{-2}$ & $-1.78\times10^{-4}$ & $0.1719$ & $7.17\times10^{-4}$ & $1.000$ \\
Constant-opportunity cap & $3\times10^{-6}$ & $2.5730\times10^{-2}$ & $2.91\times10^{-4}$ & $0$ & $2.22\times10^{-4}$ & $1.000$ \\
Constant-opportunity cap & $10^{-6}$ & $2.5877\times10^{-2}$ & $4.35\times10^{-4}$ & $0$ & $7.48\times10^{-5}$ & $1.000$ \\
\addlinespace
Affine-factor cap & $10^{-5}$ & $5.3009\times10^{-2}$ & $1.33\times10^{-4}$ & $0$ & $8.47\times10^{-4}$ & $1.000$ \\
Affine-factor cap & $3\times10^{-6}$ & $5.3598\times10^{-2}$ & $6.09\times10^{-4}$ & $0$ & $2.58\times10^{-4}$ & $1.000$ \\
Affine-factor cap & $10^{-6}$ & $5.3769\times10^{-2}$ & $7.55\times10^{-4}$ & $0$ & $8.64\times10^{-5}$ & $1.000$ \\
\bottomrule
\end{tabular}
\end{table}

The same fixed states and adjoint inputs are reused throughout.  The ratios of
QP gap to $\varepsilon_{\rm bar}$ are $71.7$, $74.0$, and $74.8$ in the
constant-opportunity cap and $84.7$, $85.9$, and $86.4$ in the affine-factor
cap.  This near-linear scaling is consistent with the sum of
$O(\varepsilon_{\rm bar})$ complementarity contributions from many
simultaneously active nonnegativity constraints; it is not an identification
of the abstract constant in Proposition~\ref{prop:barrier_policy}.  At $10^{-5}$,
the constant-opportunity-cap mean gain remains positive but $17.19\%$ of
states have a small negative pointwise gain.  The common $3\times10^{-6}$
specification has positive gain at every IID state.

Using the common QP reference, the constant-opportunity-cap aggregate region counts are
$(95,8,89)$ for upper-active, near-switching, and far-interior states,
respectively; the corresponding affine-factor counts are $(132,20,104)$.
Stress samples contain active and far-interior states but no observations in
the $5\%$ near-switching band.  All reported regional joint gains are
positive.

\subsection{External Neural and Runtime Details}
\label{app:external_comparison_details}

Table~\ref{tab:hpin_external} reports the main hPINN-style diagnostic.  The
baseline trains a value network and feasible control head using HJB, terminal,
and soft projected-KKT residuals, without local QP/KKT recovery.  The matched
DPO and analytical KKT comparisons remain primary.

The supplementary affine-factor model uses diagonal OU predictors
$Y\in\mathbb R^{d_Y}$ and
\[
\boldsymbol\mu(Y)-r\mathbf1
=
\operatorname{diag}(\boldsymbol\sigma_{\rm asset})
A_{\rm load}Y.
\]
The hPINN and DPO-based entries use the same model calibration; the latter
use the anchored shifted-input estimator.  Table~\ref{tab:timing_profile}
reports the corresponding runtime decomposition.

\begin{table}[!htbp]
\centering
\small
\setlength{\tabcolsep}{3.5pt}
\caption{Representative timing profile for the supplementary $n=100$
affine-factor benchmark.  The independent holdout column is an audit cost,
not part of a single deployment pass.  Times are wall-clock minutes.}
\label{tab:timing_profile}
\begin{tabular}{lccccc}
\toprule
$d_Y$ & DPO warm-up & \OLBPTT+shift+QP & Primary total & Holdout audit & QP KKT \\
\midrule
1 & $1.883$ & $2.334$ & $4.217$ & $2.286$ & $2.65\times10^{-16}$ \\
3 & $2.225$ & $2.791$ & $5.016$ & $2.758$ & $1.53\times10^{-16}$ \\
5 & $2.200$ & $3.482$ & $5.682$ & $3.361$ & $1.70\times10^{-16}$ \\
\bottomrule
\end{tabular}
\end{table}

The timing numbers are not a method ranking: architectures, objectives, and
tuning budgets differ.  The QP implementation is neither batched nor
optimized.

\bibliographystyle{apalike}
\bibliography{gf_bib}
\end{document}